\documentclass[aps,prx,superscriptaddress, reprint]{revtex4-2}
\usepackage{times}
\usepackage[T1]{fontenc}
\usepackage{graphicx}
\usepackage{bbm}
\usepackage[usenames,dvipsnames,table]{xcolor}
\usepackage{tikz}
\usetikzlibrary{shapes.geometric, arrows.meta, positioning, calc, backgrounds, fit}
\usepackage{amsthm,amsmath}
\usepackage{mathtools}
\usepackage{float}
\usepackage{braket}
\usepackage{amssymb }
\usepackage{comment}
\usepackage{bbold}
\usepackage{dsfont}
\usepackage{ulem} 
\usepackage{empheq}
\usepackage{cancel}
\usepackage{soul}
\usepackage{booktabs}
\usepackage{array}
\usepackage{multirow}
\usepackage{enumitem}
\usepackage[section]{placeins} 
\newcommand{\cmark}{\ensuremath{\checkmark}}
\newcommand{\xmark}{\ensuremath{\times}}

\definecolor{blush}{HTML}{E7F0FB}
\definecolor{rose}{HTML}{6E97C9}
\definecolor{rosedeep}{HTML}{30537E}
\definecolor{tabink}{HTML}{2A3340}
\definecolor{tabsoft}{HTML}{70808A}
\newcommand{\hd}[1]{\textcolor{tabink}{\textbf{#1}}}
\newcommand{\pf}[1]{\textcolor{tabsoft}{\footnotesize #1}}

\usepackage[colorlinks=true,citecolor=blue,linkcolor=RubineRed,urlcolor=blue]{hyperref}
\usepackage{xurl}

\usepackage[most]{tcolorbox}
\tcbuselibrary{skins,breakable}

\newtheorem{theorem}{Theorem}
\newtheorem{lemma}[theorem]{Lemma}
\usepackage{etoolbox}
\usepackage{mathrsfs}

\tcolorboxenvironment{lemma}{
  enhanced,
  colback=blue!5,
  colframe=white,
  boxrule=0pt,
  sharp corners,
  left=6pt,right=6pt,top=6pt,bottom=6pt
}

\tcolorboxenvironment{theorem}{
  enhanced,
  colback=green!5,
  colframe=white,
  boxrule=0pt,
  sharp corners,
  left=6pt,right=6pt,top=6pt,bottom=6pt
}

\newtcolorbox{examplebox}[1][]{
    breakable,
  colback=green!6!white,
  colframe=green!40!black,
  coltext=black,
  fonttitle=\bfseries,
  title=#1,
  sharp corners,
  boxrule=0.8pt,
  left=6pt,
  right=6pt,
  top=6pt,
  bottom=6pt
}

\DeclareMathOperator{\csch}{csch}
\newcommand{\Li}{\operatorname{Li}}

\usepackage{minitoc}
\usepackage{tocloft}

\begin{document}

\doparttoc 
\faketableofcontents 
\part{} 
\vspace{-2.5cm} 

\title{Field-Space Entanglement Dynamics Between Tunnel-Coupled Luttinger Liquids}

\author{Léonce Dupays \href{https://orcid.org/0000-0002-3450-1861}{\includegraphics[scale=0.05]{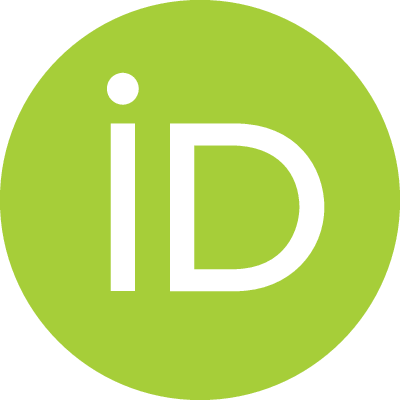}}}
\thanks{These authors contributed equally.}

\affiliation{Department of Mathematics, King’s College London, Strand, London WC2R 2LS, UK}

\author{Taufiq Murtadho\href{https://orcid.org/0000-0003-0324-8119}
{\includegraphics[scale=0.05]{orcidid.eps}}}
\thanks{These authors contributed equally.}
\affiliation{\mbox{Nanyang Quantum Hub, School of Physical and Mathematical Sciences, Nanyang Technological University, 637371, Singapore}}
\affiliation{\mbox{Centre for Quantum Technologies, National University of Singapore, 3 Science Drive 2, Singapore 117543, Singapore}}

\author{Bi Hong Tiang\href{https://orcid.org/0000-0002-5460-7081}{\includegraphics[scale=0.05]{orcidid.eps}}}
\thanks{These authors contributed equally.}
\affiliation{\mbox{Centre for Quantum Technologies, National University of Singapore, 3 Science Drive 2, Singapore 117543, Singapore}}

\author{Nelly H.Y. Ng \href{https://orcid.org/0000-0003-0007-4707}{\includegraphics[scale=0.05]{orcidid.eps}}}
\email{nelly.ng@ntu.edu.sg}
\affiliation{\mbox{Nanyang Quantum Hub, School of Physical and Mathematical Sciences, Nanyang Technological University, 637371, Singapore}}
\affiliation{\mbox{Centre for Quantum Technologies, National University of Singapore, 3 Science Drive 2, Singapore 117543, Singapore}}

\author{Paola Ruggiero\href{https://orcid.org/0000-0002-0891-0477}{\includegraphics[scale=0.05]{orcidid.eps}}}
\email{paola.ruggiero@kcl.ac.uk}
\affiliation{Department of Mathematics, King’s College London, Strand, London WC2R 2LS, UK}

\date{\today}

\begin{abstract}
Entanglement dynamics depend not only on how a quantum system is partitioned, but critically on how interactions across that partition are structured. For a spatial bipartition of a locally interacting system, entanglement is generated near the boundary and then propagates into the bulk. By contrast, when two extended quantum fields are coupled locally along their entire length, the interaction crosses the field-space partition everywhere, and this generates correlations throughout the system. Here, we study the entanglement dynamics between two gapless one-dimensional quantum many-body systems described by Luttinger liquid theory. The systems are initially decoupled and prepared at zero or finite temperature, after which a time-dependent tunneling interaction is activated uniformly along their length. Within a Gaussian approximation, we derive general analytical expressions for the logarithmic negativity, mutual information, and R{\'e}nyi entropies under arbitrary coupling protocols. At zero temperature, entanglement displays an early-time power-law growth whose exponent is fixed solely by the first non-null derivative of the tunneling protocol. Once the coupling saturates, we obtain exact long-time averages of the information-theoretic quantities and characterise how the correlations scale with temperature and the final coupling strength. We also analyse how mutual information and logarithmic negativity approach the adiabatic limit for a very slow protocol with respect to intrinsic system timescale. This work extends the study of entanglement dynamics in nonequilibrium field theory to field-space partitions and mixed initial states. 
\end{abstract}

\maketitle
\section{Introduction}
Entanglement provides a sharp lens through which we may probe non-equilibrium dynamics of quantum many-body systems. In quantum field theories (QFTs), a central question is how quantum correlations can emerge and spread in a relaxation process. Much progress on this has come from theoretical studies of entanglement after quantum quenches \cite{Wen_book}, most commonly with respect to a spatial bipartition of the field \cite{Calabrese2004}. This perspective reveals a rich dynamical structure: the interplay between entanglement growth and thermalization \cite{Calabrese2020}, the crossover from area-law to volume-law scaling \cite{Abanin2019}, and universal scaling laws governing the propagation of entanglement \cite{Calabrese2005}. These theoretical developments are now being brought increasingly close to experiments, through direct probes of out-of-equilibrium entanglement and new proposals for measuring entanglement dynamics in quantum simulators \cite{Islam2015,Lukin2019, yang2020simulating, viermann2022quantum, agullo2024toward, Mathe2026, Matsoukasroubeas2026,Yang2026}.

However, the choice of a spatial cut is only one way of asking where quantum correlations in a field theory reside. When correlations are organized by energy scales or other field degrees of freedom, alternative partitions offer complementary perspectives on entanglement. Momentum space partitioning, for example, reveals entanglement between length scales, and connects naturally into Wilsonian effective action and renormalization groups \cite{balasubramanian2012momentum, martins2022momentum}. In systems with multiple interacting fields, it is natural to ask how much entanglement is generated between the fields themselves. Such \textit{field-space entanglement} arises between different degrees of freedom such as particle species, spins, or polarizations~\cite{mollabashi2014entanglement, mozaffar2016on, taylor2016generalized, hufel2017field}, and has been explored in a variety of settings where the relevant subsystems are not spatial regions, but physically distinct components of a composite quantum system. In condensed-matter physics, it provides a natural measure of correlations between two legs of a spin ladder \cite{Poilblanc2010, Lauchli2012}, between electrons and phonons \cite{roosz2021entanglement, roosz2022densitymatrix}, or condensate mixtures \cite{yoshino2021intercomponent}.
Most theoretical studies focused on ground-state properties, analyzing Rényi entropies and the entanglement spectrum to probe different inter-field couplings \cite{Xu2011, Furukawa2011, Lundgren2013, Chen2013}. These works show that the structure of field-space entanglement depends sensitively on whether the coupled modes are gapped or gapless, and in some cases leads to unconventional, nonlocal entanglement Hamiltonians \cite{Lundgren2013}. Related ideas also appear in high-energy physics, including studies of entanglement in particle colliders \cite{zhang2026quantum}, interacting quantum fields in the early universe \cite{nakai2017entanglement, choudhury2022four, colas2022four}, and geometric structures in AdS/CFT correspondence \cite{mollabashi2014entanglement,mozaffar2016on, hufel2017field, taylor2016generalized}.

What remains unexplored is the \textit{dynamics} of field-space entanglement. If two initially decoupled fields are brought into contact, how is entanglement generated? Does it spread and saturate in ways analogous to spatial entanglement? Do universal scaling features observed for spatial entanglement carry over, or is field-space entanglement governed by different dynamical structures? These questions are no longer purely theoretical: modern quantum simulators can now realize and probe out-of-equilibrium field dynamics with remarkable precision. For example, experiments with one-dimensional Bose gases have provided detailed insights into thermalization \cite{gring2012relaxation, langen2013local, langen2015experimental} and correlation dynamics in such quantum field simulators \cite{aimet2025experimentally, schweigler2021decay}. However, identifying genuinely quantum contributions to these correlations remains challenging. In realistic platforms~\cite{schweigler2021decay,jarema2025information,viermann2022quantum, tajik2023experimental, schweigler2017experimental, aimet2025experimentally, agullo2024toward}, states are prepared at finite temperatures, where the standard entanglement entropy ceases to be a faithful measure of entanglement due to thermal noise.

We comprehensively study field-space entanglement dynamics 
for tunnel-coupled Luttinger liquids (LLs) with a time-dependent tunneling strength initialized at zero or finite temperature. This setup can be realized in 1D ultracold Bose gases trapped on a double-well potential \cite{Gritsev2007,Hofferberth2007, Betz2011}, and recent work \cite{murtadho2026extensive} shows that detecting field-space entanglement in such systems is within experimental reach. The dynamical question, however, remains open: how does entanglement emerge and evolve, when two finite-temperature quantum fields undergo a time-dependent interaction? 

Our work fills this gap by providing general analytical solutions to field-space entanglement dynamics in the Gaussian regime. We analyse three correlation measures: logarithmic negativity (LN), mutual information (MI), and R{\'e}nyi entropies. Their distinctions are essential yet complementary: at finite temperature the state is mixed, and LN is an appropriate entanglement measure. By contrast, MI and R{\'e}nyi entropies quantify entanglement only for pure initial states, while MI continues to provide a meaningful measure of total (classical and quantum) correlations in the finite temperature regime.

This paper is organized as follows. In Sec.~\ref{sec: PhysicalSetting}, we introduce the physical setting of two Luttinger liquids, dynamically coupled by turning on a time-dependent tunneling in the Hamiltonian. Within the harmonic approximation, the dynamics is casted in terms of a mode-dependent Ermakov equation. The solution set, i.e. \textit{Ermakov factors}, is the central object of analysis, which we use to derive analytical expressions for time-dependent LN, MI, and R{\'e}nyi entropies in Sec. \ref{sec: Entanglement_dynamics_study}. 
In Sec.~\ref{early_time_MI}, we establish that the early-time behaviour of mutual information is governed by first non-null derivative of the ramp protocol at initial times. Sec.~\ref{sec: long time average} demonstrates that the asymptotic long-time behaviour of all three measures is fully captured by just three parameters: the final saturation value of the protocol, the Ermakov factors and their time derivatives at the onset of saturation. 
Sec.~\ref{sec: adiabatic smooth_quench_protocol} examines the crossover from nonadiabatic to adiabatic entanglement dynamics, where a simple analytic solution is available. Lastly, we provide a discussion and outline future directions in Sec.\ref{conclusion}. The symbols and notations used in the text are summarized in Appendix~\ref{summary_notations}.

\section{Exact Gaussian dynamics in tunnel-coupled Luttinger Liquids \label{sec: PhysicalSetting}}
\begin{figure}[t]
    \centering    \includegraphics[width=\linewidth]{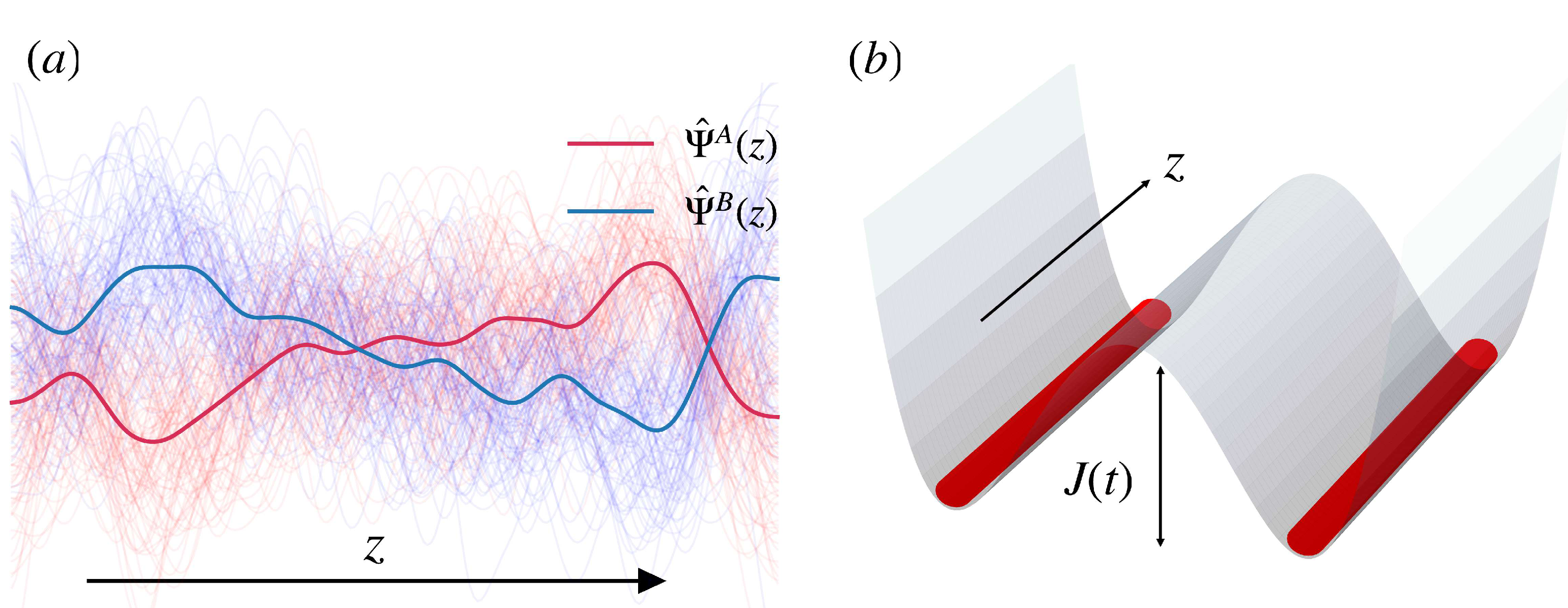}
\caption{\textbf{Schematic of the concept of field-space entanglement in 1D and its experimental implementation.} (a) Two overlapping 1D quantum fields $\hat{\Psi}^A(z)$ and $\hat{\Psi}^B(z)$ interacting through a global time-dependent interaction strength $J(t)$. (b) Field-space entanglement can be generated by trapping bosonic quantum fluids in a 1D elongated potential with a double-well transverse cross section. }
\label{fig:schematic_diagram}
\end{figure}
\subsection{Tunnel-coupled Luttinger liquids}

We define two 1D bosonic fields $\hat{\Psi}^a(z)$ labeled by the index $a\in \{A,B\}$. For concreteness, we let $\hat{\Psi}^a(z)$ represent the field operators of a pair of 1D bosonic quantum fluids. Using the standard hydrodynamic decomposition \cite{haldane1981effective}
\begin{align}
\hat{\Psi}^{a}(z)&=e^{i\hat{\phi}^{a}(z)}\sqrt{n_{0} +\delta\hat{n}^a(z)},
\label{eq:hydro_deco}
\end{align}
the field is parametrized by a phase fluctuation $\hat{\phi}^{a}(z)$ and its conjugate density fluctuation $\delta \hat{n}^a(z)$, where $n_{0}$ is the mean density, assumed uniform and identical for both fields. These fluctuations satisfy the canonical commutation relation $
[\delta \hat{n}^{a}(z),\hat{\phi}^{b}(z')]=i\delta(z-z')\delta_{ab}$.
At low-energies, they are well described by the Luttinger liquid (LL) theory \cite{Cazalilla2004, Haldane1981}
\begin{align}
\hat{H}^{a}_{\rm LL}&=\frac{\hbar v}{2}\int dz\left[\frac{K}{\pi}(\partial_{z}\hat{\phi}^{a}(z))^{2}+\frac{\pi}{K}(\delta \hat{n}^{a}(z))^{2}\right], \label{independent_TLL}
\end{align}
with the speed of sound $v$ and the Luttinger parameter $K$ assumed to be identical for the two LLs. 

The two LLs are initially decoupled, and a coupling is gradually introduced through     the time-dependent Hamiltonian $\hat{H}(t) = \hat{H}^{A}_{\rm LL}+\hat{H}^{B}_{\rm LL}+\hat{H}^{AB}(t)$. Our goal is to characterise the growth of entanglement between the two fields $\hat{\Psi}^A(z,t)$ and $\hat{\Psi}^B(z,t)$. This field-space entanglement is physically distinct from entanglement between regions of space, or momentum modes, of a single-field system. Unlike spatial bipartitions, the two fields are directly coupled everywhere in space [Fig.~~\ref{fig:schematic_diagram} (a)], and we investigate how this coupling generates field-space entanglement out of equilibrium. This question is timely, since experimental access to individual phase fields $\hat{\phi}^{A}(z)$ and $\hat{\phi}^{B} (z)$ has recently become accessible, through a combination of interference measurements and density ripple analysis \cite{murtadho2025measurement}, which further motivates theoretical studies of correlations in the $A/B$ partition \cite{murtadho2026extensive}.

We focus on the typical interaction $\hat{H}^{AB} \propto (\hat{\Psi}^{A})^\dagger \hat{\Psi}^B +\text{h.c.}$ in ultracold atoms experiments, which at low energies induces a time-dependent sine-Gordon coupling \cite{Gritsev2007, schweigler2017experimental}
\begin{align}
    \hat{H}^{AB}(t) &\approx -2\hbar n_0 J(t) \int \cos(\hat{\phi}^A(z) - \hat{\phi}^B(z))\; dz,
    \label{eq: coupling_Hamiltonian}
\end{align}
with the time-dependent tunneling strength $J(t)$, realised in parallel 1D Bose gases experiments by lowering the double-well barrier, see Fig.~\ref{fig:schematic_diagram} (b). We do not consider the reverse protocol, where the tunneling is switched from on to off \cite{foini2015non, aimet2025experimentally, ruggiero2021quenches, ruggiero2021large}, since there the post-quench Hamiltonian no longer couples the physical subsystems $A$ and $B$.
Therefore, the time-evolution has the form $U_{A}(t)\otimes U_{B}(t)$, which \textit{cannot} change the amount of entanglement contained in the initial state. 

We work within the harmonic approximation, where the cosine term is approximated with a quadratic function. 
This is justified at equilibrium when the tunneling energy is dominant compared to other energy scales \cite{tajik2023verification,schweigler2017experimental}, so that fluctuation can be expanded around a minimum of the cosine. In out-of-equilibrium cases,  it may be well suited to describe a fast quench to the strongly coupled Gaussian regime. We discuss this Gaussian approximation more thoroughly in Sec. \ref{conclusion}. 
To solve the dynamics generated by Eq.~\eqref{eq: coupling_Hamiltonian}, we transform into symmetric ($+$) and antisymmetric ($-$) fields
\begin{align}
    \hat{\phi}^{\pm}(z)&=\frac{1}{\sqrt{2}}[\hat{\phi}^{A}(z)\pm \hat{\phi}^{B}(z)],\\
     \delta\hat{n}^{\pm}(z) &= \frac{1}{\sqrt{2}}[\delta\hat{n}^{A}(z)\pm \delta\hat{n}^{B}(z)],
\end{align}
so that one can decouple the Hamiltonian into 
\begin{equation}
    \hat{H}(t) \approx \hat{H}^+_{\rm LL}+\hat{H}_{\rm LL}^-+2\hbar n_{0} |J(t)|\int \left(\hat{\phi}^-(z)\right)^2 \; dz, \label{eq_coupled_Luttinger_Hamiltonian}
\end{equation}
up to a constant. This decoupling arises from the exchange symmetry $A\leftrightarrow B$ between the two LLs, which here originates from the LLs sharing the same speed of sound and Luttinger parameter. The absolute value reflects the fact that the Gaussian approximation is always performed around the phase configuration that minimizes the coupling Hamiltonian. The minimizing relative phase depends on the sign of $J(t)$, but in both positive and negative cases the quadratic expansion yields a positive mass term proportional to $|J(t)|$. Without loss of generality, we assume $J(t)\geq 0$ throughout the manuscript. Thus, from Eq.~\eqref{eq_coupled_Luttinger_Hamiltonian} we see that the symmetric sector evolves as a stationary gapless LL, while the antisymmetric sector involves a time-dependent gapped Klein-Gordon (KG) Hamiltonian. We denote this configuration as the gapped-gapless regime, where one sector becomes gapped while the other remains gapless. 

Lastly, we consider thermal initial states, diagonal in the initial energy basis. Since the initial Hamiltonian is decoupled ($J(0)=0$), the initial state is then a tensor product both in $A/B$ and $\pm$ decompositions, by virtue of the exchange symmetry, i.e. $\rho(0) = \rho^A \otimes \rho^B = \rho^+ \otimes \rho^-$. The tunneling dynamics preserve the tensor product form in the $\pm$ basis, i.e. $\rho(t) = \rho^+(t)\otimes \rho ^{-}(t)$, but entanglement can still emerge dynamically between $A$ and $B$. 

\subsection{Exact Gaussian dynamics}

The Hamiltonian in Eq.~\eqref{eq_coupled_Luttinger_Hamiltonian} naturally maps into a collection of harmonic oscillators in momentum space, with a time-dependent frequency for the antisymmetric sector. We can see this by introducing the cosine transform 
\begin{align}
\hat{\phi}^{\sigma}(z)&=\frac{1}{\sqrt{n_0L}}\hat{\phi}_0^\sigma+\sqrt{\frac{2}{n_0L}}\sum_{q>0}\hat{\phi}^{\sigma}_q\cos (qz),\label{Fourier_1}\\
\delta \hat{n}^{\sigma}(z)&=\sqrt{\frac{n_0}{L}}\delta\hat{n}_0^{\sigma}+\sqrt{\frac{2n_0}{L}}\sum_{q>0}\delta \hat{n}^{\sigma}_{q}\cos(qz), \label{Fourier_2}
\end{align}
assuming the fields' domain to be $[0,L]$ with open boundary conditions \cite{Cazalilla2004}, giving discrete allowed momenta $q = N\pi/L$ with $N \in \{1,2,\cdots\}$ positive integers. Here, $\delta\hat{n}_q^{\kappa}, \hat{\phi}_q^{\sigma}$ are dimensionless field quadratures satisfying canonical commutation relation $[\delta\hat{n}^{\kappa}_{q},\hat{\phi}^{\sigma}_{q'}]=i\delta_{\kappa \sigma}\delta_{qq^\prime}$, and the superscripts $\sigma, \kappa\in \{+,-\}$ label the field sector. 

We henceforth drop the zero-mode $q = 0$ as it is fully decoupled from the $q>0$ sector, and that we do not expect it to significantly influence the amount of entanglement, given that the zero mode is only one out of a continuum of modes (in thermodynamic limit). By defining an intrinsic Luttinger frequency scale $\omega =( \pi n_0 v)/K$, the joint Hamiltonian for modes $q>0$ is expressed as $\hat{H}_{q>0} = \sum_{\sigma = \pm}\hat{H}_{q>0}^\sigma$, where
\begin{equation}
    \hat{H}_{q>0}^{\sigma} =\frac{1}{2}\hbar \omega\sum_{q>0}\left[(\delta\hat{n}^{\sigma}_{q})^2+\left(\frac{\Omega^\sigma_{q}(t)}{\omega}\right)^2(\hat{\phi}^\sigma_{q})^2\right],
    \label{eq:Hamiltonian_time_dep_mt}
\end{equation}
and the mode frequencies are
\begin{equation}
    \Omega_{q}^{\sigma}(t) = \begin{cases}
        \Omega_{q0}\equiv vq & \text{if}\; \sigma = +\\
        \Omega_q(t) \equiv\sqrt{\Omega_{q0}^2+4\omega J(t)} & \text{if}\; \sigma = -
    \end{cases}.\label{frequency_LL}
\end{equation}

Because different momentum modes are decoupled, we can track the time-evolution of the field quadratures by
\begin{equation}\label{eq:quadrature_propagation}
\begin{pmatrix}
    \delta\hat{n}_q^-(t)\\
    \hat{\phi}_q^-(t)
\end{pmatrix}   = \mathbf{G}_q^-(t)\begin{pmatrix}
    \delta\hat{n}_q^-(0)\\
    \hat{\phi}_q^-(0)
\end{pmatrix},
\end{equation}
where $\mathbf{G}_q^{-}(t) \in \text{Sp}(2, \mathbb{R})$ is the symplectic propagator. In Appendix ~\ref{Appendix_dynamics_of_mode_quadratures}, we solve the Heisenberg equation of motion and show that $\mathbf{G}_q^-(t)$ can be factorized into a time-dependent shearing and squeezing part, denoted by $\mathbf{U}_q(t)$, and symplectic rotation part, denoted by $\mathbf{R}_q(t)$, i.e.
\begin{equation}
    \mathbf{G}_q^-(t) = \underbrace{\begin{pmatrix}
        \gamma_q^{-1}(t) & \frac{-\dot{\gamma}_q(t)}{\omega}\\
        0 & \gamma_q(t)\end{pmatrix}\!}_{\mathbf{U}_q(t)}\underbrace{\begin{pmatrix}
        \cos\theta_q(t) &\frac{\Omega_{q0}}{\omega}\sin\theta_q(t)\\
        \frac{-\omega}{\Omega_{q0}}\sin\theta_q(t) & \cos\theta_q(t)
    \end{pmatrix}}_{\mathbf{R}_q(t)}.
    \label{eq:quadrature_dynamics}
\end{equation}
In particular, the $\mathbf{U}_q(t)$ is determined by a time-dependent scaling factor $\gamma_q(t) > 0$, which is a solution to the mode-dependent Ermakov equation \cite{Pinney1950}
\begin{align}
\ddot{\gamma}_{q}+\Omega^{2}_q(t)\gamma_{q}&=\frac{\Omega_{q0}^2}{\gamma^{3}_{q}}.
\label{Ermakov}
\end{align}
Throughout, we refer to the solution set $\{\gamma_q(t)\}$ as the \textit{Ermakov factors}. The requirement that $\mathbf{G}_q^-(0)$ is the identity fixes the initial conditions $\gamma_q(0) = 1$ and $\dot{\gamma}_q(0) = 0$. The analytical form of $\theta_q(t)$ can be derived as a function of $\gamma_q(t)$, 
\begin{equation}
    \theta_q(t) = \Omega_{q0}\int_{0}^{t}\frac{ds}{\gamma_q^2(s)}.
    \label{eq:theta_sol_main_text}
\end{equation}
In the symmetric sector, by substituting $\Omega_q^2(t) = \Omega_{q0}^2$ into the Ermakov equation, we get a stationary Ermakov factor $\gamma_q^+(t) = 1$. Plugging this solution into Eq.~\eqref{eq:quadrature_dynamics} and \eqref{eq:theta_sol_main_text}, one finds that $\mathbf{G}_q^+(t)$ reduces to a pure symplectic rotation
\begin{equation}
    \mathbf{G}_q^+(t) = \begin{pmatrix}
        \cos\Omega_{q0}t &\frac{\Omega_{q0}}{\omega}\sin\Omega_{q0}t\\
        \frac{-\omega}{\Omega_{q0}}\sin\Omega_{q0}t & \cos\Omega_{q0}t
    \end{pmatrix}.
    \label{eq:symplectic_plus}
\end{equation}

In short, solving the quadrature dynamics for a given tunneling protocol $J(t)$ amounts to solving the mode-dependent Ermakov equations~\eqref{Ermakov}. This can be done numerically or analytically via the Pinney form \cite{Pinney1950}, with the latter constructed from linearly independent solutions to the associated homogeneous time-dependent harmonic oscillator equations. 

\subsection{Covariance matrix formalism}
For Gaussian systems, a quantum state is fully characterised by its covariance matrix $\mathbf{\Gamma}(t)$ \cite{serafini2023quantum}. Since we start with states that are diagonal in the energy basis, and the tunneling preserves independence between different momentum modes, $\mathbf{\Gamma}(t)$ can be expressed as direct sum of each mode, 
\begin{equation} \label{Gamma_t}
    \mathbf{\Gamma}(t) = \bigoplus_{q>0}\mathbf{\Gamma}_q(t).
\end{equation}
The elements of the matrix $\mathbf{\Gamma}_q(t)$ are given by the expectation value of the quadrature operator $\mathbf{\hat{r}}_q^\sigma(t) = \begin{pmatrix}
          \delta\hat{n}^\sigma_{q}(t) & \hat{\phi}^\sigma_{q}(t)
\end{pmatrix}^T$.
The first moments $\bar{r}_q^\sigma = \langle \mathbf{\hat{r}}_q^\sigma \rangle$ vanish for the initial states we consider and remain zero throughout the evolution. All the relevant information for correlations, entanglement, and their dynamics, is therefore encoded in the second moments:
\begin{equation}
\mathbf{\Gamma}_q^{\kappa \sigma} = \frac{1}{2}\left\langle\left\{\mathbf{\hat{r}}_q^\kappa, \left(\mathbf{\hat{r}}_q^\sigma\right)^T\right\}\right\rangle,
\label{element_covariance}
\end{equation}
where $\{\mathbf{a}, \mathbf{b}\} \equiv \mathbf{a}\mathbf{b}+(\mathbf{a}\mathbf{b})^T$ is the symmetrized product. 
The indices $\kappa, \sigma \in \{+,-\}$ label the field sector, so $\mathbf{\Gamma}_q^{++}$ and $\mathbf{\Gamma}_q^{--}$ are the covariance matrices of the symmetric and antisymmetric sectors respectively, while $\mathbf{\Gamma}_q^{+-}$ encodes the cross-correlations between them.

For initial thermal states of the Hamiltonian \eqref{eq_coupled_Luttinger_Hamiltonian} with inverse temperature $\beta =  (k_BT)^{-1}$, the covariance matrix takes a block diagonal form $\mathbf{\Gamma}_q(0) = \mathbf{\Gamma}_q^{++}(0) \oplus \mathbf{\Gamma}_q^{--}(0)$, with 
\begin{equation}\label{eq:Gamma_pm_teq0}
    \mathbf{\Gamma}_q^{++}(0) = \mathbf{\Gamma}_q^{--}(0) = C_{q}(\beta)\begin{pmatrix}
        \Omega_{q0}/\omega& 0\\
        0&\omega/\Omega_{q0}
\end{pmatrix},
\end{equation}
where $\beta=(k_B T)^{-1}$ denotes the inverse temperature, and
\begin{equation} \label{eq:C_q}
    C_{q}(\beta)\equiv\frac{1}{2}\coth\left(\frac{\beta \hbar \Omega_{q0}}{2}\right) \ ,
\end{equation}
while the off-diagonal sectors vanish. Under the quench, each diagonal sector evolves independently via $\mathbf{\Gamma}_q^{\sigma\sigma}(t) = \mathbf{G}_q^{\sigma}(t)\mathbf{\Gamma}_q^{\sigma\sigma}(0)\mathbf{G}^{\sigma}_q(t)^{ T}$, see Appendix \ref{dynamics_covariance} for details. Since the $+$ and $-$ sectors remain decoupled throughout, the off-diagonal blocks vanish at all times: $\mathbf{\Gamma}^{+-}_q = \mathbf{\Gamma}^{-+}_q = 0$. To evaluate entropy measures on local subsystems ($A$ or $B$), we define $\textbf{M} = \frac{1}{\sqrt{2}}\bigl( \begin{smallmatrix} 1 & 1 \\ 1 & -1 \end{smallmatrix} \bigr)$ and rotate the covariance matrix into $A/B$ basis
\begin{equation}\mathbf{\Gamma}_q = 
    \begin{pmatrix}
        \mathbf{\Gamma}_q^{AA} & \mathbf{\Gamma}_q^{AB}\\
        \mathbf{\Gamma}_q^{BA} & \mathbf{\Gamma}_q^{BB}
    \end{pmatrix} = \mathbf{M}\begin{pmatrix}
        \mathbf{\Gamma}_q^{++} & 0\\
        0 & \mathbf{\Gamma}_q^{--}
    \end{pmatrix}\mathbf{M}^T.
    \label{eq:sector_transformation}
\end{equation}

A key advantage of working with Gaussian states is that entropic and entanglement measures reduce to straightforward functions of \textit{symplectic eigenvalues} of the covariance matrix (see Appendix \ref{appdxB:symp_eigval_mi_dynamics} for definitions). In particular, for ${\bf \Gamma}(t)$ decomposed into independent momentum sectors as in Eq.~\eqref{Gamma_t}, the sum runs over symplectic eigenvalues of each momentum sector. We see in the next section how it allows us to compute the dynamics of several information-theoretic quantities. 
\section{Field-space entanglement dynamics \label{sec: Entanglement_dynamics_study}}

Here, we present the dynamics of correlation measures between $A$ and $B$, in particular mutual information and logarithmic negativity. The mutual information quantifies total (classical and quantum) correlation between the two subsystems, while logarithmic negativity provides a direct entanglement measure in mixed states. 
As a byproduct, we obtain results for the R{\'e}nyi and entanglement entropies.
Our analysis here holds for arbitrary initial temperature, and Gaussian-preserving tunneling protocols. Fig.~\ref{fig:protocol_schematic} shows a schematic summarizing the general formalism and main results.
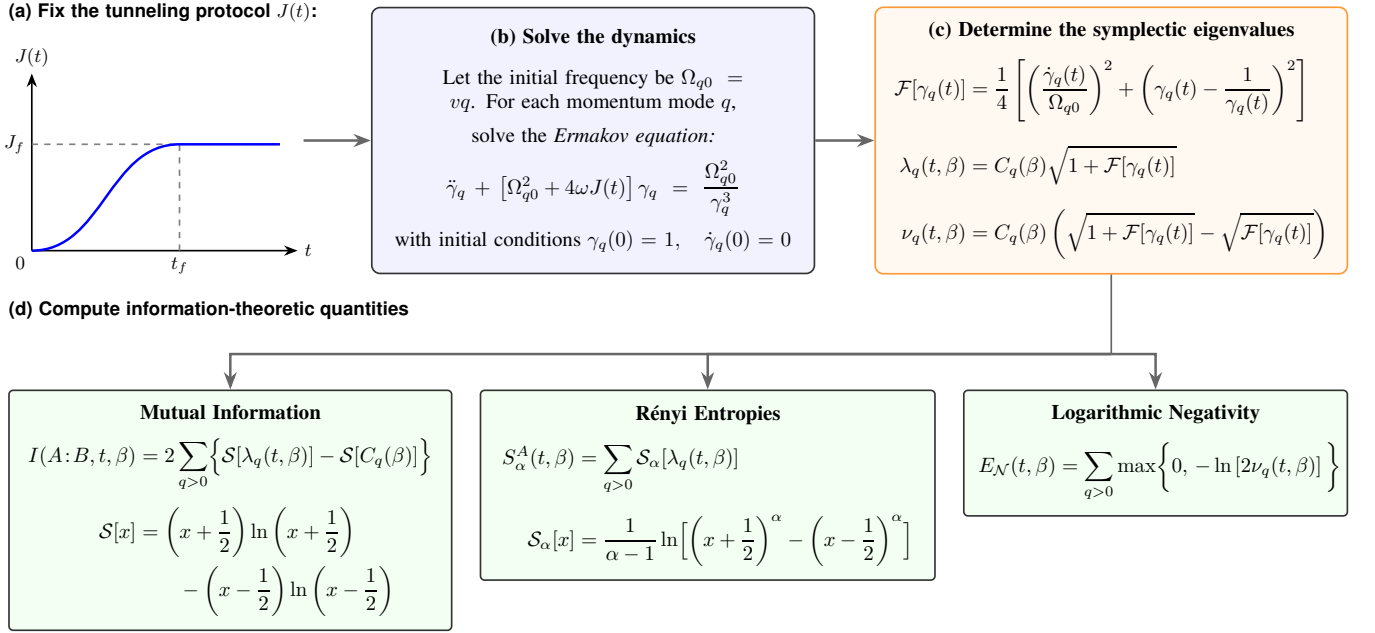
\begin{figure*}[t]
\centering
\resizebox{\textwidth}{!}{%
\begin{tikzpicture}[
    >=Stealth,
    mainbox/.style={
        draw=black!80, thick, rounded corners=4pt,
        align=center, inner sep=7pt, fill=blue!5,
        font=\normalsize
    },
    obsbox/.style={
        draw=black!80, thick, rounded corners=2pt,
        align=center, inner sep=7pt, fill=green!5,
        text width=6cm,
        font=\normalsize
    },
    bridgebox/.style={
        draw=orange!80, thick, rounded corners=4pt,
        align=center, inner sep=7pt, fill=orange!5,
        font=\normalsize
    },
    labeltext/.style={font=\bfseries\sffamily, anchor=north west}
]


\node[labeltext, text width=7.1cm, align=left] (labelA) at (0, 4.8) {
    (a) Fix the tunneling protocol $J(t)$:
};

\coordinate (origin) at (0.5, 0.5);

\draw[->, thick] (origin) -- ++(4.5, 0) node[right] {$t$};
\draw[->, thick] (origin) -- ++(0, 3.0) node[above] {$J(t)$};
\node[below left] at (origin) {$0$};

\draw[dashed, thick, gray] ($(origin)+(0, 1.8)$) node[left, text=black] {$J_f$} -- ++(4.2, 0);
\draw[dashed, thick, gray] ($(origin)+(2.5, 0)$) node[below, text=black] {$t_f$} -- ++(0, 1.8);

\draw[very thick, blue] (origin) to[out=0, in=180] ($(origin)+(2.5, 1.8)$) -- ++(1.7, 0);

\node[mainbox, text width=7.0cm, minimum height=4.5cm,anchor=south] (ermakov) at (10.0, 0.1){
    \textbf{(b) Solve the dynamics}\\[8pt]
    Let the initial frequency be $\Omega_{q0}=vq$. For each momentum mode $q$,\\[5pt]
    solve the \textit{Ermakov equation:}\\[6pt]

    $\ddot{\gamma}_{q}+ \left[\Omega_{q0}^2+4\omega J(t)\right]\gamma_{q}
    =\dfrac{\Omega_{q0}^2}{\gamma^{3}_{q}}$
    \\[6pt]
    with initial conditions $\gamma_q(0)=1,\quad\dot{\gamma}_q(0)=0$
};

\node[bridgebox, text width=7.5cm, minimum height=4cm, anchor=west] (bridge) at ($(ermakov.east)+(1,0)$) {
    \textbf{(c) Determine the symplectic eigenvalues}\\[8pt]
    {\normalsize
    $\begin{aligned}
        \mathcal{F}[\gamma_q(t)]
        &= \dfrac{1}{4}
        \left[
        \left(\dfrac{\dot{\gamma}_q(t)}{\Omega_{q0}}\right)^2
        +
        \left(\gamma_q(t)-\dfrac{1}{\gamma_q(t)}\right)^2
        \right]
        \\[8pt]
        \lambda_q(t,\beta)
        &= C_q(\beta)\sqrt{1+\mathcal{F}[\gamma_q(t)]}
        \\[8pt]
        \nu_q(t,\beta)
        &= C_q(\beta)\left(
        \sqrt{1+\mathcal{F}[\gamma_q(t)]}
        -
        \sqrt{\mathcal{F}[\gamma_q(t)]}
        \right)
    \end{aligned}$}
};
\coordinate (arrowABstart) at ($(origin)+(4.6,0)$);

\draw[->, very thick, black!60]
(arrowABstart |- ermakov.west)
--
(ermakov.west);

\draw[->, very thick, black!60] (ermakov.east) -- (bridge.west);


\node[labeltext, align=left, text width=14cm] (labelD) at (0, -0.25) {
    (d) Compute information-theoretic quantities
};

\coordinate (row2top) at (0, -1.85);

\node[obsbox,  text width=7.0cm, anchor=north west] (mi) at (0.1,0 |- row2top) {
    \textbf{Mutual Information}\\[8pt]
    {\normalsize
    $\begin{aligned}
        I(A\!:\!B,t,\beta)
        &= 2\sum_{q>0}
        \Bigl\{
        \mathcal{S}[\lambda_q(t,\beta)]
        - \mathcal{S}[C_q(\beta)]
        \Bigr\}
        \\[4pt]
        \mathcal{S}[x]
        &= \left(x+\frac{1}{2}\right)\ln\left(x+\frac{1}{2}\right)
        \\
        &\qquad-\left(x-\frac{1}{2}\right)\ln\left(x-\frac{1}{2}\right)
    \end{aligned}$}
};

\node[obsbox, anchor=north east] (ln) at (bridge.east |- row2top) {
    \textbf{Logarithmic Negativity}\\[8pt]
    {\normalsize
    $\begin{aligned}
        E_{\mathcal N}(t,\beta)
        &= \sum_{q>0}
        \max\!\bigg\{
        0,\,-\ln\left[2\nu_q(t,\beta)\right]
        \bigg\}
    \end{aligned}$}
};

\node[
    obsbox,
    text width=7.2cm,
    anchor=north
] (renyi) at ($(mi.north east)!0.5!(ln.north west)$) {
    \textbf{R\'enyi Entropies}\\[8pt]
    {\normalsize
    $\begin{aligned}
        S_{\alpha}^A(t,\beta)
        &= \sum_{q>0}\mathcal{S}_{\alpha}[\lambda_q(t,\beta)]
        \\[6pt]
        \mathcal{S}_{\alpha}[x]
        &= \frac{1}{\alpha-1}\ln\Bigl[
        \left(x+\dfrac{1}{2}\right)^\alpha
        -
        \left(x-\dfrac{1}{2}\right)^\alpha
        \Bigr]
    \end{aligned}$}
};

\coordinate (bus) at (bridge.south |- 0,-1.25);

\draw[thick, black!60] (bridge.south) -- (bus);

\draw[->, very thick, black!60] (bus) -| (mi.north);
\draw[->, very thick, black!60] (bus) -| (renyi.north);
\draw[->, very thick, black!60] (bus) -| (ln.north);
\end{tikzpicture}
}
\caption{{\bf Summary of the general formalism.}
Two initially decoupled identical Luttinger liquids are specified by their Luttinger parameter $K$, speed of sound $v$, mean density $n_0$ and length $L$. These parameters fix a Luttinger frequency scale $\omega = \pi n_0v/K$ and discrete momenta $q = N\pi/L$ with $N \in \mathbb{N}$. The formalism consists of four steps: (a) We choose a time-dependent tunneling protocol, i.e. a coupling function $J(t)$ that generates entanglement/correlation between the two Luttinger liquids. (b) The protocol determines the mode-resolved Ermakov equation, whose solution is denoted by $\gamma_q(t)$. (c) One constructs an intermediate composite function $\mathcal{F}[\gamma_q(t)]$, and computes the symplectic eigenvalues $\lambda_q(t,\beta)$ and $\nu_q(t,\beta)$. (d) Information-theoretic quantities such as mutual information, R\'enyi entropies, and logarithmic negativity can be compactly expressed in terms of these symplectic eigenvalues. 
Together, these ingredients provide a unified Gaussian framework for evaluating information-theoretic quantities that characterise field-space entanglement dynamics in tunnel-coupled Luttinger liquids.}
\label{fig:protocol_schematic}
\end{figure*}

\subsection{Mutual information (MI) dynamics \label{sec:MI_dynamics}}

For a pure global state $\rho = |\varphi\rangle\langle\varphi|$ and a bipartition of the system into $A$ and $B$, a standard entanglement measure is the entanglement entropy (EE) \cite{Calabrese2004} defined as
$S_{A} \equiv S(\rho_A)$,
where $\rho_{A}={\rm Tr}_{B}\rho$ is the reduced density matrix, and $S(\rho) = -\mathrm{Tr}\,(\rho\ln\rho)$ is the von Neumann entropy of the state $\rho$.
For Gaussian states, we have \cite{serafini2023quantum, Peschel2009,Peschel_1999}
\begin{equation}\label{eq:vNEnt_CVM}
	S(\rho) = \sum_{\lambda}\mathcal{S}[\lambda],
\end{equation}
where the sum runs over all the symplectic eigenvalues $\lambda$ of the covariance matrix $\mathbf{\Gamma}$ associated with $\rho$, and $\mathcal{S}[x] \equiv \left(x+\frac{1}{2}\right)\ln\left(x+\frac{1}{2}\right)-\left(x-\frac{1}{2}\right)\ln\left(x-\frac{1}{2}\right)$.
If $\rho$ is mixed, $S_A$ is no longer a faithful entanglement measure, since it also contains the entropy of the global state. 
A standard quantifier of correlations for mixed states is the mutual information (MI) 
\begin{equation}
    I(A:B) \equiv S_{A}+ S_{B} - S_{A\cup B}.\label{mutual_inf}
\end{equation}
The MI quantifies how much knowing $A$ reduces uncertainty about $B$ and vice versa. In particular, it reduces to twice the amount of EE for pure states (e.g. at zero temperature), while for general mixed states, it measures \textit{both} classical and quantum correlations between $A$ and $B$.

In out-of-equilibrium settings, the quantities of interest acquire a time dependence, which we make explicit below. In particular, the subsystem's entropy $S_{A,B}(t)$ can be calculated from Eq.~\eqref{eq:vNEnt_CVM}, given the symplectic eigenvalues of the time-dependent reduced covariance matrices $\mathbf{\Gamma}_q^{AA}(t) =  \mathbf{\Gamma}_q^{BB}(t)$, identical due to the $A \leftrightarrow B$ symmetry which holds even at finite temperatures, see Appendix \ref{appdxB:symp_eigval_mi_dynamics}. Consequently, the symplectic eigenvalues associated with $A$ and $B$ are identical $\lambda_q^A(t,\beta)=\lambda_q^B(t,\beta) = \lambda_q(t,\beta)$ and they are given by
\begin{align}
\lambda_{q}(t,\beta) &= C_{q}(\beta)\sqrt{1+\mathcal{F}[\gamma_q(t)]},
\label{eq: thermal_subsys_eigs}
\end{align}
with $C_q(\beta)$ defined in Eq.~\eqref{eq:C_q}. 
For compactness, we have defined a composite function
\begin{equation}
    \mathcal{F}[\gamma_q(t)]\equiv \frac{1}{4} \left[\left(\frac{\dot{\gamma}_q(t)}{\Omega_{q0}}\right)^2+\left(\gamma_q(t) - \frac{1}{\gamma_q(t)}\right)^2 \right] .
    \label{eq:ermakov_com_func}
\end{equation}

Meanwhile, the global entropy remains invariant under unitary evolution, i.e., $S_{A\cup B} (t)= S_{A\cup B} (0) $. Therefore, symplectic eigenvalues of the full covariance matrix are invariant with respect to time. 
Since they are assumed to be initially thermal, we have $\lambda_q^{\pm}(t,\beta) = \lambda_q^{\pm}(0,\beta) = C_{q}(\beta)$. This gives a total entropy of $S_{A \cup B} (0) = 2\sum_{q>0} \mathcal{S}[\lambda_q^+(0,\beta)]$, which vanishes in the zero temperature limit. Combining both subsystems' entropies and global entropy gives the following expression for time-dependent MI
\begin{equation}
    I(A:B, t,\beta) = 2\sum_{q>0}\left\{\mathcal{S}[\lambda_q(t,\beta)]-\mathcal{S}[C_q(\beta)]\right\}.\label{eq:mutual_info_dynamics}
\end{equation}
The dynamics of EE for the ground state can be easily recovered from the MI from $\beta \rightarrow\infty$ limit, i.e.,
\begin{equation}
    \lim_{\beta \to \infty} S_A(t,\beta)= \sum_{q>0}\mathcal{S}\left[\frac{1}{2}\sqrt{1+\mathcal{F}[\gamma_q(t)]}\right]\geq0.
    \label{eq:ent_entropy_zero_temp}
\end{equation}
Since $\gamma_q(0) = 1$ and $\dot{\gamma_q}(0)= 0$, it holds that $\mathcal{F}[\gamma_q(0)] = 0$, and by construction $\mathcal{F}[\gamma_q(t)]\geq 0$.
This means that if the system is initialized in the ground state, any tunneling protocol where $\mathcal{F}[\gamma_q(t)]$ increases would result in a dynamical generation of entanglement between the two Luttinger liquids.

Finally, the more general family of R{\'e}nyi entropies, $S_{\alpha}(\rho)=(1-\alpha)^{-1}\ln {\rm tr}[\rho^{\alpha}]$, for Gaussian states \cite{Camilo2019,serafini2023quantum} can be expressed as function of the symplectic eigenvalues of the covariance matrix as $S_\alpha(\rho) 
= \sum_{\lambda} \mathcal{S}_{\alpha}[\lambda]$, where 
\begin{align}
\mathcal{S}_{\alpha}[x] = \frac{1}{\alpha-1} \ln \left[
(x + {1}/{2})^\alpha
-
(x - {1}/{2})^\alpha\right]. \label{eq: Renyi entropies as function of symplectic eigenvalues}
\end{align}
Note that they converge to the von Neumann entropy in the replica limit, i.e., $\lim_{\alpha \to 1} S_{\alpha}=S$. 
Therefore, as a byproduct of our analysis, using the same expression in Eq.~\eqref{eq: thermal_subsys_eigs} for the symplectic eigenvalues $\lambda_q(t,\beta)$, we also get the dynamics of the R{\'e}nyi-entropies $S_{\alpha}(t,\beta)$, which serves as additional measures of entanglement at zero temperature. 
\subsection{Entanglement dynamics via logarithmic negativity (LN)}
To quantify entanglement in mixed states, we consider logarithmic negativity (LN), a measure based on the positive-partial-transpose (PPT) criterion \cite{Peres1996,Zyczkowski1998,Plenio2005}. For a state $\rho$ on a bipartite system $A\cup B$, the LN of $\rho$ is defined as $E_{\mathcal{N}} \equiv\log \|\rho^{T_{B}}\|_{1}$, where $\|O\|_{1}={\rm Tr}\sqrt{O^{\dagger}O}$ denotes the trace norm for the operator $O$, and $\rho^{T_{B}}$ denotes the partial transpose of $\rho$, such that $\langle \varphi_{A}\varphi_{B}|\rho^{{T}_{B}}|\varphi'_{A}\varphi'_{B}\rangle=\langle \varphi_{A}\varphi'_{B}|\rho|\varphi'_{A}\varphi_{B}\rangle$ for product basis $\{\ket{\varphi_A\varphi_B}\}$ on $A\cup B$. For general mixed states, $E_{\mathcal{N}}>0$ is only a sufficient condition for entanglement since bound (undistillable) entangled states exist~\cite{horodecki1998mixed}. However, for two-mode Gaussian states, and Gaussian $1\times N$ bipartitions, the PPT criterion is necessary and sufficient for separability \cite{duan2000inseparability,simon2000peres}. In our setting, different momenta do not mix, and hence the problem of multi-mode entanglement reduces to a collection of two-mode problems within each momentum sector $q$. More precisely, since the state factorizes over momentum sectors, the global state is separable if and only if every two-mode sector is separable; hence for this family of states, LN gives a \textit{necessary and sufficient} diagnostic and measure of entanglement.

Similar to the entropies, for Gaussian states LN can be expressed in terms of the symplectic eigenvalues but now of the \textit{partially transposed} covariance matrix ${\bf \Gamma}^{T_B}$ \cite{serafini2023quantum}
\begin{equation} \label{def_LN}
    E_{\mathcal{N}} = \sum_{\nu }\max\{0, -\ln[2\nu]\},
\end{equation}
where $\nu$ are the symplectic eigenvalues of ${\bf \Gamma}^{T_{B}}$. The partial transposition acts on the covariance matrix as a phase space reflection in the subsystem's quadrature \cite{simon2000peres}, which can be performed on each sector separately,
\begin{equation}  \mathbf{\Gamma}_q^{T_B}(t) = \mathscr{T}_B\mathbf{\Gamma}_q (t) \mathscr{T}_B,
\label{partial_transpose}
\end{equation}
where $\mathscr{T}_B = \text{diag}(1,1,1,-1)$. The operation in \eqref{partial_transpose} is equivalent to flipping the sign of the momentum quadratures of $B$, i.e. $(\delta\hat n_q^A,\hat\phi_q^A,\delta\hat n_q^B,\hat\phi_q^B)
\mapsto
(\delta\hat n_q^A,\hat\phi_q^A,\delta\hat n_q^B,-\hat\phi_q^B).$
Note also that while in \eqref{def_LN} the sum is over all symplectic eigenvalues $\nu$, due to the $\max$ function, those larger than $1/2$ do not contribute to LN. With that, the LN between two Luttinger liquids $A$ and $B$ in our case can be written as
\begin{equation}
    E_{\mathcal{N}}(t,\beta) = \sum_{q>0}\max\{0, -\ln[2\nu_q(t,\beta)]\}, \label{LN_sum_q}
\end{equation}
in terms of a subset of symplectic eigenvalues $\{\nu_q(t,\beta)\}$, whose explicit expression is given by (see Appendix~\ref{appendix_LN}) 
\begin{equation}
    \nu_q(t,\beta) = C_{q}(\beta)\bigg(\sqrt{1+\mathcal{F}[\gamma_q(t)]}-\sqrt{\mathcal{F}[\gamma_q(t)]}\bigg).
\label{eq:symp_eigvals_partial_transpose}
\end{equation}
Eqs. \eqref{LN_sum_q} and \eqref{eq:symp_eigvals_partial_transpose} fully determine the evolution of entanglement between A and B. 

In the zero temperature limit, Eq. \eqref{LN_sum_q} simplifies to
\begin{equation}
\lim_{\beta\rightarrow\infty}E_{\mathcal{N}}(t,\beta)=\sum_{q>0} \text{arcsinh}\left(\sqrt{\mathcal{F}[\gamma_q(t)]}\right) \geq 0,
    \label{eq:LN_zero_temp}
\end{equation}
which provides an upper-bound \cite{vidal2002computable, calabrese2012entanglementNeg} to the EE \eqref{eq:ent_entropy_zero_temp}.
Meanwhile, at finite temperature, a competition emerges between the tunneling encoded in the Ermakov factor $\gamma_q(t)$, and thermal effects characterised by the inverse temperature $\beta$. 
As temperature increases, the factor $\coth(\beta\hbar\Omega_{q0}/2)$ increases and thus $\nu_q(t,\beta)$ also increases, while ${E}_{\mathcal{N}}(t)$ decreases (cf. \eqref{LN_sum_q}).
Hence, increasing temperature always suppresses entanglement. By contrast, the tunneling contribution $\sqrt{1+\mathcal{F}[\gamma_q(t)]}-\sqrt{\mathcal{F}[\gamma_q(t)]}\leq 1$ counteracts this effect by reducing $\nu_q(t,\beta)$ and generating entanglement. 

The expressions presented in this section admit a simple physical interpretation: the symplectic eigenvalues $\lambda_q$ and $\nu_q$, and hence all the above correlation measures depend on the dynamics only through the function $\mathcal{F}[\gamma_q(t)]$.
This function is entirely determined by $\mathbf{U}_q(t)$ in Eq.~\eqref{eq:quadrature_dynamics}, which encodes the squeezing and shearing of the antisymmetric-sector mode through the Ermakov factor $\gamma_q(t)$ and its derivative $\dot\gamma_q(t)$.
By contrast, the symmetric sector evolves only through $\mathbf{R}_q(t)$, leaving symplectic eigenvalues unchanged.
Crucially, $\mathbf{U}_q(t)$ completely determines the relative dynamics between the symmetric and antisymmetric sectors, which precisely generates field-space entanglement between subsystems $A$ and $B$.
Indeed, when transformed back to the $A/B$ basis, the inter-field covariance matrix is
\begin{align}
{\bf \Gamma}^{AB}_{q}(t)
=\frac12\left({\bf \Gamma}^{++}_{q}(t)-{\bf \Gamma}^{--}_{q}(t)\right),
\end{align}
while the reduced covariance matrices are determined by the sum of the two sectors. If the symmetric and antisymmetric sectors evolved identically, no field-space entanglement would be generated, irrespective of their individual dynamics.
In the present model, this difference originates from the squeezing and shearing induced by the time-dependent tunneling in the antisymmetric sector.

Finally, the above interpretation also provides qualitative insight into the momentum dependence of entanglement and correlations.
For sufficiently large momenta, the two dispersions become asymptotically identical,
$\Omega_q^+\simeq\Omega_q^-$ (cf. Eq. \eqref{frequency_LL}), so the difference between their dynamics becomes progressively smaller. Consequently, high-momentum modes are expected to contribute only weakly to the entanglement.
Conversely, low-momentum modes experience the largest difference and therefore provide a dominant contribution.

In summary, we present a general solution for the dynamics of MI, LN, and Rényi entropies for arbitrary initial temperatures and tunneling protocols by tracking how symplectic eigenvalues evolve within the Gaussian approximation. This general solution will be used to derive a scaling law for entanglement in the early-time (Sec.~\ref{early_time_MI}), investigate long-time averaged behavior (Sec.~\ref{sec: long time average}), and convergence toward adiabatic limit for very slow tunneling protocol (Sec.~\ref{sec: adiabatic smooth_quench_protocol}), see Table~\ref{table_section_contents} for the scope of the next sections. 

\begin{table}[h]
\centering
\scriptsize
\setlength{\tabcolsep}{4pt}
\renewcommand{\arraystretch}{1.15}
\label{tab:main_text_scope}
\resizebox{\columnwidth}{!}{%
\begin{tabular}{l|ccc|cccc}
\hline

& MI 
& LN 
& R\'enyi entropies 
& $T=0$ 
& $T > 0 $ 
& TDL 
& finite $L$ \\
\hline
Sec.~\ref{early_time_MI}
& \cmark
& \xmark
& $\circ$
& \cmark
& \xmark
& \cmark
& \cmark
\\
Sec.~\ref{sec: long time average}
& \cmark
& \cmark
& \cmark
& \cmark
& \cmark
& \cmark
& \cmark
\\
Sec.~\ref{sec: adiabatic smooth_quench_protocol}
& \cmark
& \cmark
& $\circ$
& \cmark
& \cmark
& \xmark
& \cmark
\\
\hline
\end{tabular}
}
\caption{
Scope of results in Secs.~\ref{early_time_MI}--\ref{sec: adiabatic smooth_quench_protocol}. A $\cmark$ indicates that the corresponding quantity or scenario is explicitly analysed, a $\xmark$ means it is not addressed, while $\circ$ represents the cases that are not being analysed but can be straightforwardly treated using the analytically derived symplectic eigenvalues.
Here, $T$ denotes temperature, $L$ is the system size, and TDL is short for thermodynamic limit. 
\label{table_section_contents}
}
\end{table}
\section{Early entanglement dynamics at zero initial temperature \label{early_time_MI}}
Here, we analyse field-space early-time entanglement growth.
Recall that the two fields in the present setting interact directly at every point in space, hence the distance between two subsystems plays no role, and no propagating entanglement front is present. It is then natural to ask if field-space entanglement exhibits simple universal scaling laws, analogous to the well-established scaling behaviour of spatial entanglement. We focus on the zero temperature case so that MI is sufficient to capture entanglement dynamics.
At finite temperature, proper characterization of the entanglement requires studying LN. However, obtaining its early-time dynamics is more challenging, so we leave this to future work.
 
We consider two classes of tunneling protocols: the sudden quench and finite-ramp protocols. In the sudden quench scenario, the tunneling strength is discontinuously changed from $J(0^-) =0$ to $J(0^+) = J_f$. This corresponds to 
\begin{align}
\Omega^{2}_{q}(t)=\begin{cases}
\Omega^{2}_{q0}=(vq)^{2} & \text{if}\quad t \leq 0\\
\Omega^{2}_{qf}=\Omega_{q0}^2+4\omega J_f & \text{if}\quad t>0,
\end{cases}\label{eq:sudden_quench}
\end{align}
where the solution to the Ermakov equation is known~\cite{Dupays2021}:
\begin{equation}
    \gamma_q(t) = \sqrt{1-\left(1-\zeta^{2}_{q}\right)\sin^2(\Omega_{qf} t)}, \label{eq:solution_sudden_quench}
\end{equation}
with $\zeta_{q}\equiv\Omega_{q0}/\Omega_{qf}$. From this exact Ermakov solution, the dynamics of MI can be derived according to the prescription given in the previous section. Instead, we study a simpler approximate solution at the early time regime by Taylor expanding Eq.~\eqref{eq:solution_sudden_quench} as detailed in Appendix~\ref{Appendix_early_time_sudden_quench}. 

For finite-ramp protocols, we study early-time evolution by performing a Taylor expansion of the tunneling protocol around $t = 0$ as
\begin{align}
   J(t)&=\frac{J^{(n)}(0)}{n!}t^{n}+ \mathcal{O}(t^{n+1}),
\end{align}
where $J^{(k)}(t)$ denotes the first non-zero $k^{\rm th}$ derivative of $J(t)$ and $n$ is a parameter defined by
\begin{align} \label{parameter_n}
    n=\min \;\{ k \geq 0 \,:\, J^{(k)}(0) \neq 0 \}.
\end{align}
Here, we introduce the convention $J^{(0)}(0) = J_f$, so that the sudden quench is associated with  $n=0$. This Taylor expansion approach is valid as long as $t\ll t_n^*$ where
\begin{align}
t^{*}_{n} =
\begin{cases}
\dfrac{1}{\sqrt{4\omega |J_f|}} & \text{for}\; n=0 \  \\
\left[\dfrac{(n+2)!}{4\omega |J^{(n)}(0)|}\right]^{\dfrac{1}{n+2}} & \;\text{for} \; n \ge 1 .
\end{cases}
\label{t_star}
\end{align}
In Appendix~\ref{Ermakov_equation_early_time}, we show that in this regime, 
 \begin{align}
     \gamma_q(t) &\approx 1- \frac{4\omega J^{(n)}(0)}{(n+2)!}t^{n+2}, 
     \label{eq:gamma_t_early_time}
\end{align}
for both the sudden quench and the finite-ramp, leading to approximated symplectic eigenvalues of the form
\begin{align}
\lambda_{q}(t,\beta)&\approx\frac{1}{2}\sqrt{1+\left(\frac{\varepsilon_{n}(t)}{q}\right)^{2}} \ , \label{eq: early_symplectic_eigenvalue}
\end{align}
with a time- and protocol-dependent momentum scale
\begin{align}
\varepsilon_{n}(t)&\equiv\frac{2\omega J^{(n)}(0)}{v(n+1)!}t^{n+1} \ . \label{eq: epsilon t}
\end{align}

\begin{figure*}[!htbp]
    \centering
\includegraphics[width=0.95\linewidth]{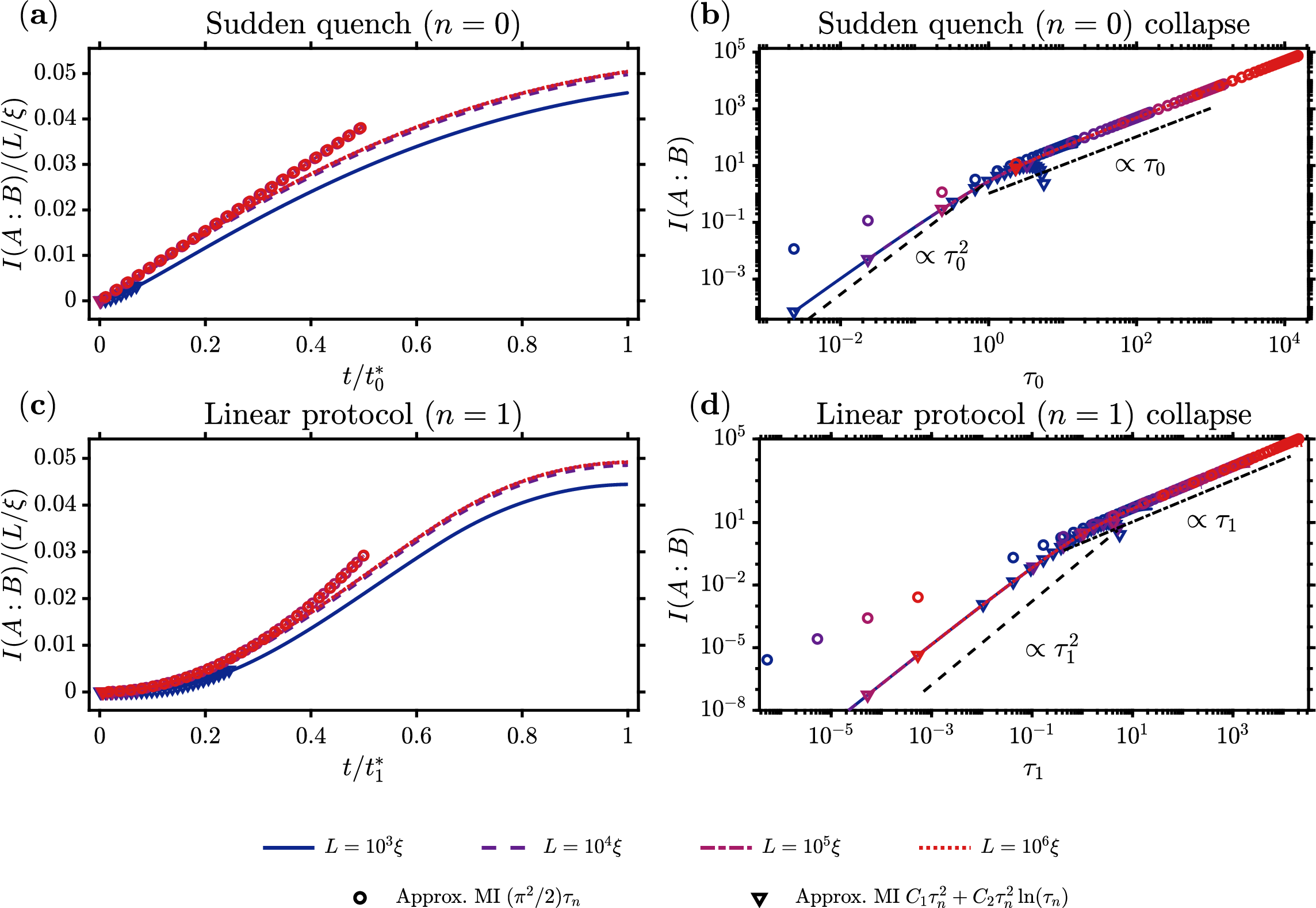}
\caption{
{\bf Comparison between approximate and exact numerical solution to the mutual information at early times and different system sizes.} Solid lines denote the exact numerical MI computed from the numerical solution of Ermakov equation with Eq.~\eqref{eq:mutual_info_dynamics}. Triangle and circle markers represent the approximate MI in the limits $\tau_n\ll1$ and $\tau_n^*\gg\tau_n\gg1$, respectively, given by Eqs.~\eqref{eq: MI_small_alpha} and \eqref{eq: MI propto epsilon}. Colors indicate different system lengths.
Panels (a) and (b) show results for a sudden quench ($n=0$), with panel (a) showing the normalized MI as a function of time, and panel (b) the MI as a function of $\tau_0$. Panels (c) and (d) show results for a linear protocol ($n=1$). We stop plotting triangle markers when the approximation $\tau_{n}\ll 1$ breaks down, this is given by $t^{(0)}_{\rm cut}/t^{*}_{0}=(v\pi)/(\sqrt{\omega J_{f}}L)$ for the sudden quench, and by $t^{(1)}_{\rm cut}/t^{*}_{1}=\sqrt{(2v\pi t^{*}_{1})/(3L)}$ for the linear protocol. We stop plotting the circle markers at $t=0.5 t^{*}_{n}$ as the approximated expressions are only valid for $t/t^{*}_{n}\ll 1$. Details on how $t^{(n)}_{\rm cut}$ is obtained are given in Appendix~\ref{subsec: early time breakdown}.  The plots are made for a final coupling strength $J_f=2\pi\times5\,\mathrm{Hz}$, mean density $n_0=75\,\mu\mathrm{m}^{-1}$, transverse trapping frequency $\omega_\perp=2\pi\times2\,\mathrm{kHz}$, scattering length $a_s=5.2 \;\mathrm{nm}$, atomic mass $m=86.91\;u$ for $^{87}$Rb, and healing length $\xi=0.24 \;\mu\mathrm{m}$. For the linear protocol, ramp duration is $t_f=1\,\mathrm{ms}$. The interaction strength is $g=2\hbar a_s\omega_\perp(2+3a_s n_{0})/(1+2a_s n_{0})^{3/2}$, and the Luttinger parameters are $v=\sqrt{gn_0/m}$ and $K=\hbar\pi\sqrt{n_0/(mg)}$.}
\label{figure_ealy_time}
\end{figure*}

Combining Eq.~\eqref{eq:mutual_info_dynamics} and Eq.~\eqref{eq: early_symplectic_eigenvalue}, one obtains early-time approximation of the MI.
%
Upon summing over all momenta, the MI
exhibits an asymptotic scaling form with the parameter $\tau_n \equiv
{\varepsilon_n(t)}/{\Delta q}$ where we denote by $\Delta q = \pi/L$ the increment between two consecutive momenta.
Two asymptotic regimes are identified (full derivation in Appendix~\ref{mutual_info_early_time}):
\begin{enumerate}[leftmargin=*]
    \item {\it Ultra-early time/finite-size regime}: In the limit $\tau_n \ll 1$, we have $\varepsilon_n(t)/q \ll 1$ for all momenta $q$. Hence this ratio can be used as small parameter to perform a perturbative expansion of the symplectic eigenvalues. Summing explicitly over momenta leads to 
    \begin{align}
I(A:B,t,0)&\approx C_{1}\tau^{2}_{n}+C_{2}\tau^{2}_{n}\ln(\tau_{n}),\label{eq: MI_small_alpha}
\end{align}
with coefficients $C_{1}=\frac{\pi^{2}}{12}\left[1+2\ln(2)\right]-\zeta'(2)$ and $C_{2}=-\frac{\pi^{2}}{6}$, where $\zeta'(\cdot)$ denotes the first derivative of the Riemann zeta function.
This \textit{transient} holds in a time window that vanishes in the thermodynamic limit (TDL) of $L \to \infty$, but could still be relevant for experiments. 

\item \textit{Early-time scaling regime}: here we consider $\tau_n^*\gg \tau_n \gg 1$, where $\tau^{*}_{n}$ is $\tau_{n}$ evaluated at $t = t_n^{*}$. This interval is well-defined even in the TDL, where 
$\Delta q \ll \varepsilon_n(t)$, so that 
the summation over $q$ is well approximated by an integral. Now, the ratio $\varepsilon_n(t)/q$ can be small or large depending on the value $q$. However, the low-momentum modes such that $\varepsilon_n(t)/q \gg 1$ dominate the integral. Using this ratio as large parameter for a perturbative expansion, the integral can be approximated, leading to a power-law scaling
\begin{align}
I(A:B,t,0)\approx \frac{\pi^2}{2} \tau_n,
\label{eq: MI propto epsilon}
\end{align}
recovering the extensive behavior of the MI with system size $L$ as obtained in equilibrium~\cite{Furukawa2011, Lundgren2013, yoshino2021intercomponent}.
\end{enumerate}
It is worth noting that 
Eq.~\eqref{eq: early_symplectic_eigenvalue} is valid only up to a cutoff. 
Nonetheless, the qualitative analysis is unaffected. 
As discussed in Sec.~\ref{sec: Entanglement_dynamics_study}, for large momenta the symmetric and antisymmetric sectors evolve almost identically, so high-momentum modes contribute only weakly to entanglement. The cutoff therefore produces only tiny corrections to MI. 

These results show that field-space entanglement exhibits simple early-time scaling laws despite the absence of a propagating entanglement front. 
In particular, within the early-time scaling regime in Eq. \eqref{eq: MI propto epsilon}, a sudden quench ($n=0$) gives linear growth of mutual information, whereas smoother protocols produce a hierarchy of power laws, determined solely by the first non-vanishing derivative of the tunneling protocol at $t=0$.
While a detailed comparison with spatial entanglement under generic driving remains an interesting open question, our results demonstrate that universal scaling behaviour can emerge also for field-space entanglement. 
At the microscopic level, however, the mechanism is fundamentally different. In the present tunnel-coupled Gaussian model, entanglement is generated by the squeezing and shearing of antisymmetric-sector induced by the time-dependent tunneling, rather than by the ballistic propagation of entangled quasiparticle pairs across a spatial bipartition.

We support the theoretical analysis by demonstrating its agreement with exact numerical solution for different system sizes in Fig.~\ref{figure_ealy_time}. We use physical parameters associated with existing ultracold 1D Bose gas experiments~\cite{tajik2023verification, aimet2025experimentally}.
Overall, we do observe cross-over from
$ I(A:B) \propto t^{2n+2}$, up to logarithmic corrections to a lower-order power-law scaling $I(A:B)\propto  t^{n+1}$, with the latter dominating as the system size increases. 
We consider examples of a sudden quench ($n=0$) in Figs.~\ref{figure_ealy_time} (a) and (b), and a linear protocol ($n=1$) in Figs.~\ref{figure_ealy_time} (c) and (d). Panels (a) and (c) focus on the agreement between the approximated vs. exact solutions at early time. We find good agreement between the $\tau_{n}\ll 1$ approximation (triangle markers) given by Eq.~\eqref{eq: MI_small_alpha} and the $\tau_{n}\gg 1$ approximation (circle markers) given by Eq.~\eqref{eq: MI propto epsilon}, with the exact numerical result in the early time regime $t\ll t^{*}_{n}$ (in colored lines).
With increasing system size, the regime $\tau_{n}\gg1$ becomes predominant, and the MI exhibits a power law $\propto t^{n+1}$, which is linear for the sudden quench ($n=0$) and quadratic for the linear protocol ($n=1$).

Panels (b) and (d) show the scale-invariant behavior of MI at early times: regardless of system sizes, the rescaled data and the approximated expressions collapse onto a single curve in their respective regime of validity. We can provide an estimate of when the approximation $\tau_{n}\ll 1$ fails. The $\tau_{0}\ll 1$ approximation, which corresponds to the quadratic behavior in $\tau_{0}$, is expected to break down for $t^{(0)}_{\rm cut}=\frac{v\pi}{\sqrt{\omega J_f}\,L}t^{*}_{0}$ (see Appendix \ref{subsec: early time breakdown}). For the chosen parameters, this time is $v\pi/(\sqrt{\omega J_{f}}L)t^{*}_{0}\approx 0.06 t^{*}_{0}$ for $L=10^{3}\xi$ where $\xi$ is the healing length (see caption of Fig. \ref{figure_ealy_time}), and gets progressively smaller for longer system sizes. In the linear protocol case, the breakdown of the $\tau_{1}\ll 1$ condition happens at $t^{(1)}_{\rm cut}\approx 0.205\,t^{*}_{1} $ for $L=10^{3}\xi$, $t^{(1)}_{\rm cut}\approx 0.064\,t^{*}_{1}$ for $L = 10^4 \xi$, and gets progressively smaller for longer system sizes.   

\section{Long-time averages of entanglement for finite initial temperatures}\label{sec: long time average}

This section studies the long-time behavior of MI, LN and R{\'e}nyi-2 entropy for sudden quench and finite-ramp protocols. We consider protocols where $J(0)=0$, with arbitrary $J(0<t<t_{f})$ during the ramp, and $J(t>t_{f})=J_{f}$ after the ramp. The sudden quench is a limiting case with $t_f \rightarrow 0^+$. We provide analytical expressions for the finite-size time averages of these quantities at arbitrary temperatures. Since the quantities of interest are nonlinear functions of the covariance matrix, the relations between finite-size time averages, thermodynamic-limit dynamics, and possible stationary values are not straightforward. Throughout this section, we first define the long-time averages at finite system size; the thermodynamic-limit scaling is obtained only afterwards by replacing the resulting momentum sums with integrals.

We then proceed to show that both the averaged MI and LN decrease with increasing temperature, with averaged LN vanishing above a threshold temperature that we characterise.
Moreover, at zero temperature and in the large $L$ limit, we show that the averaged LN, MI, and Rényi-2 entropy after a sudden quench scales linearly with $ L\sqrt{4\omega J_f}/v$.

\subsection{Sudden quench at finite temperatures \label{subsec: sudden quench long time average}}

We start by considering the simplest analytically tractable case of sudden quench protocol defined by $J(0)=0$ and $J(t> 0)=J_f\neq 0$. Here, the solution to the Ermakov equation was provided in~\eqref{eq:solution_sudden_quench}. This can be used to determine the symplectic eigenvalues necessary to compute MI, LN, and Rényi-2 entropy. To compute the MI and R{\'e}nyi-2 entropy, substituting Eq.~\eqref{eq:solution_sudden_quench} into Eqs.~\eqref{eq: thermal_subsys_eigs} -- \eqref{eq:ermakov_com_func} gives us 
\begin{align} 
\lambda_{q}(t,\beta)= C_{q}(\beta)\sqrt{1+A_q \sin^{2}(\Omega_{qf} t)}
,\label{eq: lambda for sudden quench}
\end{align}
with the oscillation amplitudes
\begin{align}
    A_q&=\dfrac{\left(4\omega J_f\right)^{2}}{4 (vq)^2[(vq)^2+4\omega J_f]}. \label{A_q_parameter}
\end{align}
To compute the LN, substituting Eq.~\eqref{eq:solution_sudden_quench} into Eq.~\eqref{eq:symp_eigvals_partial_transpose} gives
{\small
\begin{align}
    \nu_q(t,\beta)=C_q(\beta)\left(\sqrt{1+A_q \sin^{2}(\Omega_{qf} t)}-\sqrt{A_q \sin^{2}(\Omega_{qf} t)}\right). \label{eq: nu for sudden quench}
\end{align}}

We want to investigate the long-time averaged behavior defined by $\overline{f}={\rm lim}_{t \to \infty}\frac{1}{t}\int_{0}^{t}f(s)\;ds$ for a smooth function $f$. Importantly, the Ermakov solution for the sudden quench is a periodic function (with a mode-dependent period), which implies that all symplectic eigenvalues are also periodic functions (cf. Eq.~\eqref{eq: lambda for sudden quench} and Eq.~\eqref{eq: nu for sudden quench}). Given this, the calculation of long-time averages reduce to averaging over a single (mode-dependent) period $\pi/\Omega_{qf}$. The expressions of long-time averages for the quantities of interest can then be recast in terms of a new variable $u \equiv \Omega_{qf}t$ as:
\small
\begin{align}
      \overline{I(A:B,\beta)} &= \frac{2}{\pi}  \sum_{q >0}\int_0^{\pi}\left\{
        \mathcal{S}\left[2\lambda_{q}\left(\frac{u}{\Omega_{qf}},\beta\right)\right]
        - \mathcal{S}[C_q(\beta)]
        \right\} \; du,\label{eq: MI for sudden quench at finite system size} \\ 
     \overline{E_{\mathcal{N}}(\beta)}&=\frac{1}{\pi} \sum_{q >0}  \int_0^{\pi} \max \left\{0, -\ln\left[2\nu_{q}\left(\frac{u}{\Omega_{qf}},\beta\right)\right] \right\}  du  \ ,\\
    \overline{S_{2}^A(\beta)}&=  \frac{1}{\pi}  \sum_{q >0} \int_0^{\pi}\ln \left[2\lambda_{q}\left(\frac{u}{\Omega_{qf}},\beta\right)\right]\; du   \ , \label{eq: S2 integral form for a sudden quench}
\end{align}
\normalsize
where $\lambda_q, \nu_q$ are given by Eq.~\eqref{eq: lambda for sudden quench} and Eq.~\eqref{eq: nu for sudden quench} respectively. Moreover, the integral form of the $S_2$ entropy in Eq.~(\ref{eq: S2 integral form for a sudden quench}) can be directly generalized to Rényi entropies of arbitrary order via Eq.~(\ref{eq: Renyi entropies as function of symplectic eigenvalues}).
\begin{table*}[!htbp]
\centering
%
\renewcommand{\arraystretch}{1.6}
\setlength{\arrayrulewidth}{1pt}
{\arrayrulecolor{rosedeep}
\begin{tabular}{@{}l >{\centering\arraybackslash}m{13cm} c@{}}
\toprule
\rowcolor{blush}\hd{Quantity} & \multicolumn{1}{c}{\hd{Expression}} & \hd{Proof} \\
\arrayrulecolor{rose}\specialrule{0.9pt}{0pt}{0pt}\arrayrulecolor{rosedeep}
\addlinespace[0.6em]
MI ($\beta\to\infty$) &
$\displaystyle \sum_{q >0} \Big\{ \ln A_q +2\sqrt{1+A_q} \left[
F\left(\frac{\pi}{2},k_q \right)
- E\left(\frac{\pi}{2},k_q \right) \right] +2(1-\ln 4)\Big\}$ &
\pf{App.~\ref{sec: MI zero temperature limit and sudden quench}} \\
\addlinespace[0.8em]
MI ($\beta \geq 0$) &
$\displaystyle
4\sum_{q >0}  \left\{ C_q \sqrt{1+A_q}
\left(F \left[\theta_q, k_q\right]
- E\left[\theta_q, k_q\right] \right)+ \frac{1}{2} \left[ 1 + \ln \left( Z_q\right) \right] - \frac{1}{2} C_q \ln\left(\frac{2C_q+1}{2C_q-1}\right) \right\}$ &
\pf{App.~\ref{sec: MI zero temperature limit and sudden quench}} \\
\addlinespace[0.8em]
LN ($\beta\to \infty$) &
$\displaystyle \frac{2}{\pi}  \sum_{q>0} \mathrm{Ti}_2\left(\sqrt{A_q}\right)$ &
\pf{App.~\ref{sec: LN for sudden quench at T=0}} \\
\addlinespace[0.8em]
LN ($\beta\geq0$) &
$\displaystyle \frac{1}{\pi} \sum_{q >0:\; A_q \geq V_q^2}
\left\{
\phi_q\ln \left(\frac{A_q}{4 C_q^2}\right)
+\mathrm{Im}\left(
\Li_2 \left[\frac{1}{4C_q^2}e^{2i\phi_q} \right]
-\Li_2[-e^{2i\phi_q}]
\right)
\right\}$ &
\pf{App.~\ref{sec: LN for sudden quench at T >0}} \\
\addlinespace[0.8em]
$S_2$ ($\beta\geq0$) &
$\displaystyle
\sum_{q>0} \left\{
\ln [2C_q] + \ln\left[\frac{1+\sqrt{1+A_q}}{2}\right]
\right\}$ &
\pf{App.~\ref{expression_Renyi2_arbitrary}} \\
\addlinespace[0.5em]
\bottomrule
\end{tabular}}
\par\vspace{0.7em}
{\renewcommand{\arraystretch}{1.4}\arrayrulecolor{rosedeep}
\begin{tabular}{@{}c@{\hspace{3.5em}}c@{}}
\toprule
\rowcolor{blush}\hd{Primary variables} & \hd{Auxiliary variables} \\
\arrayrulecolor{rose}\specialrule{0.9pt}{0pt}{0pt}\arrayrulecolor{rosedeep}
\addlinespace[0.6em]
$\begin{aligned}
A_q(J_f) &= \frac{(4\omega J_f)^2}{4(vq)^2[(vq)^2+4\omega J_f]} \\[0.8em]
C_q(\beta) &= \frac{1}{2}\coth\left(\frac{\beta \hbar vq}{2}\right) \\[0.8em]
V_q(\beta) &= \csch\left(\beta \hbar vq\right)
\end{aligned}$
&
$\begin{aligned}
k_q(A_q) &= (1+A_q)^{-1/2} \qquad 
\theta_q(C_q) = \arcsin\left(C_q^{-1}/2\right) \\[0.8em]
\phi_q(A_q,V_q) &= \arcsin\sqrt{\frac{A_q-V_q^2}{1+A_q}} \\[0.8em]
Z_q(C_q,A_q) &= \frac{\sqrt{4C_q^2(1+A_q)-1}+\sqrt{4C_q^2-1}}{2\sqrt{4C_q^2-1}}
\end{aligned}$ \\
\addlinespace[0.5em]
\bottomrule
\end{tabular}}
\caption{Exact expressions of the long time-averaged MI and LN after a sudden quench at inverse temperature $\beta$, with the $\beta\to \infty$ results included as limiting cases.
The table involves several special functions. The incomplete elliptic integrals of the first and second kinds are denoted by $\textstyle F[\theta,k]=\int_{0}^{\theta}d\phi\;(1/\sqrt{1-k^{2}\sin^{2}\phi})$ and $\textstyle E[\theta,k]=\int_{0}^{\theta}d\phi\;\sqrt{1-k^{2}\sin^{2}\phi}$, respectively, with modulus $k$ and amplitude $\theta$.
The inverse tangent integral $\mathrm{Ti}_2(x)$ is given by $\textstyle \mathrm{Ti}_2(x)=\int_0^x \frac{\tan^{-1}t}{t}\,dt$. 
The function $\operatorname{Li}_2(z)$ is the dilogarithm, defined by
$ \textstyle  \operatorname{Li}_2(z)= \sum_{n=1}^{\infty} \frac{z^n}{n^2}\;,\;\; \text{ for } \;\; |z| \leq 1 $. For short, we write $C_{q}(\beta)$ as $C_{q}$ in the table. \label{Results_sudden_quench_table}}
\end{table*}
\begin{figure*}[!htbp] 
\centering
\includegraphics[width=0.9\linewidth]{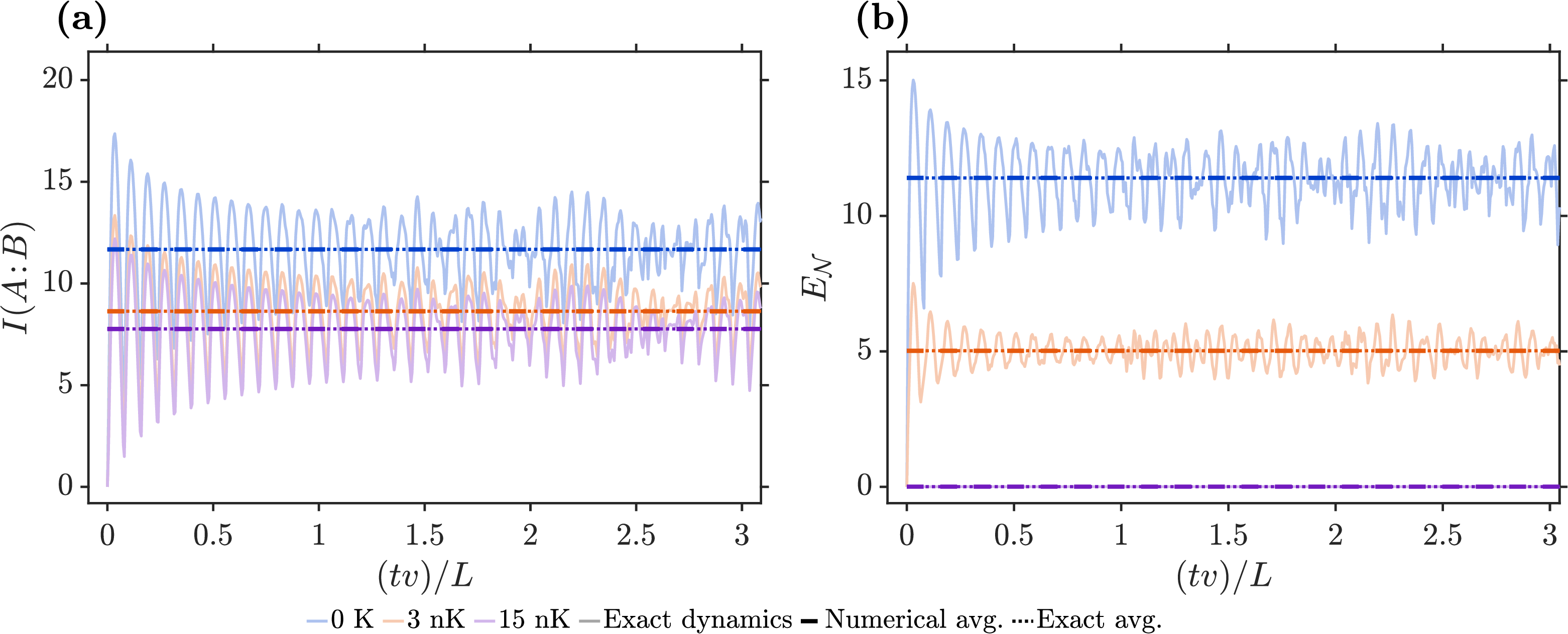}
    \caption{
    {\bf Comparison between (a) the mutual information and (b) the logarithmic negativity for a sudden quench at different temperatures.} The different colors correspond to the different temperatures:  $T=0\;K$ (plain blue), $T=3\; {\rm nK}$ (plain orange), and $T=15\; {\rm nK}$ (plain violet). The numerical time-average of the oscillatory curve is shown in dashed line for each temperature, and this is compared to the exact expressions of Table~\ref{Results_sudden_quench_table} in dotted line. We use the same plotting parameters as in Fig.~\ref{figure_ealy_time}, with a system length $L=400\ \xi$. The timescale $L/v$ provides a useful reference time for finite-size effects, but the precise onset of the observed oscillations depends on the distribution of post-ramp frequencies contributing to the chosen information-theoretic quantity.}
    \label{figcomparison_mutual_logneg_sudden_quench}
\end{figure*}
We summarize the analytical evaluation of Eqs. \eqref{eq: MI for sudden quench at finite system size} -- \eqref{eq: S2 integral form for a sudden quench} in Table~\ref{Results_sudden_quench_table}. The MI is expressed in terms of incomplete elliptic integrals of the first and second kinds, while the finite-temperature LN
involves Dilogarithm functions. These analytical expressions are numerically verified in
Fig.~\ref{figcomparison_mutual_logneg_sudden_quench}.
Since the Rényi-$2$ entropy is not central to the entanglement dynamics studied here, its plot is deferred to Appendix~\ref{expression_Renyi2_arbitrary}.

Fig.~\ref{figcomparison_mutual_logneg_sudden_quench} (a) displays the
evolution of the MI, while Fig.~\ref{figcomparison_mutual_logneg_sudden_quench} (b)
shows the LN following a sudden quench. Both quantities exhibit pronounced
oscillations whose amplitude decreases at early times, followed by smaller and
more irregular modulations at later times.
This behaviour can be understood from the mode dependence of the post-quench
frequencies, $\Omega_{qf}=\sqrt{(vq)^2+4\omega J_f}$. The symplectic
eigenvalues depend on time through factors of $\sin^2(\Omega_{qf}t)$, so each
mode oscillates with period $T_N=\pi/\Omega_{qf}$, with $q=N\pi/L$. 
In the
low-momentum regime, $(vq)^2\ll 4\omega J_f$, the dispersion is effectively
massive and the frequencies are generically incommensurate. This leads to
dephasing between different modes and hence a reduction of the oscillation
amplitude.
By contrast, in the high-momentum regime $(vq)^2\gg 4\omega J_f$, the
dispersion becomes approximately linear, $\Omega_{qf}\simeq vq=Nv\pi/L$.
Therefore $T_N\simeq L/(Nv)$, so the corresponding periods are integer
fractions of $L/v$. These modes can thus partially rephase at time
$L/v$, producing finite-size revival-like features.
For the parameters considered here, the crossover momentum defined by
$(vq)^2\sim 4\omega J_f$ lies below the ultraviolet cutoff $\Lambda$. Both the
massive low-momentum modes and the approximately massless high-momentum modes
therefore contribute to the dynamics. The observed behaviour should thus be
interpreted as the result of a competition between dephasing in the
low-momentum sector and partial rephasing of the high-momentum modes, leading
to damped oscillations and irregular finite-size modulations rather than a
simple monotonic approach to a stationary value.

The exact averaged expressions from Table~\ref{Results_sudden_quench_table} for the MI and LN are displayed in dotted colored lines in Fig. \ref{figcomparison_mutual_logneg_sudden_quench}. They can be compared to the respective dashed colored lines, which are the numerical averages for the MI and LN. We find a good agreement between the numerical averages and the discrete sum formula at zero and finite temperatures. 

{\it Finite temperature effects.---} One can observe that the time-averaged value of the MI decays as the temperature increases in Fig.~\ref{figcomparison_mutual_logneg_sudden_quench}.
We interpret this behaviour as resulting from the competition between quench-induced shared correlations and local thermal fluctuations. As the temperature is increased, the reduced entropies $S_A$ and $S_B$ grow due to the enhanced thermal mixedness. However, a large part of this additional entropy is local and therefore also contributes to the joint entropy $S_{A\cup B}$, which is subtracted in the definition of the MI, Eq.~\eqref{mutual_inf}. 
At the same time, the quench generates correlations genuinely shared between the two LLs, which increases MI. At higher temperatures, these shared quench-induced correlations become progressively less significant compared to the thermal background, leading to an overall suppression of MI.
By contrast, as LN only captures quantum correlations, the thermal noise eventually suppresses entanglement until it completely vanishes at a threshold temperature $\mathcal{T}_{*}$. This threshold can be determined for each mode from the expression of LN, as we discuss in Subsec. \ref{subsec:threshold_temperature}.
\subsection{Finite-ramp at finite temperatures\label{subsec:ln_dynamics} }
Consider now a finite-ramp protocol of the form $J(0)=0$, $J(0<t<t_{f})$ arbitrary, and $J(t>t_{f})=J_{f}$. 
After the protocol time ($t>t_{f}$), the frequency is fixed, so that the solution to the Ermakov equation is the one for fixed frequency with modified initial conditions $\gamma_{q}(t_{f})$ and $\dot{\gamma}_{q}(t_{f})$.  
Hence, for $t>t_{f}$, the dynamics is qualitatively equivalent to that of a sudden quench, with the end of the ramp playing the role of the new initial time and with modified initial conditions at time $t_{f}$. 
We demonstrate (Appendix~\ref{exact_solution}) that the Ermakov factors after any finite-ramp is
\begin{align}
   \gamma_q(t\geq t_{f})
    =\sqrt{M_{q}+N_q-2N_q\sin^2(\theta_q)},\label{general_solution_ermakov}
\end{align}
where we define
\begin{align}
M_q&=\frac{1}{2} \left[ \gamma_q^2(t_{f})+ \left(\frac{\dot \gamma_q(t_{f})}{\Omega_{qf}}\right)^2+ \left(\frac{\zeta_q}{\gamma_q(t_{f})} \right)^2\right],\\
N_q&=\sqrt{M_q^2-\zeta_q^2},\quad
\theta_q=\Omega_{qf}(t-t_{f})-\vartheta_q/2,
\end{align}
and $\vartheta_q$ is a mode-dependent phase-shift parameter that does not affect the long-time average (full expression {in Appendix \ref{exact_solution}}). Here we used $\zeta_q = \Omega_{q0}/\Omega_{qf}$ where $\Omega_{qf} = \Omega_{q}(t_f)$.

{\it Behaviour of the MI, LN, and R{\'e}nyi-2 Entropy ---}
Similar to the sudden-quench case, using the symplectic eigenvalue expressions for arbitrary finite-ramp protocols, we derive integral representations for the long-time averages of the MI, LN, and Rényi-2 entropy. To average over interval $t \rightarrow \infty$, it suffices to compute the average from $t_f$ as $\lim_{t\rightarrow\infty}\frac{1}{t}\int_{0}^{t}f(s)\; ds = \lim_{t\rightarrow\infty}\frac{1}{t}\int_{t_f}^{t}f(s)\; ds$ assuming the integral from $[0,t_f]$ is bounded. In particular, Eqs.~\eqref{eq: MI for sudden quench at finite system size} -- \eqref{eq: S2 integral form for a sudden quench} remain valid, but in this case the symplectic eigenvalues are
\begin{align}
    \lambda_{q}(t\geq t_{f})&= C_{q}(\beta)\sqrt{D_q + E_q\sin^2 \theta_q}, \label{exact_symplectic}\\
     \nu_q(t\geq t_{f})&=C_q(\beta)\left(\sqrt{D_q+E_q \sin^{2}\theta_q}-\sqrt{E_q \sin^{2}\theta_q}\right),
     \label{exact_symplectic_2}
\end{align}
where $E_{q}=N_q\left(\frac{1}{\zeta^{2}_{q}}-1\right)/2,$ and
\begin{align}
     D_{q}&=\frac{1}{2}+\frac{M_{q}}{4}\left(1+\frac{1}{\zeta^{2}_{q}}\right)-\frac{E_q}{2}.
     \label{value_Dq}
\end{align}
Note that $D_q, E_q$ are fully determined by the ramp protocol (Appendix~\ref{symplectic_eigenvalues}). {The sudden quench limit is recovered from the finite-ramp case by taking the limit $D_{q}\to 1$, and  $E_{q}\to A_{q}$.}
Unlike the sudden-quench, the integrals do not always yield closed-form expressions. 
Still, for arbitrary protocols we obtain exact analytical results for the MI and Rényi-2 entropy, while the LN integral remains analytically intractable. 
The final expressions are summarized in Table~\ref{Table_general_quench}. 

\begin{table*}[!htbp]
\centering
%
\renewcommand{\arraystretch}{1.6}
\setlength{\arrayrulewidth}{1pt}
{\arrayrulecolor{rosedeep}
\begin{tabular}{@{}l >{\centering\arraybackslash}m{13cm} c@{}}
\toprule
\rowcolor{blush}\hd{Averaged Quantity} & \multicolumn{1}{c}{\hd{Expression}} & \hd{Proof} \\
\arrayrulecolor{rose}\specialrule{0.9pt}{0pt}{0pt}\arrayrulecolor{rosedeep}
\addlinespace[0.6em]
MI ($\beta \to \infty$) &
$\displaystyle
\begin{aligned}
2\sum_{q >0} \Big\{&
\sqrt{D_q+E_q}
\Bigl(
F  [\theta_q'', k_q']
- E [\theta_q'', k_q']
\Bigr) + 1
+ \ln \Bigl(
\frac{\sqrt{D_q+E_q-1}+\sqrt{D_q-1}}{4}
\Bigr)
\Bigr\}
\end{aligned}
$ &
\pf{App.~\ref{sec: MI formula for general protocol at finite temperature}} \\
\addlinespace[0.8em]
MI ($\beta \geq 0$) &
$\displaystyle
4\sum_{q >0} \left\{
C_q \sqrt{D_q+E_q}
\Bigl(
F \left[\theta_q', k_q'\right]
- E\left[\theta_q', k_q'\right]
\Bigr)
+ \frac{1}{2} \left[ 1 + \ln \left( Z_q'\right) \right]
- \frac{1}{2}C_q \ln\left(\frac{2C_q+1}{2C_q-1}\right)
\right\}
$ &
\pf{App.~\ref{sec: MI formula for general protocol at finite temperature}} \\
\addlinespace[0.8em]
LN ($\beta  \geq 0$) &
$\displaystyle
\frac{2}{\pi}  \sum_{q>0}
\left\{
\int_{0}^{\pi/2}
 \max \left[ 0,
{\rm arcsinh}\left(\sqrt{D_q-1+E_q\sin^2 u}\right)
-\ln\coth\left(\frac{\beta\hbar v q}{2}\right) \right] \, du
\right\}
$ &
\pf{App.~\ref{sec: exact_expression_LN_finite_temp}} \\
\addlinespace[0.8em]
$S_2$ ($\beta  \geq 0$) &
$\displaystyle
\sum_{q>0} \left\{
\ln (2C_q) + \ln\left[\frac{ \sqrt{D_q}+\sqrt{D_q+E_q} }{2}\right]
\right\}
$ &
\pf{App.~\ref{expression_Renyi2_arbitrary}} \\
\addlinespace[0.5em]
\bottomrule
\end{tabular}}
\par\vspace{0.7em}
%
{\renewcommand{\arraystretch}{1.4}\arrayrulecolor{rosedeep}
\begin{tabular}{@{}c@{\hspace{3.5em}}c@{}}
\toprule
\rowcolor{blush}\hd{Primary variables} & \hd{Auxiliary variables} \\
\arrayrulecolor{rose}\specialrule{0.9pt}{0pt}{0pt}\arrayrulecolor{rosedeep}
\addlinespace[0.6em]
$\begin{aligned}
C_q(\beta) &= \frac{1}{2}\coth\!\left(\frac{\beta\hbar\Omega_{q0}}{2}\right) \\[0.9em]
\zeta_q(J_f) &= \frac{vq}{\sqrt{v^2q^2+4\omega J_f}} \\[0.9em]
M_q[\gamma_q(t_{f}),\dot{\gamma_q}(t_{f}),J_f] &= \frac{1}{2}\left[\gamma_q^2(t_{f})+\frac{\dot\gamma_q^2(t_{f})\gamma_q^2(t_{f})+v^2q^2}{\gamma_q^2(t_{f})(v^2q^2+4\omega J_f)}\right] \\[0.9em]
D_q(M_q,\zeta_q) &= \frac{\left(M_q+\sqrt{M_q^2-\zeta_q^2}+1\right)^2}{4\left(M_q+\sqrt{M_q^2-\zeta_q^2}\right)} \\[0.9em]
E_q(M_q,\zeta_q) &= \frac{\sqrt{M_q^2-\zeta_q^2}}{2}\left(\frac{1}{\zeta_q^2}-1\right)
\end{aligned}$
&
$\begin{aligned}
k_q'(D_q,E_q) &= \sqrt{\frac{D_q}{D_q+E_q}} \\[0.9em]
\theta_q'(C_q,D_q) &= \arcsin\!\left(\frac{1}{2C_q\sqrt{D_q}}\right) \\[0.9em]
\theta_q''(D_q) &= \arcsin\!\left(\frac{1}{\sqrt{D_q}}\right) \\[0.9em]
Z_q^\prime(C_q, D_q, E_q) &= \frac{\sqrt{4C^{2}_{q}(D_{q}+E_{q})-1}+\sqrt{4C^{2}_{q}D_{q}-1}}{2\sqrt{4C^{2}_{q}-1}}
\end{aligned}$ \\
\addlinespace[0.5em]
\bottomrule
\end{tabular}}
\caption{Long-time averaged MI and LN after a general protocol at inverse temperature $\beta$, with the $\beta \to \infty $ results included as limiting cases. The incomplete elliptic integrals of the first and second kinds are denoted by $F[\theta,k]$ and $E[\theta,k]$. 
}
\label{Table_general_quench}
\end{table*}
\begin{figure*}[!htbp]
    \centering
    \includegraphics[width=0.86\linewidth]{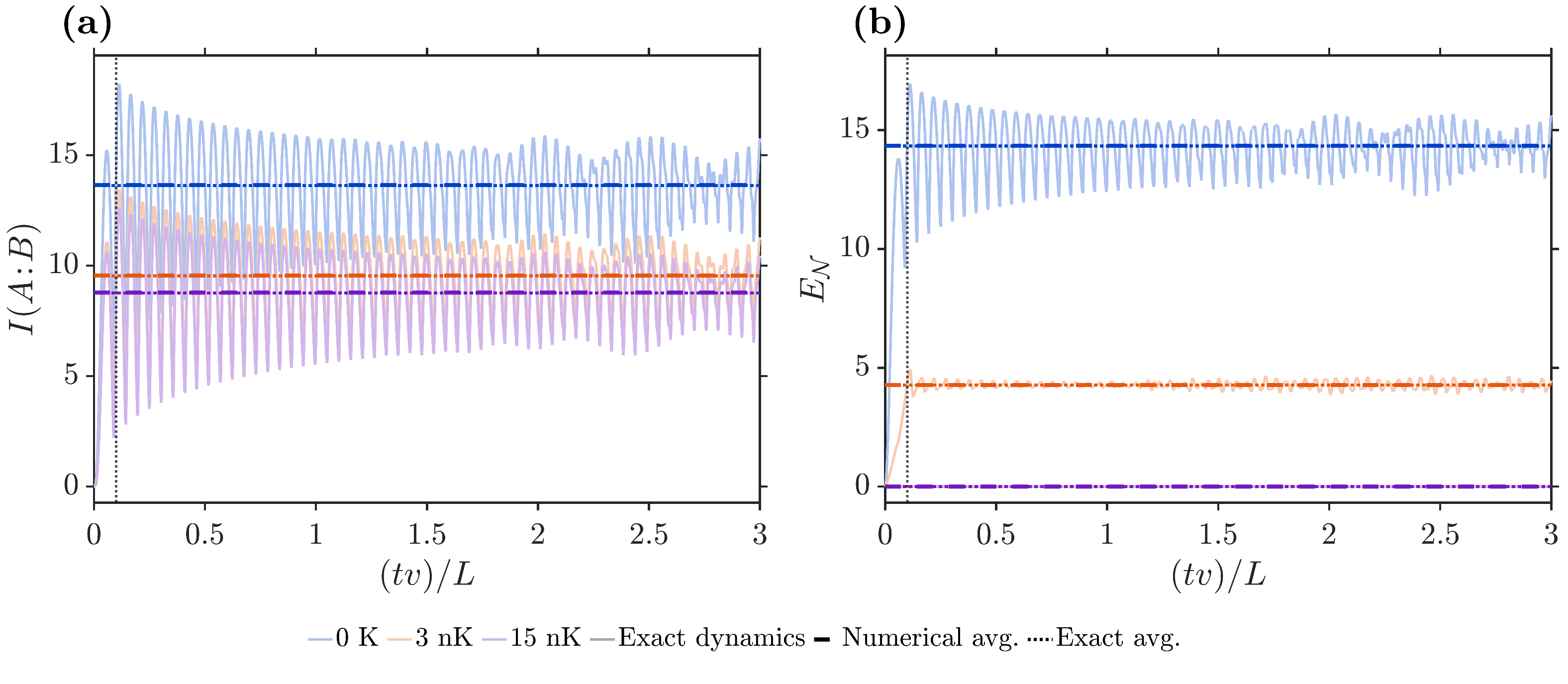}
    \caption{
    {\bf Time-evolution for the MI and the LN for a linear protocol.} We consider the linear protocol of duration $t_f=0.1 L/v$, described in Eq.~\eqref{eq:linear-protocol}. The dashed vertical line marks the time at which the protocol enters the saturation regime. Panel (a) displays the MI and panel (b) displays the LN for the linear protocol at zero (plain blue line) and finite temperatures (plain orange and violet lines). The exact time-averages for the MI and LN at zero and finite temperatures are displayed in colored dashed lines, and they correspond to the expressions given in Table \ref{Table_general_quench}. They are compared to the numerical time-averaging over the interval $t_f$ to $10 t_f$ in dotted lines, which agree well with our formula. We use the same parameters as in Fig.~\ref{figcomparison_mutual_logneg_sudden_quench}, except $L=600\xi$.}
\label{fig:arbitrary_quench_protocol}
\end{figure*}

We perform a numerical verification by picking a linear ramp protocol
\begin{align}
J(t) &=
\begin{cases}
0                 &\quad t \le 0\\
J_f\dfrac{t}{t_f} &\quad 0 < t < t_f\\
J_f               &\quad t \ge t_f
\end{cases}
\label{eq:linear-protocol}
\end{align}
for which analytical expressions of $\gamma_{q}(t_{})$ and $\dot{\gamma}_{q}(t_{})$ are available \cite{dupays2024exact}.
The resulting entanglement dynamics are illustrated in Fig.~\ref{fig:arbitrary_quench_protocol}.
The figure compares the long-time averages of the MI and LN obtained from the numerical time averaging with the exact analytical predictions of Table~\ref{Table_general_quench}, evaluated using $\gamma_{q}(t_{f})$ and $\dot{\gamma}_{q}(t_{f})$.
To neglect the early-time transient contribution, the numerical averaging is performed from $t_{f}$ to $10\,t_{f}$. 
 We find good agreement at both zero and finite temperatures.

Furthermore, similar to the sudden-quench case, the long-time averages of MI and LN decrease with increasing temperature.
This behavior is not restricted to the linear ramp protocol: for an arbitrary ramp protocol, the long-time-averaged MI and LN are monotonically decreasing functions of temperature.
We prove the above for MI in Appendix~\ref{subsec: MI decays monotonically}, while the case for averaged LN is observed immediately from Eq.~\eqref{LN_sum_q} and Eq.~\eqref{eq:symp_eigvals_partial_transpose}. The physical interpretation of the monotonic decay of the averaged MI for the arbitrary protocol is similar to the sudden quench case, i.e. the competition between the thermal growth of the entropies of the reduced system and the full system. 

\subsection{Threshold temperature for entanglement}\label{subsec:threshold_temperature}

Using Eq. \eqref{exact_symplectic_2}, one obtains the long-time averaged LN 
{\small
\begin{align}
    \overline{E_{\mathcal{N}}(\beta)} &= \frac{2}{\pi} \sum_{q >0}^{\Lambda}  \int_0^{\pi/2} \max \Big\{0, {\rm arcsinh}\left(\sqrt{D_{q}-1+E_q\sin^2 u} \right) \nonumber\\
    &-\ln[2C_{q}(\beta)] \Big\} \, du \geq 0\ ,\label{log_neg_general_quench}
\end{align}}
where the inequality holds by definition, and is saturated only if $\sqrt{D_q-1+E_q\sin^2 u}\geq  \csch\left(\beta\hbar\Omega_{q0}\right)$.
When $\overline{E_{\mathcal{N}}(\beta)}= 0$, this means that LN is zero at all times.
We define the threshold temperature $\mathcal{T}_{*}$ for entanglement by introducing a momentum cutoff $\Lambda$, and take the supremum over all momentum modes contributing to entanglement, i.e.
\begin{align}
    \mathcal{T}_{*} &\equiv 
   \sup_{0<q\leq\Lambda} \frac{\hbar v q}{k_B \operatorname{arcsinh}\left((D_q-1+E_q)^{-1/2}\right)}.
\label{eq: true_critical_temp}
\end{align} 
In general, the threshold temperature depends on the full protocol through $D_q$ and $E_q$.

We focus on the sudden-quench, where the threshold temperature admits a simpler analytical form.
To recover the sudden-quench limit, we set $\gamma_{q}(t_f=0)=1$ and $\dot{\gamma}_{q}(t_f=0)=0$ and hence $D_q=1$ and $E_q=A_q$, reducing the threshold temperature to (Appendix~\ref{sec: threshold temperature})
\begin{align}
\mathcal{T}_{*} = 
   \sup_{0<q\leq\Lambda} \frac{\hbar v q}{2k_B \operatorname{arcsinh}\left(vq/\sqrt{4\omega J_f}\right)}, \label{eq: threshold temperature for sudden quench}
\end{align}
which is monotonically increasing with $q$, so that the $\sup$ is obtained for $q=\Lambda$. We numerically verify Eq.~\eqref{eq: threshold temperature for sudden quench} for a sudden quench, where the momentum cutoff $\Lambda=\xi^{-1}$ is taken to be the inverse of the healing length. 
Fig.~\ref{fig:critical_temperature} displays the time-averaged LN as a function of temperature. We clearly observe a monotonic decay of the entanglement with temperature, and obtain an excellent match for the threshold temperature $\mathcal{T}_{*}=16.55\; {\rm nK}$ obtained from Eq. \eqref{eq: threshold temperature for sudden quench} with the temperature for which $\overline{E_{\mathcal{N}}(t,\beta_{*})}=0$ as obtained from numerical data. 
Notwithstanding, this threshold should be interpreted with care: in the absence of a momentum cutoff, or for a sufficiently large cutoff, the entanglement does not truly vanish; rather, the residual entanglement is effectively undetectable.
\begin{figure}[!htbp]
    \centering
    \includegraphics[width=1\linewidth]{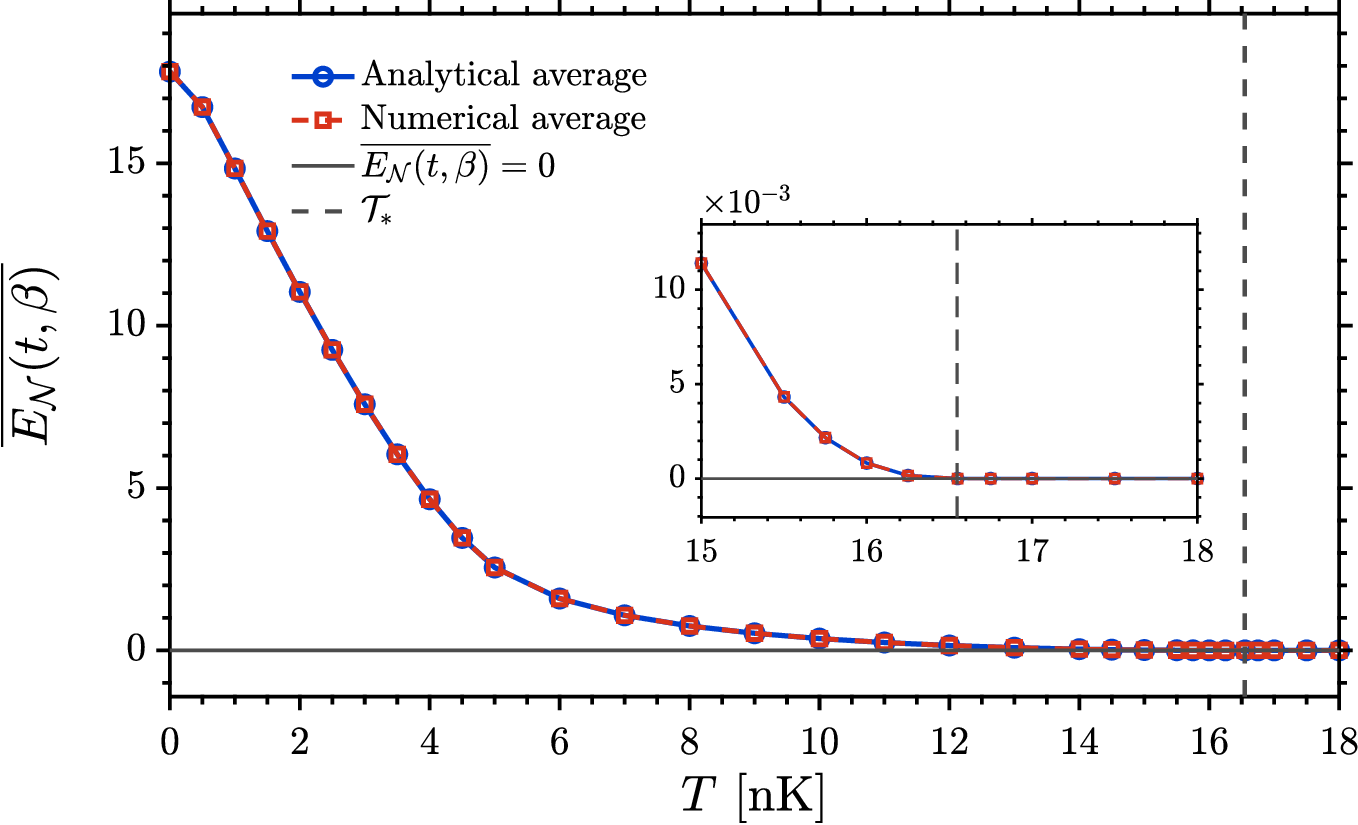}
    \caption{{\bf Time-averaged logarithmic negativity for the sudden quench as a function of the temperature}. The blue line indicates the time-averaged LN, as given by the third and fourth equations in Table~\ref{Results_sudden_quench_table}. The red dashed line is the numerically time-averaged LN. 
     The vertical dashed line corresponds to the threshold temperature predicted by Eq.~\eqref{eq: threshold temperature for sudden quench}. Plotting parameters are identical to those in Fig.~\ref{figcomparison_mutual_logneg_sudden_quench}. The inset magnifies the range from $15\;{\rm nK}$ to $18\;{\rm nK}$, close to the threshold temperature $\mathcal{T}^* = 16.55\; \rm nK$.
}
    \label{fig:critical_temperature}
\end{figure}

\subsection{Long-time average scalings in the thermodynamic limit\label{Sec: TDL scaling}}
{
In Secs.~\ref{subsec: sudden quench long time average} and
\ref{subsec:ln_dynamics}, we derive analytical expressions for the
long-time averaged MI, LN, and Rényi entropies after arbitrary
protocols reaching the saturation regime. The thermodynamic limit is taken by replacing the momentum sum with an integral,
$\sum_{q>0}\to (L/\pi)\int dq$. 
Here, we focus
on the sudden quench and show that in the zero temperature limit averaged information-theoretic quantities of interest scale as $L\sqrt{4\omega J_f}/v$.

For simplicity, we focus on the MI after a sudden quench (explicit expression in Table~\ref{Results_sudden_quench_table}). 
In the TDL, the sum over momenta becomes an integral, and the integrand depends on $q$ only through the oscillation amplitude $A_{q}$ and the temperature factor $C_q(\beta)$, given by  Eqs.~\eqref{A_q_parameter} and ~\eqref{eq:C_q} respectively. Motivated by the form of $A_q$, we introduce the momentum scale $q_J=\sqrt{4\omega J_f}/v$ and the associated dimensionless momentum $x=q/q_J$. 
At finite temperatures, after this change 
of variables, the long-time averaged MI takes the form
\begin{equation}
    \overline{I(A:B,\beta)}
=
\frac{L\sqrt{4\omega J_f}}{v}
\mathcal I\left(\beta\hbar\sqrt{4\omega J_f}\right),
\end{equation}
where $\mathcal I$ is a dimensionless quantity. In the zero-temperature limit (cf. Appendix~\ref{scaling_arbitrary_protocol}), this expression reduces to 
\begin{align}
{\rm lim}_{\beta\to \infty}\overline{I(A:B,\beta)}&=\frac{L\sqrt{4\omega J_{f}}}{2\pi v}(1.201). \nonumber
\end{align}
The same argument applies to LN and R{\'e}nyi entropies
\begin{align}
\lim_{\beta\rightarrow \infty}\overline{E_{\mathcal{N}}(\beta)} &= \frac{2L\sqrt{4\omega J_f}}{\pi^2v}\int_0^{\infty} d\theta\;  \text{Ti}_2 \left[\operatorname{csch}(2\theta) \right]\cosh \theta , \nonumber\\
\lim_{\beta\rightarrow \infty}\overline{S_{2}^A(\beta)}&=\frac{L}{2\pi}\frac{\sqrt{4\omega J_{f}}}{v}(4-\pi), \nonumber
\end{align}
which follows from evaluating the integrand in Table~\ref{Results_sudden_quench_table}.
The zero-temperature scaling above relies on the fact that, for the sudden quench, the only momentum scale entering the integrand is $q_J$. This scaling argument fails at finite temperature and for generic finite-time protocols, see Appendix~\ref{scaling_arbitrary_protocol}.

We finally provide a heuristic interpretation of the
zero-temperature sudden-quench scaling. As discussed in Sec.~\ref{sec: Entanglement_dynamics_study},
field-space correlations arise from the different evolution of the
symmetric and antisymmetric sectors. After a sudden quench, their
dispersions differ appreciably for
$q\lesssim \sqrt{4\omega J_f}/v \equiv q_J$, when the Luttinger liquid dispersion and mass term are comparable, whereas for larger momenta the mass
term becomes progressively negligible and the two sectors evolve
similarly. 
This suggests that the dominant contribution comes from modes with
$q\lesssim q_J$.
If each relevant mode contributes an amount of order unity to $S_A$, their number scales as $Lq_J$, yielding the heuristic estimate
\begin{equation}
S_A\sim Lq_J
=
L\frac{\sqrt{4\omega J_f}}{v}.
\end{equation}
The same mode-counting argument accounts for the common $L \, q_J$ scaling of MI, LN, and Rényi entropies, although their numerical prefactors remain measure dependent.
}

\section{Convergence to the adiabatic solution}
\label{sec: adiabatic smooth_quench_protocol}

We turn to the adiabatic limit, where the ramp duration is much longer than intrinsic timescales of the system.
A complete treatment requires full sine-Gordon dynamics, since the system spends substantial time at intermediate tunneling strengths where the quadratic approximation breaks down.
For simplicity, we study the adiabatic limit within the Gaussian approximation to provide a first benchmark. 

For long tunneling protocol times, we can use the adiabatic solution to the Ermakov Eq.~\eqref{Ermakov}, 
which reduces to the scaling 
\begin{align}
\gamma^{\rm ad}_{q}(t)&=\sqrt{\frac{\Omega_{q0}}{\Omega_q(t)}}. \label{adiabatic_solution_ermakov}
\end{align}
This solution comes from the vanishing of the second derivative at large protocol times $\ddot{\gamma}_{q}(t) \approx 0$, that sets a condition on the first and second derivatives of the time-dependent oscillator frequency, detailed in Appendix~\ref{adiabatic_approx_details}. 
The adiabatic condition can only be satisfied for finite size systems, to avoid divergence at $t\to 0$ and $q\to 0$.
If the adiabatic condition is satisfied for all modes, the symplectic eigenvalues of the reduced covariance matrix can be approximated as 
\begin{align}
 \lambda^{\rm ad}_{q}(t, \beta)\approx \frac{C_q(\beta)}{2}\left(\gamma^{\rm ad}_{q}(t)+\frac{1}{\gamma^{\rm ad}_{q}(t)}\right), \label{eq: thermal state symplectic eigenvalues}
\end{align}
where we have approximated $\dot{\gamma}_q(t) \approx 0$.
Similarly, we have
\begin{align}
\nu_{q}(t,\beta)&\approx C_q(\beta)\gamma_q^{\rm ad}(t).\label{eq: adiabatic_symplectic_partial}
\end{align}
{In many adiabatic settings, the initial state is set to be the ground state, which is easily recovered from our solution by substituting $C_q \rightarrow 1/2$.} Eqs.~\eqref{eq: thermal state symplectic eigenvalues} and~\eqref{eq: adiabatic_symplectic_partial} correspond to an adiabatically evolved thermal state: initial populations are preserved, while each energy eigenstate follows the instantaneous energy eigenbasis of the time-dependent Hamiltonian (see Appendix~\ref{adiabatic_approx_details}).
Note that this is not the intantaneous thermal state, which would have time-dependent populations.

To illustrate convergence to the adiabatic solution, we plot the MI and LN for a polynomial ramp and compare with the adiabatic solution. Fig.~\ref{smooth_quench} (a) displays the ramp protocol for different protocol durations. Fig.~\ref{smooth_quench} (b) shows MI and a comparison to the adiabatic solution, while Fig.~\ref{smooth_quench} (c) shows the same for LN. Both exhibit rapid oscillations characteristic of nonadiabatic behavior. Nevertheless, for very slow protocols, we expect the solution to converge to a smooth adiabatic solution.
Indeed, as ramp duration increases, the MI and LN both approach their corresponding adiabatic evolutions, as can be seen from the overlap between the dashed and plain lines in Figs.~\ref{smooth_quench} (b) -- (c). For the chosen tunneling protocol and parameters, we numerically find that convergence to the adiabatic solution occurs around $t_f \sim 10L/v$ for MI and $t_f \sim L/v$ for LN. These timescales can be understood from the slowest mode: since the oscillation period of this mode without the tunneling is of order $L/v$, adiabatic behaviour is expected only when the protocol duration is longer than this timescale. 

\section{Discussion and Conclusion}
\label{conclusion}
\begin{figure*}[!htbp]
    \centering
\includegraphics[width=1\linewidth]{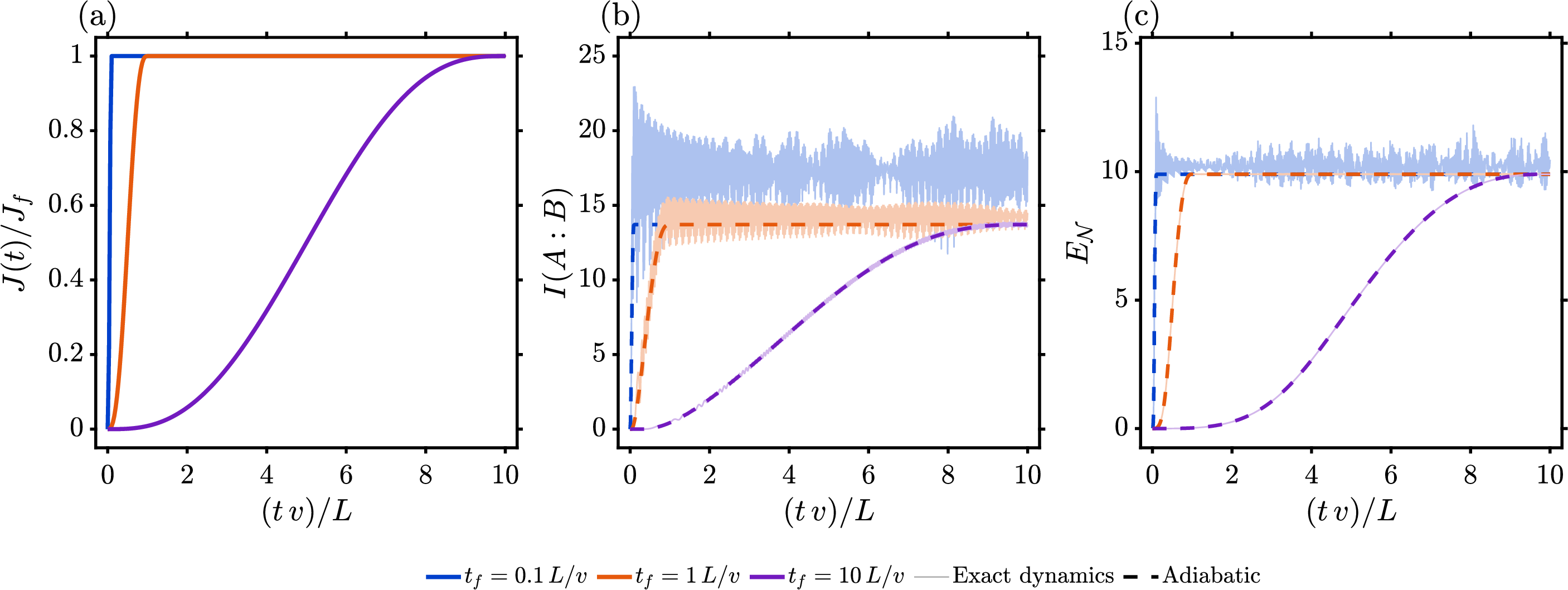}
    \caption{{\bf (a) Smooth polynomial protocol for different protocol times. (b) Mutual information and (c) logarithmic negativity for the smooth protocol in (a) with different ramp durations, compared with the adiabatic limit.} For illustration, we consider a tunneling protocol given by a fifth-order polynomial,
$J(t)=J_{f}\mathcal{P}(t/t_{f})$, with
$\mathcal{P}(s)=6s^{5}-15s^{4}+10s^{3}$, which satisfies the boundary conditions
$J(0)=0$, $J(t_{f})=J_{f}$, and
$\dot{J}(0)=\ddot{J}(0)=\dot{J}(t_{f})=\ddot{J}(t_{f})=0$.
The different colours correspond to different protocol durations for the smooth protocol.
In panels (b) and (c), the thin lines show the MI and LN respectively, with its oscillatory behaviour, while the thick coloured lines correspond to the adiabatic solution. We observe convergence to the adiabatic solution at large ramp time for both the MI and LN. We use the same plotting parameters as in Fig.~\ref{figcomparison_mutual_logneg_sudden_quench}, except that $L=1000\chi$ and we set the initial temperature to $T=2\; {\rm nK}$.}
    \label{smooth_quench}
\end{figure*}
In this work, we have shown that field-space entanglement dynamics in tunnel-coupled Luttinger liquids admits a compact analytical description in the Gaussian regime, even under arbitrary time-dependent tunneling protocols. For zero- or finite-temperature initial states, we derived explicit expressions for the mutual information, R{\'e}nyi entropies, and logarithmic negativity in terms of the solution to the mode-resolved Ermakov equation. The Ermakov factors act as a compressed description of the dynamics: once known, they determine the time evolution of all the relevant information-theoretic quantities, and provide a clear conceptual route to studying the behavior of field-space entanglement.

The Gaussian dynamics exhibits distinctive early-time scaling, governed by the first non-vanishing derivative of the tunneling ramp at initial time. For protocols satisfying $J(0)=0$, the early-time MI first grows as $I(A:B;t)\propto t^{2n+2}$, up to logarithmic corrections in the transient regime, before crossing over to the thermodynamic-limit scaling $I(A:B;t)\propto t^{n+1}$, with the sudden quench included as the $n =0$ case. While the behaviour of the transient perturbative regime is expected to depend on the first non-vanishing derivative of $J(t)$, we see that this information alone continues to determine the macroscopic early-time scaling in the thermodynamic limit. The microscopic parameters $K$ and $v$ enter only through the overall scale, via $\omega\propto K/v$. This contrasts with spatial entanglement in Luttinger liquids, where leading scaling laws are typically controlled by universal quantities such as the central charge, or, after a quench, by quasiparticle dynamics.

A similar simplification occurs in the long-time averaged dynamics. Although the information-theoretic quantities generally exhibit oscillations after the ramp, their time averages are fully determined by the final tunneling strength $J_f$ and the mode-dependent quantities $\{\gamma_q(t_f),\dot{\gamma}_q(t_f)\}$. 
Different protocols that share the same values at $t=t_f$ therefore become dynamically equivalent at the level of the long-time averages, providing another example of this reduction in complexity.
For a sudden quench at zero temperature, we further show that the averaged MI, LN, and Rényi entropies scale as $L\sqrt{4\omega J_f}/v$. 

Finally, we consider the adiabatic regime, where the ramp duration greatly exceeds the intrinsic timescales of the system.
In this regime, the Ermakov equation yields a simple analytical solution and provides a clear physical picture of the evolution of those information-theoretic measures. However, the Gaussian approximation may no longer hold in the adiabatic regime, and whether the resulting physical predictions remain accurate requires further investigation.

Our analysis relies on the harmonic approximation to the time-dependent sine-Gordon model, which makes correlation dynamics in this setting analytically tractable.
The validity of the harmonic approximation is clear in equilibrium; one picks a strong tunneling strength $J$ so that the relative phase fluctuation is suppressed around a local minima of the cosine potential. 
Away from equilibrium, it is expected to be accurate in situations where non-Gaussian correlations have limited time to develop such as in sudden quenches and sufficiently rapid smooth ramps into the strongly coupled regime. A direct confirmation requires treating the full time-dependent sine-Gordon dynamics, using methods such as Hamiltonian truncation~\cite{Horvath2019}, self-consistent time-dependent harmonic approximation~\cite{Van_nieuwkerk2019,Van_Nieuwkerk2020}, and truncated Wigner approximation~\cite{agarwal2023caustics, bayoboc2022dynamics}.
These methods have been applied mainly to time-dependent correlation functions, but await extension to information-theoretic quantities.
Our work provides a useful benchmark to further examine the extent of non-Gaussian effects on field-space entanglement dynamics.

Besides understanding the influence of non-Gaussianity, other natural extensions follow from our work. For instance, we focused on symmetric Luttinger liquids with $v_A = v_B = v$ and $K_A = K_B = K$. When this symmetry is broken, the symmetric and antisymmetric sectors no longer decouple, although the resulting dynamics remains analytically tractable through the generalized Bogoliubov transformations~\cite{ruggiero2021large}. We have also focused on tunneling coupling, which produces a gapped antisymmetric sector while leaving the symmetric sector gapless. Other interactions can instead generate gapped-gapped or gapless-gapless low-energy structures, whose ground-state entanglement was studied in Ref.~\cite{Lundgren2013}. Extending these settings to nonequilibrium protocols and finite-temperature initial states would reveal which features identified here are specific to tunnel coupling and which persist more broadly across coupled-field dynamics. Furthermore, one can consider other scenarios involving multipartite entanglement, e.g. field-space entanglement between three or more fields \cite{mozaffar2016on}. A minimal setting is a system of three coupled Luttinger liquids, where one could test whether pairwise correlations obey, violate, or generalize entanglement-sharing constraints analogous to monogamy relations \cite{Coffman2000,Osborne2006}. 

Taken together, our results establish field-space entanglement as a distinctive probe for nonequilibrium dynamics, where interaction geometry, driving protocol, and thermal fluctuations leave directly identifiable signatures. Extending this analysis beyond Gaussian dynamics and symmetric couplings will reveal which of these structures are truly universal, and which are specific to tunnel-coupled Luttinger liquids.

\begin{acknowledgments}
    It is a pleasure to thank Thierry Giamarchi, Patrizia Vignolo, Stefan Aimet and Sam Carr for interesting comments. P.R. and L.D. acknowledge support from the UK Engineering and Physical Sciences Research Council (EPSRC) through a New Investigator Award, grant number EP/Y015363/1. The NTU team is supported by the National Research Foundation, Singapore through the National Quantum Office, hosted in A*STAR, under its Centre for Quantum Technologies Funding Initiative
(S24Q2d0009), Singapore; and the Ministry of Education, Singapore, under its “Probing Entanglement and Thermalization Dynamics in Quantum Field Simulators” (Tier 2, Project No.: MOE-T2EP50225-0022).
\end{acknowledgments}

\bibliography{REF_manuscript}

@article{Horvath2019,
  title = {{Nonequilibrium time evolution and rephasing in the quantum sine-Gordon model}},
  author = {Horv\'ath, D. X. and Lovas, I. and Kormos, M. and Tak\'acs, G. and Zar\'and, G.},
  journal = {Phys. Rev. A},
  volume = {100},
  issue = {1},
  pages = {013613},
  numpages = {21},
  year = {2019},
  month = {Jul},
  publisher = {American Physical Society},
  doi = {10.1103/PhysRevA.100.013613},
  url = {https://link.aps.org/doi/10.1103/PhysRevA.100.013613}
}

@article{Van_Nieuwkerk2020,
	title = {{On the low-energy description for tunnel-coupled one-dimensional Bose gases}},
	pages = {025},
	author = {van Nieuwkerk, Yuri Daniel and Essler, Fabian H. L.},
	journal = {SciPost Phys.},
	volume = {9},
	year = {2020},
	publisher = {SciPost},
	doi = {10.21468/SciPostPhys.9.2.025},
	url = {https://scipost.org/10.21468/SciPostPhys.9.2.025}
}

@article{Van_nieuwkerk2019,
doi = {10.1088/1742-5468/ab3579},
url = {https://doi.org/10.1088/1742-5468/ab3579},
year = {2019},
month = {aug},
publisher = {IOP Publishing and SISSA},
volume = {2019},
number = {8},
pages = {084012},
author = {van Nieuwkerk, Yuri D and Essler, Fabian H L},
title = {Self-consistent time-dependent harmonic approximation for the sine-Gordon model out of equilibrium},
journal = {J. Stat. Mech.: Theory Exp.}
}

@article{Peschel2009,
doi = {10.1088/1751-8113/42/50/504003},
url = {https://doi.org/10.1088/1751-8113/42/50/504003},
year = {2009},
month = {dec},
publisher = {},
volume = {42},
number = {50},
pages = {504003},
author = {Peschel, Ingo and Eisler, Viktor},
title = {Reduced density matrices and entanglement entropy in free lattice models},
journal = {J. Phys. A: Math. Theor.}
}

@misc{Yang2026,
      title={Probing Entanglement and Symmetries in Random States Using a Superconducting Quantum Processor}, 
      author={Jia-Nan Yang and Lata Kh Joshi and Filiberto Ares and Yihang Han and Pengfei Zhang and Pasquale Calabrese},
      year={2026},
      eprint={2601.22224},
      archivePrefix={arXiv},
      primaryClass={quant-ph},
      url={https://arxiv.org/abs/2601.22224}, 
}

@Article{Islam2015,
author={Islam, Rajibul
and Ma, Ruichao
and Preiss, Philipp M.
and Eric Tai, M.
and Lukin, Alexander
and Rispoli, Matthew
and Greiner, Markus},
title={Measuring entanglement entropy in a quantum many-body system},
journal={Nature},
year={2015},
month={Dec},
day={01},
volume={528},
number={7580},
pages={77-83},
issn={1476-4687},
doi={10.1038/nature15750},
url={https://doi.org/10.1038/nature15750}
}

@misc{Matsoukasroubeas2026,
      title={Randomised measurements of a disorder-induced entanglement transition in a neutral atom quantum processor}, 
      author={Apollonas S. Matsoukas-Roubeas and Oscar Scholin and Lucas Sá and Arinjoy De and Majd Hamdan and Alexei Bylinskii and Andrew J. Daley and Dorian A. Gangloff},
      year={2026},
      eprint={2604.24854},
      archivePrefix={arXiv},
      primaryClass={quant-ph},
      url={https://arxiv.org/abs/2604.24854}, 
}

@article{Camilo2019,
  title = {{Strong subadditivity of the R\'enyi entropies for bosonic and fermionic Gaussian states}},
  author = {Camilo, Giancarlo and Landi, Gabriel T. and Eli\"ens, Sebas},
  journal = {Phys. Rev. B},
  volume = {99},
  issue = {4},
  pages = {045155},
  numpages = {8},
  year = {2019},
  month = {Jan},
  publisher = {American Physical Society},
  doi = {10.1103/PhysRevB.99.045155},
  url = {https://link.aps.org/doi/10.1103/PhysRevB.99.045155}
}

@article{Furukawa2011,
  title = {{Entanglement entropy between two coupled Tomonaga-Luttinger liquids}},
  author = {Furukawa, Shunsuke and Kim, Yong Baek},
  journal = {Phys. Rev. B},
  volume = {83},
  issue = {8},
  pages = {085112},
  numpages = {17},
  year = {2011},
  month = {Feb},
  publisher = {American Physical Society},
  doi = {10.1103/PhysRevB.83.085112},
  url = {https://link.aps.org/doi/10.1103/PhysRevB.83.085112}
}

@misc{Mathe2026,
      title={Thermal Entanglement and Out-of-Equilibrium Thermodynamics in 1D Bose gases}, 
      author={Julia Mathé and Nicky Kai Hong Li and Pharnam Bakhshinezhad and Giuseppe Vitagliano},
      year={2026},
      eprint={2604.01157},
      archivePrefix={arXiv},
      primaryClass={quant-ph},
      url={https://arxiv.org/abs/2604.01157}, 
}

@article{Coffman2000,
  title = {Distributed entanglement},
  author = {Coffman, Valerie and Kundu, Joydip and Wootters, William K.},
  journal = {Phys. Rev. A},
  volume = {61},
  issue = {5},
  pages = {052306},
  numpages = {5},
  year = {2000},
  month = {Apr},
  publisher = {American Physical Society},
  doi = {10.1103/PhysRevA.61.052306},
  url = {https://link.aps.org/doi/10.1103/PhysRevA.61.052306}
}

@article{Osborne2006,
  title = {General Monogamy Inequality for Bipartite Qubit Entanglement},
  author = {Osborne, Tobias J. and Verstraete, Frank},
  journal = {Phys. Rev. Lett.},
  volume = {96},
  issue = {22},
  pages = {220503},
  numpages = {4},
  year = {2006},
  month = {Jun},
  publisher = {American Physical Society},
  doi = {10.1103/PhysRevLett.96.220503},
  url = {https://link.aps.org/doi/10.1103/PhysRevLett.96.220503}
}

@article{Calabrese2005,
doi = {10.1088/1742-5468/2005/04/P04010},
url = {https://doi.org/10.1088/1742-5468/2005/04/P04010},
year = {2005},
month = {apr},
publisher = {},
volume = {2005},
number = {04},
pages = {P04010},
author = {Calabrese, Pasquale and Cardy, John},
title = {Evolution of entanglement entropy in one-dimensional systems},
journal = {J. Stat. Mech.: Theory Exp.}
}

@article{Betz2011,
  title = {Two-Point Phase Correlations of a One-Dimensional Bosonic Josephson Junction},
  author = {Betz, T. and Manz, S. and B\"ucker, R. and Berrada, T. and Koller, Ch. and Kazakov, G. and Mazets, I. E. and Stimming, H.-P. and Perrin, A. and Schumm, T. and Schmiedmayer, J.},
  journal = {Phys. Rev. Lett.},
  volume = {106},
  issue = {2},
  pages = {020407},
  numpages = {4},
  year = {2011},
  month = {Jan},
  publisher = {American Physical Society},
  doi = {10.1103/PhysRevLett.106.020407},
  url = {https://link.aps.org/doi/10.1103/PhysRevLett.106.020407}
}

@Article{Hofferberth2007,
author={Hofferberth, S.
and Lesanovsky, I.
and Fischer, B.
and Schumm, T.
and Schmiedmayer, J.},
title={Non-equilibrium coherence dynamics in one-dimensional Bose gases},
journal={Nature},
year={2007},
month={Sep},
day={01},
volume={449},
number={7160},
pages={324-327},
issn={1476-4687},
doi={10.1038/nature06149},
url={https://doi.org/10.1038/nature06149}
}

@article{Lauchli2012,
  title = {Entanglement spectra of coupled $S=\frac{1}{2}$ spin chains in a ladder geometry},
  author = {L\"auchli, Andreas M. and Schliemann, John},
  journal = {Phys. Rev. B},
  volume = {85},
  issue = {5},
  pages = {054403},
  numpages = {6},
  year = {2012},
  month = {Feb},
  publisher = {American Physical Society},
  doi = {10.1103/PhysRevB.85.054403},
  url = {https://link.aps.org/doi/10.1103/PhysRevB.85.054403}
}

@article{Cazalilla2004,
doi = {10.1088/0953-4075/37/7/051},
url = {https://doi.org/10.1088/0953-4075/37/7/051},
year = {2004},
month = {mar},
publisher = {},
volume = {37},
number = {7},
pages = {S1},
author = {M A Cazalilla},
title = {Bosonizing one-dimensional cold atomic gases},
journal = {J. Phys. B: At. Mol. Opt. Phys.},
}

@article{Poilblanc2010,
  title = {{Entanglement Spectra of Quantum Heisenberg Ladders}},
  author = {Poilblanc, Didier},
  journal = {Phys. Rev. Lett.},
  volume = {105},
  issue = {7},
  pages = {077202},
  numpages = {4},
  year = {2010},
  month = {Aug},
  publisher = {American Physical Society},
  doi = {10.1103/PhysRevLett.105.077202},
  url = {https://link.aps.org/doi/10.1103/PhysRevLett.105.077202}
}

@article{Lundgren2013,
  title = {{Entanglement spectra between coupled Tomonaga-Luttinger liquids: Applications to ladder systems and topological phases}},
  author = {Lundgren, Rex and Fuji, Yohei and Furukawa, Shunsuke and Oshikawa, Masaki},
  journal = {Phys. Rev. B},
  volume = {88},
  issue = {24},
  pages = {245137},
  numpages = {14},
  year = {2013},
  month = {Dec},
  publisher = {American Physical Society},
  doi = {10.1103/PhysRevB.88.245137},
  url = {https://link.aps.org/doi/10.1103/PhysRevB.88.245137}
}

@article{Peschel_1999,
doi = {10.1088/0305-4470/32/48/305},
url = {https://doi.org/10.1088/0305-4470/32/48/305},
year = {1999},
month = {dec},
publisher = {},
volume = {32},
number = {48},
pages = {8419},
author = {Ingo Peschel and Ming-Chiang Chung},
title = {Density 
matrices for a chain of oscillators},
journal = {Journal of Physics A: Mathematical and General},
}

@article{Chen2013,
doi = {10.1088/1742-5468/2013/08/P08013},
url = {https://doi.org/10.1088/1742-5468/2013/08/P08013},
year = {2013},
month = {aug},
publisher = {IOP Publishing and SISSA},
volume = {2013},
number = {08},
pages = {P08013},
author = {Chen, Xiao and Fradkin, Eduardo},
title = {{Quantum entanglement and thermal reduced density matrices in fermion and spin systems on ladders}},
journal = {J. Stat. Mech.: Theory Exp.}
}

@article{Xu2011,
  title = {{Entanglement entropy of coupled conformal field theories and Fermi liquids}},
  author = {Xu, Cenke},
  journal = {Phys. Rev. B},
  volume = {84},
  issue = {12},
  pages = {125119},
  numpages = {6},
  year = {2011},
  month = {Sep},
  publisher = {American Physical Society},
  doi = {10.1103/PhysRevB.84.125119},
  url = {https://link.aps.org/doi/10.1103/PhysRevB.84.125119}
}

@article{Abanin2019,
  title = {Colloquium: Many-body localization, thermalization, and entanglement},
  author = {Abanin, Dmitry A. and Altman, Ehud and Bloch, Immanuel and Serbyn, Maksym},
  journal = {Rev. Mod. Phys.},
  volume = {91},
  issue = {2},
  pages = {021001},
  numpages = {26},
  year = {2019},
  month = {May},
  publisher = {American Physical Society},
  doi = {10.1103/RevModPhys.91.021001},
  url = {https://link.aps.org/doi/10.1103/RevModPhys.91.021001}
}

@Article{Calabrese2020,
	title={{Entanglement spreading in non-equilibrium integrable systems}},
	author={Pasquale Calabrese},
	journal={SciPost Phys. Lect. Notes},
	pages={20},
	year={2020},
	publisher={SciPost},
	doi={10.21468/SciPostPhysLectNotes.20},
	url={https://scipost.org/10.21468/SciPostPhysLectNotes.20},
}

@book{Wen_book,
  author    = {Wen, X. G.},
  title     = {Quantum Field Theory of Many-Body Systems},
  isbn      = {9780198530947},
  lccn      = {2004301677},
  series    = {Oxford Graduate Texts},
    doi = {10.1093/acprof:oso/9780199227259.001.0001},
    url = {https://doi.org/10.1093/acprof:oso/9780199227259.001.0001},
  year      = {2004},
  publisher = {OUP Oxford}
}

@article{Pinney1950,
  author       = {Pinney, Edmund},
  title        = {The nonlinear differential equation $y'' + p(x)\,y + c\,y^{-3} = 0$},
  journal      = {Proc. Amer. Math. Soc.},
  volume       = {1},
  number       = {5},
  pages        = {581},
  year         = {1950},
  doi          = {10.1090/S0002-9939-1950-0037979-4},
}

@article{Calabrese2004,
doi = {10.1088/1742-5468/2004/06/P06002},
url = {https://doi.org/10.1088/1742-5468/2004/06/P06002},
year = {2004},
month = {jun},
publisher = {},
volume = {2004},
number = {06},
pages = {P06002},
author = {Pasquale Calabrese and John Cardy},
title = {Entanglement entropy and quantum field theory},
journal = {J. Stat. Mech},
}

@article{Plenio2005,
  title = {Logarithmic Negativity: A Full Entanglement Monotone That is not Convex},
  author = {Plenio, M. B.},
  journal = {Phys. Rev. Lett.},
  volume = {95},
  issue = {9},
  pages = {090503},
  numpages = {4},
  year = {2005},
  month = {Aug},
  publisher = {American Physical Society},
  doi = {10.1103/PhysRevLett.95.090503},
  url = {https://link.aps.org/doi/10.1103/PhysRevLett.95.090503}
}

@article{Zyczkowski1998,
  title = {Volume of the set of separable states},
  author = {\ifmmode \dot{Z}\else \.{Z}\fi{}yczkowski, Karol and Horodecki, Pawe\l{} and Sanpera, Anna and Lewenstein, Maciej},
  journal = {Phys. Rev. A},
  volume = {58},
  issue = {2},
  pages = {883--892},
  numpages = {0},
  year = {1998},
  month = {Aug},
  publisher = {American Physical Society},
  doi = {10.1103/PhysRevA.58.883},
  url = {https://link.aps.org/doi/10.1103/PhysRevA.58.883}
}

@article{Peres1996,
  title = {Separability Criterion for Density Matrices},
  author = {Peres, Asher},
  journal = {Phys. Rev. Lett.},
  volume = {77},
  issue = {8},
  pages = {1413--1415},
  numpages = {0},
  year = {1996},
  month = {Aug},
  publisher = {American Physical Society},
  doi = {10.1103/PhysRevLett.77.1413},
  url = {https://link.aps.org/doi/10.1103/PhysRevLett.77.1413}
}

@article{Dupays2021,
  title = {Delta-kick cooling, time-optimal control of scale-invariant dynamics, and shortcuts to adiabaticity assisted by kicks},
  author = {Dupays, L\'eonce and Spierings, David C. and Steinberg, Aephraim M. and del Campo, Adolfo},
  journal = {Phys. Rev. Res.},
  volume = {3},
  issue = {3},
  pages = {033261},
  numpages = {14},
  year = {2021},
  month = {Sep},
  publisher = {American Physical Society},
  doi = {10.1103/PhysRevResearch.3.033261},
  url = {https://link.aps.org/doi/10.1103/PhysRevResearch.3.033261}
}

@article{Gritsev2007,
  title = {Linear response theory for a pair of coupled one-dimensional condensates of interacting atoms},
  author = {Gritsev, Vladimir and Polkovnikov, Anatoli and Demler, Eugene},
  journal = {Phys. Rev. B},
  volume = {75},
  issue = {17},
  pages = {174511},
  numpages = {11},
  year = {2007},
  month = {May},
  publisher = {American Physical Society},
  doi = {10.1103/PhysRevB.75.174511},
  url = {https://link.aps.org/doi/10.1103/PhysRevB.75.174511}
}

@article{Haldane1981,
doi = {10.1088/0022-3719/14/19/010},
url = {https://doi.org/10.1088/0022-3719/14/19/010},
year = {1981},
month = {jul},
publisher = {},
volume = {14},
number = {19},
pages = {2585},
author = {F D M Haldane},
title = {{'Luttinger liquid theory' of one-dimensional quantum fluids. I. Properties of the Luttinger model and their extension to the general 1D interacting spinless Fermi gas}},
journal = {J. Phys. C: Solid State Phys.}
}

@article{
Lukin2019,
author = {Alexander Lukin  and Matthew Rispoli  and Robert Schittko  and M. Eric Tai  and Adam M. Kaufman  and Soonwon Choi  and Vedika Khemani  and Julian Léonard  and Markus Greiner },
title = {Probing entanglement in a many-body–localized system},
journal = {Science},
volume = {364},
number = {6437},
pages = {256-260},
year = {2019},
doi = {10.1126/science.aau0818},
URL = {https://www.science.org/doi/abs/10.1126/science.aau0818}}

@article{balasubramanian2012momentum,
  title = {Momentum-space entanglement and renormalization in quantum field theory},
  author = {Balasubramanian, Vijay and McDermott, Michael B. and Van Raamsdonk, Mark},
  journal = {Phys. Rev. D},
  volume = {86},
  issue = {4},
  pages = {045014},
  numpages = {18},
  year = {2012},
  month = {Aug},
  publisher = {American Physical Society},
  doi = {10.1103/PhysRevD.86.045014},
  url = {https://link.aps.org/doi/10.1103/PhysRevD.86.045014}
}

@article{martins2022momentum,
  title = {{Momentum space entanglement from the Wilsonian effective action}},
  author = {Martins Costa, Matheus H. and van den Brink, Jeroen and Nogueira, Flavio S. and Krein, Gast\~ao I.},
  journal = {Phys. Rev. D},
  volume = {106},
  issue = {6},
  pages = {065024},
  numpages = {18},
  year = {2022},
  month = {Sep},
  publisher = {American Physical Society},
  doi = {10.1103/PhysRevD.106.065024},
  url = {https://link.aps.org/doi/10.1103/PhysRevD.106.065024}
}

@article{mollabashi2014entanglement,
	author = {Mollabashi, Ali and Shiba, Noburo and Takayanagi, Tadashi},
	date = {2014/04/30},
	doi = {10.1007/JHEP04(2014)185},
	id = {Mollabashi2014},
	isbn = {1029-8479},
	journal = {J. High Energ. Phys.},
	number = {4},
	pages = {185},
	title = {{Entanglement between two interacting CFTs and generalized holographic entanglement entropy}},
	url = {https://doi.org/10.1007/JHEP04(2014)185},
	volume = {2014},
	year = {2014}
}

@article{mozaffar2016on,
	author = {Mozaffar, M. Reza Mohammadi and Mollabashi, Ali},
	date = {2016/03/03},
	id = {Mozaffar2016},
	isbn = {1029-8479},
	journal = {J. High Energ. Phys.},
	number = {3},
	pages = {15},
	title = {On the entanglement between interacting scalar field theories},
	url = {https://doi.org/10.1007/JHEP03(2016)015},
	volume = {2016},
	year = {2016}
}

@article{taylor2016generalized,
	author = {Taylor, Marika},
	date = {2016/07/08},
	doi = {10.1007/JHEP07(2016)040},
	id = {Taylor2016},
	isbn = {1029-8479},
	journal = { J. High Energy Phys.},
	number = {7},
	pages = {40},
	title = {Generalized entanglement entropy},
	url = {https://doi.org/10.1007/JHEP07(2016)040},
	volume = {2016},
	year = {2016}
}

@article{hufel2017field,
title = {{Field space entanglement entropy, zero modes and Lifshitz models}},
journal = {Phys. Lett. B},
volume = {775},
pages = {229-232},
year = {2017},
issn = {0370-2693},
doi = {https://doi.org/10.1016/j.physletb.2017.10.051},
url = {https://www.sciencedirect.com/science/article/pii/S0370269317308638},
author = {Helmuth Huffel and Gerald Kelnhofer},
}

@article{nakai2017entanglement,
  title = {Entanglement entropy and decoupling in the Universe},
  author = {Nakai, Yuichiro and Shiba, Noburo and Yamada, Masaki},
  journal = {Phys. Rev. D},
  volume = {96},
  issue = {12},
  pages = {123518},
  numpages = {20},
  year = {2017},
  month = {Dec},
  publisher = {American Physical Society},
  doi = {10.1103/PhysRevD.96.123518}
}

@article{murtadho2026extensive,
doi = {10.1088/2058-9565/ae8b1d},
url = {https://doi.org/10.1088/2058-9565/ae8b1d},
year = {2026},
month = {aug},
publisher = {IOP Publishing},
volume = {11},
number = {3},
pages = {035060},
author = {Murtadho, Taufiq and Gluza, Marek and Ng, Nelly H Y},
title = {{Extensive entanglement between coupled Tomonaga–Luttinger liquids in and out of equilibrium}},
journal = {Quantum Sci. and Technol.},
}

@article{tajik2023verification,
	author = {Tajik, Mohammadamin and Kukuljan, Ivan and Sotiriadis, Spyros and Rauer, Bernhard and Schweigler, Thomas and Cataldini, Federica and Sabino, Jo{\~a}o and M{\o}ller, Frederik and Sch{\"u}ttelkopf, Philipp and Ji, Si-Cong and Sels, Dries and Demler, Eugene and Schmiedmayer, J{\"o}rg},
	date = {2023/07/01},
	doi = {10.1038/s41567-023-02027-1},
	id = {Tajik2023},
	isbn = {1745-2481},
	journal = {Nat. Phys.},
	number = {7},
	pages = {1022--1026},
	title = {Verification of the area law of mutual information in a quantum field simulator},
	url = {https://doi.org/10.1038/s41567-023-02027-1},
	volume = {19},
	year = {2023}
    }

@article{schweigler2017experimental,
  title={Experimental characterization of a quantum many-body system via higher-order correlations},
  author={Schweigler, Thomas and Kasper, Valentin and Erne, Sebastian and Mazets, Igor and Rauer, Bernhard and Cataldini, Federica and Langen, Tim and Gasenzer, Thomas and Berges, J{\"u}rgen and Schmiedmayer, J{\"o}rg},
  journal={Nature},
  volume={545},
  number={7654},
  pages={323--326},
  year={2017},
  publisher={Nature Publishing Group UK London},
doi={10.1038/nature22310}
}

@article{aimet2025experimentally,
	author = {Aimet, Stefan and Tajik, Mohammadamin and Tournaire, Gabrielle and Sch{\"u}ttelkopf, Philipp and Sabino, Jo{\~a}o and Sotiriadis, Spyros and Guarnieri, Giacomo and Schmiedmayer, J{\"o}rg and Eisert, Jens},
	date = {2025/08/01},
	doi = {10.1038/s41567-025-02930-9},
	id = {Aimet2025},
	isbn = {1745-2481},
	journal = {Nat. Phys},
	number = {8},
	pages = {1326--1331},
	title = {{Experimentally probing Landauer's principle in the quantum many-body regime}},
	url = {https://doi.org/10.1038/s41567-025-02930-9},
	volume = {21},
	year = {2025}
}

@article{haldane1981effective,
  title = {{Effective Harmonic-Fluid Approach to Low-Energy Properties of One-Dimensional Quantum Fluids}},
  author = {Haldane, F. D. M.},
  journal = {Phys. Rev. Lett.},
  volume = {47},
  issue = {25},
  pages = {1840--1843},
  numpages = {0},
  year = {1981},
  month = {Dec},
  publisher = {American Physical Society},
  doi = {10.1103/PhysRevLett.47.1840},
  url = {https://link.aps.org/doi/10.1103/PhysRevLett.47.1840}
}

@article{ruggiero2021large,
  title = {{Large-scale thermalization, prethermalization, and impact of temperature in the quench dynamics of two unequal Luttinger liquids}},
  author = {Ruggiero, Paola and Foini, Laura and Giamarchi, Thierry},
  journal = {Phys. Rev. Res.},
  volume = {3},
  issue = {1},
  pages = {013048},
  numpages = {16},
  year = {2021},
  month = {Jan},
  publisher = {American Physical Society},
  doi = {10.1103/PhysRevResearch.3.013048},
  url = {https://link.aps.org/doi/10.1103/PhysRevResearch.3.013048}
}

@article{dupays2024exact,
  title = {Exact dynamics of the Tomonaga-Luttinger liquid and shortcuts to adiabaticity},
  author = {Dupays, L\'eonce and D\'ora, Bal\'azs and del Campo, Adolfo},
  journal = {Phys. Rev. B},
  volume = {110},
  issue = {15},
  pages = {155104},
  numpages = {23},
  year = {2024},
  month = {Oct},
  publisher = {American Physical Society},
  doi = {10.1103/PhysRevB.110.155104},
  url = {https://link.aps.org/doi/10.1103/PhysRevB.110.155104}
}

@misc{serafini2023quantum,
  author    = {Serafini, Alessio},
   title     = {Quantum Continuous Variables: A Primer of Theoretical Methods (1st ed.)},
  edition   = {1},
  year      = {2017},
  publisher = {CRC Press},
    doi       = {10.1201/9781315118727},
  howpublished = {CRC Press},
}

@article{duan2000inseparability,
  title = {Inseparability Criterion for Continuous Variable Systems},
  author = {Duan, Lu-Ming and Giedke, G. and Cirac, J. I. and Zoller, P.},
  journal = {Phys. Rev. Lett.},
  volume = {84},
  issue = {12},
  pages = {2722--2725},
  numpages = {0},
  year = {2000},
  month = {Mar},
  publisher = {American Physical Society},
  doi = {10.1103/PhysRevLett.84.2722},
  url = {https://link.aps.org/doi/10.1103/PhysRevLett.84.2722}
}

@article{simon2000peres,
  title = {{Peres-Horodecki Separability Criterion for Continuous Variable Systems}},
  author = {Simon, R.},
  journal = {Phys. Rev. Lett.},
  volume = {84},
  issue = {12},
  pages = {2726--2729},
  numpages = {0},
  year = {2000},
  month = {Mar},
  publisher = {American Physical Society},
  doi = {10.1103/PhysRevLett.84.2726},
  url = {https://link.aps.org/doi/10.1103/PhysRevLett.84.2726}
}

@article{vidal2002computable,
  title = {Computable measure of entanglement},
  author = {Vidal, G. and Werner, R. F.},
  journal = {Phys. Rev. A},
  volume = {65},
  issue = {3},
  pages = {032314},
  numpages = {11},
  year = {2002},
  month = {Feb},
  publisher = {American Physical Society},
  doi = {10.1103/PhysRevA.65.032314},
  url = {https://link.aps.org/doi/10.1103/PhysRevA.65.032314}
}

@article{horodecki1998mixed,
  title = {Mixed-State Entanglement and Distillation: Is there a ``Bound'' Entanglement in Nature?},
  author = {Horodecki, Micha\l{} and Horodecki, Pawe\l{} and Horodecki, Ryszard},
  journal = {Phys. Rev. Lett.},
  volume = {80},
  issue = {24},
  pages = {5239--5242},
  numpages = {0},
  year = {1998},
  month = {Jun},
  publisher = {American Physical Society},
  doi = {10.1103/PhysRevLett.80.5239},
  url = {https://link.aps.org/doi/10.1103/PhysRevLett.80.5239}
}

@article{choudhury2022four,
  title={{Four-mode Squeezed States in de Sitter Space: A Study With Two Field Interacting Quantum System}},
  author={Choudhury, Sayantan and Panda, Sudhakar and Pandey, Nilesh and Roy, Abhishek},
  journal={Fortschr. Phys},
  volume={70},
  number={12},
  pages={2200124},
  year={2022},
  publisher={Wiley Online Library},
  doi = {https://doi.org/10.1002/prop.202200124},
}

@article{colas2022four,
  title={{Four-mode squeezed states: two-field quantum systems and the symplectic group Sp (4, R)}},
  author={Colas, Thomas and Grain, Julien and Vennin, Vincent},
  journal={Eur. Phys. J. C},
  volume={82},
  number={1},
  pages={6},
  year={2022},
  publisher={Springer},
    doi={10.1140/epjc/s10052-021-09922-y}
}

@article{viermann2022quantum,
	author = {Viermann, Celia and Sparn, Marius and Liebster, Nikolas and Hans, Maurus and Kath, Elinor and Parra-L{\'o}pez, {\'A}lvaro and Tolosa-Sime{\'o}n, Mireia and S{\'a}nchez-Kuntz, Natalia and Haas, Tobias and Strobel, Helmut and Floerchinger, Stefan and Oberthaler, Markus K.},
	date = {2022/11/01},
	doi = {10.1038/s41586-022-05313-9},
	id = {Viermann2022},
	isbn = {1476-4687},
	journal = {Nature},
	number = {7935},
	pages = {260--264},
	title = {Quantum field simulator for dynamics in curved spacetime},
	url = {https://doi.org/10.1038/s41586-022-05313-9},
	volume = {611},
	year = {2022}
}

@article{
tajik2023experimental,
author = {Mohammadamin Tajik  and Marek Gluza  and Nicolas Sebe  and Philipp Schüttelkopf  and Federica Cataldini  and João Sabino  and Frederik Møller  and Si-Cong Ji  and Sebastian Erne  and Giacomo Guarnieri  and Spyros Sotiriadis  and Jens Eisert  and Jörg Schmiedmayer },
title = {Experimental observation of curved light-cones in a quantum field simulator},
journal = {PNAS},
volume = {120},
number = {21},
pages = {e2301287120},
year = {2023},
doi = {10.1073/pnas.2301287120},
URL = {https://www.pnas.org/doi/abs/10.1073/pnas.2301287120},
}

@article{schweigler2021decay,
	author = {Schweigler, Thomas and Gluza, Marek and Tajik, Mohammadamin and Sotiriadis, Spyros and Cataldini, Federica and Ji, Si-Cong and M{\o}ller, Frederik S. and Sabino, Jo{\~a}o and Rauer, Bernhard and Eisert, Jens and Schmiedmayer, J{\"o}rg},
	date = {2021/05/01},
	doi = {10.1038/s41567-020-01139-2},
	id = {Schweigler2021},
	isbn = {1745-2481},
	journal = {Nat. Phys.},
	number = {5},
	pages = {559--563},
	title = {{Decay and recurrence of non-Gaussian correlations in a quantum many-body system}},
	url = {https://doi.org/10.1038/s41567-020-01139-2},
	volume = {17},
	year = {2021}}

@article{gring2012relaxation,
  title={Relaxation and prethermalization in an isolated quantum system},
  author={Gring, Michael and Kuhnert, Maximilian and Langen, Tim and Kitagawa, Takuya and Rauer, Bernhard and Schreitl, Matthias and Mazets, Igor and Smith, D Adu and Demler, Eugene and Schmiedmayer, J{\"o}rg},
  journal={Science},
  volume={337},
  number={6100},
  pages={1318--1322},
  year={2012},
  publisher={American Association for the Advancement of Science},
doi={10.1126/science.1224953}
}

@article{langen2015experimental,
  title={{Experimental observation of a generalized Gibbs ensemble}},
  author={Langen, Tim and Erne, Sebastian and Geiger, Remi and Rauer, Bernhard and Schweigler, Thomas and Kuhnert, Maximilian and Rohringer, Wolfgang and Mazets, Igor E and Gasenzer, Thomas and Schmiedmayer, J{\"o}rg},
  journal={Science},
  volume={348},
  number={6231},
  pages={207--211},
  year={2015},
  publisher={American Association for the Advancement of Science},
doi={10.1126/science.1257026}
}

@article{langen2013local,
  title={Local emergence of thermal correlations in an isolated quantum many-body system},
  author={Langen, Tim and Geiger, Remi and Kuhnert, Maximilian and Rauer, Bernhard and Schmiedmayer, Joerg},
  journal={Nat. Phys.},
  volume={9},
  number={10},
  pages={640--643},
  year={2013},
  publisher={Nature Publishing Group UK London},
doi = {https://doi.org/10.1038/nphys2739}
}

@misc{jarema2025information,
      title={Information in quantum field theory simulators: Thin-film superfluid helium}, 
      author={Maciej T. Jarema and Cameron R. D. Bunney and Vitor S. Barroso and Mohammadamin Tajik and Chris Goodwin and Silke Weinfurtner},
      year={2025},
      eprint={2508.07247},
      archivePrefix={arXiv},
      primaryClass={quant-ph},
}

@article{yang2020simulating,
  title = {Simulating quantum field theory in curved spacetime with quantum many-body systems},
  author = {Yang, Run-Qiu and Liu, Hui and Zhu, Shining and Luo, Le and Cai, Rong-Gen},
  journal = {Phys. Rev. Res.},
  volume = {2},
  issue = {2},
  pages = {023107},
  numpages = {9},
  year = {2020},
  month = {Apr},
  publisher = {American Physical Society},
  doi = {10.1103/PhysRevResearch.2.023107},
  url = {https://link.aps.org/doi/10.1103/PhysRevResearch.2.023107}
}

@article{murtadho2025measurement,
  title = {Measurement of total phase fluctuation in cold-atomic quantum simulators},
  author = {Murtadho, Taufiq and Cataldini, Federica and Erne, Sebastian and Gluza, Marek and Tajik, Mohammadamin and Schmiedmayer, J\"org and Ng, Nelly H. Y.},
  journal = {Phys. Rev. Res.},
  volume = {7},
  issue = {2},
  pages = {L022031},
  numpages = {8},
  year = {2025},
  month = {May},
  publisher = {American Physical Society},
  doi = {10.1103/PhysRevResearch.7.L022031},
  url = {https://link.aps.org/doi/10.1103/PhysRevResearch.7.L022031}
}

@article{calabrese2012entanglementNeg,
  title = {Entanglement Negativity in Quantum Field Theory},
  author = {Calabrese, Pasquale and Cardy, John and Tonni, Erik},
  journal = {Phys. Rev. Lett.},
  volume = {109},
  issue = {13},
  pages = {130502},
  numpages = {5},
  year = {2012},
  month = {Sep},
  publisher = {American Physical Society},
  doi = {10.1103/PhysRevLett.109.130502},
  url = {https://link.aps.org/doi/10.1103/PhysRevLett.109.130502}
}

@article{zhang2026quantum,
	author = {Zhang, Wenyu and Qian, Wenyang and Zhou, Yiyu and Li, Yang and Wang, Qun},
	date = {2026/07/08},
	doi = {10.1007/JHEP07(2026)061},
	id = {Zhang2026},
	isbn = {1029-8479},
	journal = {Journal of High Energy Physics},
	number = {7},
	pages = {61},
	title = {Quantum entanglement between partons in a strongly coupled quantum field theory},
	url = {https://doi.org/10.1007/JHEP07(2026)061},
	volume = {2026},
	year = {2026}
}

@article{yoshino2021intercomponent,
  title = {{Intercomponent entanglement entropy and spectrum in binary Bose-Einstein condensates}},
  author = {Yoshino, Takumi and Furukawa, Shunsuke and Ueda, Masahito},
  journal = {Phys. Rev. A},
  volume = {103},
  issue = {4},
  pages = {043321},
  numpages = {14},
  year = {2021},
  month = {Apr},
  publisher = {American Physical Society},
  doi = {10.1103/PhysRevA.103.043321},
  url = {https://link.aps.org/doi/10.1103/PhysRevA.103.043321}
}

@article{roosz2022densitymatrix,
  title = {{Density matrix of electrons coupled to Einstein phonons and the electron-phonon entanglement content of excited states}},
  author = {Ro\'osz, Gerg\ifmmode \mbox{\H{o}}\else \H{o}\fi{} and Held, Karsten},
  journal = {Phys. Rev. B},
  volume = {106},
  issue = {19},
  pages = {195404},
  numpages = {10},
  year = {2022},
  month = {Nov},
  publisher = {American Physical Society},
  doi = {10.1103/PhysRevB.106.195404},
  url = {https://link.aps.org/doi/10.1103/PhysRevB.106.195404}
}

@article{roosz2021entanglement,
  title = {{Entanglement of electrons and lattice in a Luttinger system}},
  author = {Ro\'osz, Gerg\ifmmode \mbox{\H{o}}\else \H{o}\fi{} and Timm, Carsten},
  journal = {Phys. Rev. B},
  volume = {104},
  issue = {3},
  pages = {035405},
  numpages = {12},
  year = {2021},
  month = {Jul},
  publisher = {American Physical Society},
  doi = {10.1103/PhysRevB.104.035405},
  url = {https://link.aps.org/doi/10.1103/PhysRevB.104.035405}
}

@article{ruggiero2021quenches,
  title={{Quenches in initially coupled Tomonaga-Luttinger Liquids: a conformal field theory approach}},
  author={Ruggiero, Paola and Calabrese, Pasquale and Foini, Laura and Giamarchi, Thierry},
  journal={SciPost Phys.},
  volume={11},
  number={3},
  pages={055},
  year={2021},
  doi = {10.21468/SciPostPhys.11.3.055}
}

@article{foini2015non,
  title = {{Nonequilibrium dynamics of coupled Luttinger liquids}},
  author = {Foini, L. and Giamarchi, T.},
  journal = {Phys. Rev. A},
  volume = {91},
  issue = {2},
  pages = {023627},
  numpages = {8},
  year = {2015},
  month = {Feb},
  publisher = {American Physical Society},
  doi = {10.1103/PhysRevA.91.023627},
  url = {https://link.aps.org/doi/10.1103/PhysRevA.91.023627}
}

@article{agarwal2023caustics,
  title = {{Caustics in the sine-Gordon model from quenches in coupled one-dimensional Bose gases}},
  author = {Agarwal, Aman and Kulkarni, Manas and O'Dell, D. H. J.},
  journal = {Phys. Rev. A},
  volume = {108},
  issue = {1},
  pages = {013312},
  numpages = {27},
  year = {2023},
  month = {Jul},
  publisher = {American Physical Society},
  doi = {10.1103/PhysRevA.108.013312},
  url = {https://link.aps.org/doi/10.1103/PhysRevA.108.013312}
}

@article{bayoboc2022dynamics,
  title = {Dynamics of thermalization of two tunnel-coupled one-dimensional quasicondensates},
  author = {Bayocboc, F. A. and Davis, M. J. and Kheruntsyan, K. V.},
  journal = {Phys. Rev. A},
  volume = {106},
  issue = {2},
  pages = {023320},
  numpages = {14},
  year = {2022},
  month = {Aug},
  publisher = {American Physical Society},
  doi = {10.1103/PhysRevA.106.023320},
  url = {https://link.aps.org/doi/10.1103/PhysRevA.106.023320}
}

@article{agullo2024toward,
  title = {{Toward the observation of entangled pairs in BEC analog expanding universes}},
  author = {Agullo, Ivan and Delhom, Adri\`a and Parra-L\'opez, \'Alvaro},
  journal = {Phys. Rev. D},
  volume = {110},
  issue = {12},
  pages = {125023},
  numpages = {21},
  year = {2024},
  month = {Dec},
  publisher = {American Physical Society},
  doi = {10.1103/PhysRevD.110.125023},
  url = {https://link.aps.org/doi/10.1103/PhysRevD.110.125023}
}
\onecolumngrid
\appendix
\newpage
\part{Appendix}
\parttoc 
\newpage
\section{Summary of common notations \label{summary_notations}}
\begin{table}[h]
\centering
\renewcommand{\arraystretch}{1.3} 
\begin{tabular}{c   l} 
\hline
& \multicolumn{1}{c}{\textbf{Description of the notation}}\\ 
\hline
\textbf{Fields and Operators} & \\
\hline 
  $\hat{\phi}^{\alpha}(z)$ & Phase fields in real space ($\alpha\in \{A,B,+,-\}$)  \\
  $\delta \hat{n}^{\alpha}(z)$ & Density fields  in real space ($\alpha \in \{A,B,+,-\}$)\\
  $\hat{\phi}_{q}^{\alpha}, \delta \hat{n}_{q}^{\alpha}$ & Dimensionless field quadratures in momentum space ($\alpha \in \{A,B,+,-\}$) \\
  $b_q^\dagger, b_q$ & Bosonic creation and annihilation operators for momentum mode $q$\\
\hline
\textbf{System Parameters} & \\
\hline
  $v$ & Speed of sound \\
  $K$ & Luttinger parameter \\
  $n_0$ & Mean density \\
  $L$ & System length \\
  $\xi$ & Healing length \\ 
  $q$ & Positive finite momentum \\
  $\Lambda$ & Momentum cutoff \\
  $\beta$ & Inverse temperature ($1/k_B T$) \\
  $\displaystyle C_{q}(\beta)=\frac{1}{2}\coth\left(\frac{\beta \hbar\Omega_{q0}}{2}\right)$  & Temperature factor \\
\hline
\textbf{Dynamics} & \\
\hline
  $J(t)$ & Time-dependent tunneling strength \\
    $\Omega_{q0}=vq$ & Initial dispersion relation \\
  $\Omega_q(t)$ & Time-dependent frequency spectrum (Klein-Gordon) \\
  $\omega=\pi n_0 v / K$ & Intrinsic Luttinger frequency scale \\
  $\gamma_q(t)$ & Ermakov factor (solution to Ermakov Eq. \eqref{Ermakov}) \\
\hline
\textbf{Entanglement} & \\
\hline
  $\mathbf{\Gamma}(t)$ & Covariance matrix \\
  $\lambda_q(t,\beta)$ & Symplectic eigenvalues of the (reduced) covariance matrix \\
    $S_A$ & Entanglement Entropy \\
    $S_{\alpha}^A$ & $\alpha$-R{\'e}nyi entropy\\
    
  $I(A:B, t)$ & Mutual information (MI) at time $t$ \\
    $\nu_q(t,\beta)$ & Symplectic eigenvalues of the partially transposed covariance matrix \\
  $E_{\mathcal{N}}$ & Logarithmic negativity (LN) \\
\hline
\textbf{Early-time Behaviour} & \\
\hline
$n=\min \;\{ k \geq 0 \,:\, J^{(k)}(0) \neq 0 \}$ & The smallest order of non-vanishing derivative of $J$ \\
$\varepsilon_n(t)$ & Time-dependent momentum scale\\
$\tau_n = \varepsilon_n(t) L/\pi$ & Scaling parameter of early-time mutual information \\
\hline 
\textbf{Long-time Behaviour} & \\
\hline
$J_f$ & Final interaction strength after protocol saturation\\
$t_f$ & The time at which the protocol is saturated\\
$\mathcal{T}_{*}(q)$ & Mode-dependent threshold temperature for entanglement \\
$A_q$ & Oscillation amplitude for each Ermakov factor after a sudden quench \\
$D_q$ & Constant factor for each Ermakov factor after a finite-ramp protocol \\
$E_q$ & Oscillation amplitude for each Ermakov factor after a finite-ramp protocol \\
\hline
\end{tabular} 
\end{table}
\newpage

\section{Field-space entanglement dynamics for arbitrary Gaussian-preserving protocol and finite temperature \label{Appendix_dynamics}}
This appendix supplies the derivations underlying the exact Gaussian solution presented in Sec.~\ref{sec: PhysicalSetting}.
First, in Appendix~\ref{Appendix_dynamics_of_mode_quadratures}, we solve the Heisenberg equations of motion for the field quadratures and derive the symplectic propagator of Eq.~\eqref{eq:quadrature_dynamics} together with the phase $\theta_q(t)$ of Eq.~\eqref{eq:theta_sol_main_text}.
Then, in Appendix~\ref{dynamics_covariance}, we propagate the thermal initial covariance matrix of Eq.~\eqref{eq:Gamma_pm_teq0} into the $A/B$ basis of Eq.~\eqref{eq:sector_transformation}.
This allows us to use the resulting covariance matrix to prove the mutual information formula, (cf. Eq.~\eqref{eq:mutual_info_dynamics}), and the logarithmic negativity formula (cf. Eqs.~\eqref{LN_sum_q}--\eqref{eq:symp_eigvals_partial_transpose}).
Lastly, we remark that the R{\'e}nyi entropy of arbitrary order can also be obtained by substituting the symplectic
eigenvalues in Eq.~\eqref{eq:subsys_symp_appdx} into Eq.\eqref{eq: Renyi entropies as function of symplectic eigenvalues}.

\subsection{Dynamics of the mode quadratures}\label{Appendix_dynamics_of_mode_quadratures}
In this section, we derive the solution to the Heisenberg equations of motion for the fields' quadratures in terms of the solution to the Ermakov equation, starting from the time-dependent Hamiltonian \eqref{eq:Hamiltonian_time_dep_mt} in the main text.
As the dynamics in the symmetric $(+)$ sector is a special case of the antisymmetric $(-)$ one with a constant frequency, we can first focus on the dynamics in the antisymmetric sector. The Heisenberg equation of motion  for operators $\hat{O}$ in the antisymmetric sector is given by $\partial_{t}\hat{O}(t)=\frac{i}{\hbar}[\hat{H}^-(t),\hat{O}(t)]$. Applying this to the dimensionless field quadratures $\{\delta\hat{n}_q^-, \hat{\phi}_q^-\}$ and using the commutation relation $[\delta\hat{n}_q^-, \hat{\phi}_k^-] = i\delta_{qk}$ yields a coupled equation of motion 
\begin{equation}
    \partial_t\begin{pmatrix}
        \delta\hat{n}_q^-(t)\\
        \hat{\phi}_q^-(t)
    \end{pmatrix} = \begin{pmatrix}
        0 & \Omega_q^2(t)/\omega\\
        -\omega & 0
    \end{pmatrix} \begin{pmatrix}
        \delta\hat{n}_q^-(t)\\
        \hat{\phi}_q^-(t)
    \end{pmatrix}.
    \label{eq:coupled_EOM}
\end{equation}
Since the Hamiltonian is quadratic and independent for each mode $q$, the time-evolution for the quadratures of mode $q$ is described by a symplectic propagator $\mathbf{G}_q^-(t)\in  \text{Sp}(2, \mathbb{R})$
\begin{equation}
    \begin{pmatrix}
\delta\hat{n}_q^-(t)\\
\hat{\phi}_q^-(t)
\end{pmatrix} = \mathbf{G}_q^-(t) \begin{pmatrix}
\delta\hat{n}_q^-(0) \\
\hat{\phi}_q^-(0)
\end{pmatrix} \qquad \text{with} \qquad  \det{\mathbf{G}}_q^-(t) = 1,
\label{eq:symplectic_propagator}
\end{equation}
and $\mathbf{G}_q^-(0) = I_2$ is a $2\times 2$ identity matrix. Note that a $2 \times 2$ matrix with real entries is symplectic if and only if it has a unit determinant. Solving for the propagator $\mathbf{G}_q^-(t)$ is equivalent to solving for the quadrature dynamics. Substituting Eq.~\eqref{eq:symplectic_propagator} to Eq.~\eqref{eq:coupled_EOM} yields a matrix differential equation for $\mathbf{G}_q^-(t)$
\begin{equation}
    \partial_t \mathbf{G}_q^-(t)   = \begin{pmatrix}
        0 & \Omega_q^2(t)/\omega\\
        -\omega & 0
    \end{pmatrix} 
    \mathbf{G}_q^-(t).
    \label{eq:propagator_EOM}
\end{equation}
To solve Eq.~\eqref{eq:propagator_EOM} for any given  protocol encoded in $\Omega_q^2(t)$, consider an ansatz
\begin{equation}
     \mathbf{G}^{-}_q(t) =  \begin{pmatrix}
        \dot{v}(t) & -\dot{u}(t)/\omega\\
        -\omega v(t) & u(t)
    \end{pmatrix} \qquad \text{with} \qquad  W[u(t), v(t)] \equiv u(t)\dot{v}(t) - \dot{u}(t)v(t) = 1 
    \label{eq:ansatz}
\end{equation}
where the unit Wronskian $W[u(t), v(t)] = 1$ condition is derived from the symplectic requirement $\det{\mathbf{G}_q^-(t)} = 1$, i.e. the preservation of the commutation relation. This can be seen from
\begin{align}
[\hat{\phi}_{q}(t),\delta\hat{n}_{q}(t)]&=\left(u(t)\dot{v}(t)-v(t)\dot{u}(t)\right)[\hat{\phi}_{q}(0),\delta\hat{n}_{q}(0)]=[\hat{\phi}_{q}(0),\delta\hat{n}_{q}(0)].
\label{eq:commut_preserve}
\end{align}
The functions $u(t), v(t)$ satisfy initial conditions: $u(0) = 1$, $\dot{u}(0) = 0$, $v(0) = 0$, $\dot{v}(0) = 1$ imposed by the initial condition $\mathbf{G}_q^-(0) = I_2$. The ansatz \eqref{eq:ansatz} satisfies the original Eq.~\eqref{eq:propagator_EOM} if $u(t)$ and $v(t)$ are chosen to be the solutions to homogeneous time-dependent harmonic oscillator equations:
\begin{equation}
    \ddot{u}(t) + \Omega_q^2(t) u(t) = 0, \qquad \qquad \ddot{v}(t) + \Omega_q^2(t) v(t) = 0.
    \label{eq:second_order_ODE}
\end{equation}
Importantly, we know from the unit Wronskian that $u(t)$ and $v(t)$ are linearly independent solutions of \eqref{eq:second_order_ODE}. 

It is known that the two linearly independent solutions to \eqref{eq:second_order_ODE} can be can be fully parametrized by a single dimensionless function $\gamma_q(t)$ satisfying the Ermakov Eq.~\eqref{Ermakov}. For completeness, we show this explicitly below. We start by parametrizing the two solutions with a radial coordinate ansatz
\begin{equation}
    u(t) = \gamma_q(t) \cos \theta_q(t), \qquad \qquad v(t) = \frac{\gamma_q(t)}{D}\sin \theta_q(t),
    \label{eq:radial_parametrization}
\end{equation}
where $\theta_q(t)$ is a dimensionless angle and $D$ is a constant of unit frequency to be determined. This constant is included to take into account that in \eqref{eq:ansatz}, $v(t)$ has a dimension of time whereas $u(t)$ is dimensionless. The initial conditions for $u(t)$ and $v(t)$ impose initial conditions for $\gamma_q(t)$ and $\theta_q(t)$, i.e. $\gamma_q(0) = 1$,  $\dot{\gamma}_q(0) = 0$, $\theta_q(0) = 0$, and $\dot{\theta}_q(0) = D$. 

To find the solution to our dynamics, we need to specify the corresponding $\theta_q(t)$, $\gamma_q(t)$ and $D$. We first solve the evolution of $\theta_q(t)$ from the unit Wronskian condition, which gives
\begin{equation}
    (\gamma_q^2/D)\dot{\theta}_q(t) = 1 \qquad \Rightarrow \qquad  \theta_q(t) = \int _{0}^{t}\frac{D}{\gamma_q^2(s)}ds.
    \label{eq:theta_sol}
\end{equation}
Hence, $\theta_q(t)$ is fully determined by $\gamma_q(t)$ and $D$. Note that Eq. \eqref{eq:theta_sol} already satisfies $\theta_q(0) = 0$ and $\dot{\theta}_q(0) = D$.

We next find the equation of motion for $\gamma_q(t)$. It is straightforward to notice from \eqref{eq:radial_parametrization} that $\gamma_{q}(t)$ can be written in the Pinney form $\gamma_{q}(t)=\sqrt{u^{2}(t)+D^{2}v^{2}(t)}$. To demonstrate that $\gamma_{q}(t)$ satisfies the Ermakov Eq.~\eqref{Ermakov}, we start from $\gamma^{2}_{q}=u^{2}+D^{2} v^{2}$ and then differentiating once 
\begin{align}
\gamma_{q} \dot{\gamma}_{q}&= u\dot{u}+D^{2}v\dot{v},
\label{eq:gamma_q_first_derivative}
\end{align}
and twice $\dot{\gamma}_q^{2}+\gamma_q\ddot{\gamma}_q=(\dot{u})^{2}+u\ddot{u}+D^{2}(\dot{v})^{2}+D^{2}v\ddot{v}$. Using the fact that $u(t)$ and $v(t)$ are the solutions of \eqref{eq:second_order_ODE}, i.e. $\ddot{v}=-\Omega_q^{2}(t)v$ and $\ddot{u}=-\Omega_q^{2}(t)u$, we obtain
\begin{align}
\gamma_q \ddot{\gamma}_q+\Omega_q^{2}(t)\gamma_q^{2}&=\dot{u}^{2}+D^{2}\dot{v}^{2}-\dot{\gamma}^{2}_{q}.
\end{align}
Multiplying both sides by $\gamma_q^2$ and using the expression $\gamma_q^2\dot{\gamma}^{2}_{q} =(u\dot{u}+D^{2} v\dot{v})^{2}$ from Eq. \eqref{eq:gamma_q_first_derivative} yields
\begin{align}
\gamma^{2}_{q}[\gamma_{q}\ddot{\gamma}_{q}+\Omega^{2}_{q}(t)\gamma^{2}_{q}]&=D^{2}(u\dot{v}-\dot{u}v)^{2} = D^2(W[u(t),v(t)])^2.
\end{align}
Finally, by the unit Wronskian condition $W[u(t), v(t)] = 1$, we find that $\gamma_{q}(t)$ is a solution to the Ermakov equation
\begin{align}
\ddot{\gamma}_{q}+\Omega^{2}_{q}(t)\gamma_{q}&=\frac{D^{2}}{\gamma^{3}_{q}}.
\label{eq:Ermakov_appdx}
\end{align}
The remaining task now is to find a self-consistent frequency coefficient $D$. It is instructive to solve the Ermakov equation in the special case of constant frequency $\Omega_q(t) = \Omega_{q}(0) = \Omega_{q0}$. In this case, the solution to the Ermakov equation is $\gamma_q(t)  = \sqrt{D/\Omega_{q0}} =$constant, which must also be true at $t = 0$ where $\gamma_q(0) = 1$. Therefore, we know that $D = \Omega_{q0}$. One can also verify that $D = \Omega_{q0}$ leads to $\theta_q(t) = \Omega_{q0}t$, which gives the right homogeneous solution of \eqref{eq:second_order_ODE} at constant frequency $\ddot{u}+\Omega_{q0}^2u = 0$. 

Taking everything together, we recover Eq.~\eqref{eq:quadrature_dynamics} in the main text for the propagator in the antisymmetric sector
\begin{equation}
    \mathbf{G}_q^-(t) = \underbrace{\begin{pmatrix}
        \gamma_q^{-1}(t) & -\dot{\gamma}_q(t)/\omega\\
        0 & \gamma_q(t)\end{pmatrix}\!}_{\mathbf{U}_q(t)}\underbrace{\begin{pmatrix}
        \cos\theta_q(t) &(\Omega_{q0}/\omega)\sin\theta_q(t)\\
        -(\omega/\Omega_{q0})\sin\theta_q(t) & \cos\theta_q(t)
    \end{pmatrix}}_{\mathbf{R}_q(t)}.
    \label{eq:Gq_minus_sol}
\end{equation}
Eq.~\eqref{eq:Gq_minus_sol} above together with Eqs.~\eqref{eq:theta_sol} and~\eqref{eq:Ermakov_appdx} (after substituting $D = \Omega_{q0}$) fully specify the quadrature dynamics in the antisymmetric sector. Meanwhile, in the symmetric sector the symplectic propagator $\mathbf{G}_q^+(t)$ can be found from the antisymmetric solution by solving the Ermakov equation at fixed frequency $\Omega_{q}(t) = \Omega_{q0}$. The solution to the Ermakov equation here is time-independent $\gamma_q^+(t) = \gamma_q^+(0) =  1$. Consequently, the shearing and squeezing part of the evolution reduces to identity, giving us a quadrature evolution consisting of only symplectic rotation
\begin{equation}
     \begin{pmatrix}
\delta\hat{n}_q^+(t)\\
\hat{\phi}_q^+(t)
\end{pmatrix} = \mathbf{G}_q^+(t) \begin{pmatrix}
\delta\hat{n}_q^+(0) \\
\hat{\phi}_q^+(0)
\end{pmatrix} \qquad \text{with} \qquad \mathbf{G}_q^{+}(t) = \begin{pmatrix}
        \cos(\Omega_{q0}t) & (\Omega_{q0}/\omega)\sin(\Omega_{q0}t)\\
        -(\omega/\Omega_{q0})\sin(\Omega_{q0}t) & \cos(\Omega_{q0}t)
    \end{pmatrix}.
    \label{eq:Gq_+_sol}
\end{equation}
The propagators $\mathbf{G}_q^{\pm}(t)$ derived above can be applied to any initial Gaussian state, and therefore fully specify the fields' quadrature dynamics of our time-dependent Hamiltonian.
\subsection{Dynamics of the covariance matrix \label{dynamics_covariance}}
In this section, we establish the dynamics for the covariance matrix defined in Eq.~\eqref{element_covariance} in the main text. By the $\pm$ decomposition, the covariance matrix takes a block-diagonal form where each block $\mathbf{\Gamma}_q^{++}$ and $\mathbf{\Gamma}_q^{--}$ evolve independently as $\mathbf{\Gamma}_q^{\pm\pm}(t) = \mathbf{G}_q^{\pm}(t)\mathbf{\Gamma}_q^{\pm\pm}(0)\mathbf{G}^{\pm}_q(t)^{ T}$, where $\mathbf{G}_q^{\pm}(t)$ are the symplectic propagators derived in the previous section. We first need to compute the initial covariance matrix $\mathbf{\Gamma}_q^{\pm\pm}(0)$
\begin{align}
\mathbf{\Gamma}_q^{\pm\pm}(0) &=\begin{pmatrix}
\langle \delta \hat{n}^{\pm}_{q}(0)\delta \hat{n}^{\pm}_{q}(0)\rangle & \frac{1}{2}\left[\langle \delta \hat{n}^{\pm}_{q}(0)\hat{\phi}^{\pm}_{q}(0)\rangle +\langle \hat{\phi}^{\pm}_{q}(0)\delta \hat{n}^{\pm}_{q}(0)\rangle\right]\\
\frac{1}{2}\left[\langle \delta \hat{n}^{\pm}_{q}(0)\hat{\phi}^{\pm}_{q}(0)\rangle +\langle \hat{\phi}^{\pm}_{q}(0)\delta \hat{n}^{\pm}_{q}(0)\rangle\right]& \langle \hat{\phi}^{\pm}_{q}(0)\hat{\phi}^{\pm}_{q}(0)\rangle
\end{pmatrix}.\end{align}
In particular, we assume an initial thermal state with temperature $\beta = (k_BT)^{-1}$. We derive the matrix elements of the covariance matrix by diagonalisation of the Hamiltonian \eqref{eq:Hamiltonian_time_dep_mt} at $t = 0$. Consider the decomposition of dimensionless mode quadratures for $q>0$ in terms of bosonic creation and annihilation operators
\begin{equation}
    \hat{\phi}^{\sigma}_{q} =\sqrt{\frac{\omega}{2\Omega_{q0}}}[\hat{b}^{\sigma}_q+\text{h.c.}], \qquad \delta\hat{n}_q^\sigma=\sqrt{\frac{\Omega_{q0}}{2\omega}}[i\hat{b}^{\sigma}_q+\text{h.c.}],
    \label{eq:mode_deco_0}
\end{equation}
with $\sigma \in \{+, -\}$ and the bosonic operators satisfy the commutation relations $[\hat{b}^{\kappa}_{q},(\hat{b}^{\sigma}_{q^\prime})^{\dagger}]=\delta_{qq^\prime}\delta_{\kappa\sigma}$. Using the above decomposition, the initial Hamiltonian can be diagonalised as
\begin{equation}
\hat{H}_{\pm}(0) =\sum_{q>0}\hbar \Omega_{q0} \left[(\hat{b}^{\pm}_{q})^{\dagger}\hat{b}^{\pm}_{q}+\frac{1}{2}\right],
\end{equation}
with the associated thermal density operator $\rho(0)=\rho_{+}(0)\otimes \rho_{-}(0)=\frac{e^{-\beta\hat{H}_{+}(0)}}{\text{Tr}(e^{-\beta\hat{H}_+})}\otimes \frac{e^{-\beta\hat{H}_{-}(0)}}{\text{Tr}(e^{-\beta\hat{H}_-})}$.
The corresponding covariance matrix elements for this thermal density operator are:
\begin{equation}
\langle \delta \hat{n}^{\pm}_{q}(0)\delta \hat{n}^{\pm}_{q}(0)\rangle =\frac{\Omega_{q0}}{2\omega}{\rm Tr}\left[\rho(0)\left(2(\hat{b}^{\pm}_{q})^\dagger\hat{b}_{q}^\pm+1-((\hat{b}_q^\pm)^\dagger)^2-(\hat{b}^{\pm}_q)^{2}\right)\right]=\frac{\Omega_{q0}}{\omega}C_{q}(\beta),
\end{equation}
\begin{equation}
    \langle \hat{\phi}^{\pm}_{q}(0)\hat{\phi}^\pm_{q}(0)\rangle ={\rm Tr}\left[\rho(0)\left(2(\hat{b}^{\pm}_{q})^\dagger\hat{b}_{q}^\pm+1+((\hat{b}_q^\pm)^\dagger)^2+(\hat{b}^{\pm}_q)^{2}\right)\right]=\frac{\omega}{\Omega_{q0}}C_{q}(\beta),
\end{equation}
where we defined $C_q(\beta)$ in Eq.~\eqref{eq:C_q} and all other elements are zeros. Note that we have used the bosonic mean thermal occupation $\braket{(\hat{b}_q^\pm)^\dagger\hat{b}_q^\pm } = [\exp(\beta\hbar\Omega_{q0})-1]^{-1}$. In the end, we recover the initial covariance matrix in Eq. \eqref{eq:Gamma_pm_teq0} of the main text
\begin{equation}
    \mathbf{\Gamma}_q^{++}(0) = \mathbf{\Gamma}_q^{--}(0) = \begin{pmatrix}
        \Omega_{q0}/\omega& 0\\
        0&\omega/\Omega_{q0}
\end{pmatrix}C_{q}(\beta).
\end{equation}
Using the above thermal initial state and the propagators \eqref{eq:Gq_minus_sol} and \eqref{eq:Gq_+_sol} derived in the previous section, we observe that the symplectic rotation part of the evolution leaves the initial covariance matrix invariant. For instance, in the symmetric sector 
\begin{equation}
    \mathbf{\Gamma}_q^{++}(t) = \mathbf{G}_q^+(t)\mathbf{\Gamma}_q^{++}(0)\mathbf{G}_q^{+}(t)^{T} = \mathbf{\Gamma}_q^{++}(0),
\end{equation}
and similarly in the antisymmetric sector $\mathbf{\Gamma}_q^{--}(t) = \mathbf{U}_q(t)\mathbf{R}_q(t)\mathbf{\Gamma}_q^{--}(0)\mathbf{R}_q^T(t)\mathbf{U}_q^T(t) = \mathbf{U}_q(t)\mathbf{\Gamma}_q^{--}(0)\mathbf{U}_q^T(t)$.

We are interested in probing entanglement dynamics in the $A/B$ decomposition instead of the $\pm$ sectors. The two representations are connected by [see Eq. \eqref{eq:sector_transformation} in the main text]
\begin{equation}
    \mathbf{\Gamma}_q(t) = \frac{1}{2}\begin{pmatrix}
        1&1\\
        1& -1
    \end{pmatrix}\begin{pmatrix}
        \mathbf{\Gamma}_q^{++}(t) & 0\\
        0 & \mathbf{\Gamma}_q^{--}(t)
    \end{pmatrix}\begin{pmatrix}
        1&1\\
        1& -1
    \end{pmatrix} = \frac{1}{2} \begin{pmatrix}
        \mathbf{\Gamma}_q^{++}(t) +\mathbf{\Gamma}_q^{--}(t)& \mathbf{\Gamma}_q^{++}(t) - \mathbf{\Gamma}_q^{--}(t)\\
        \mathbf{\Gamma}_q^{++}(t) - \mathbf{\Gamma}_q^{--}(t) & \mathbf{\Gamma}_q^{++}(t) + \mathbf{\Gamma}_q^{--}(t)
    \end{pmatrix}.
\end{equation}
By matching the elements, we find
\begin{align}
    \mathbf{\Gamma}_q^{AA}(t) &= \mathbf{\Gamma}_q^{BB}(t) = \frac{1}{2}[\mathbf{\Gamma}_q^{++}(t)+\mathbf{\Gamma}_q^{--}(t)] = \frac{1}{2}[\mathbf{\Gamma}_q^{++}(0)+\mathbf{U}_q(t)\mathbf{\Gamma}_q^{--}(0)\mathbf{U}_q^T(t)],
    \label{eq:Gamma_AA_t}\\
    \mathbf{\Gamma}_q^{AB}(t) &= \mathbf{\Gamma}_q^{BA}(t) = \frac{1}{2}[\mathbf{\Gamma}_q^{++}(t)-\mathbf{\Gamma}_q^{--}(t)] = \frac{1}{2}[\mathbf{\Gamma}_q^{++}(0)-\mathbf{U}_q(t)\mathbf{\Gamma}_q^{--}(0)\mathbf{U}_q^T(t)].
\end{align}
And by explicit computation, one obtains
\begin{align}
{\bf \Gamma}^{AA}_{q}(t)=\frac{1}{2\Omega_{q0}}C_{q}(\beta)\begin{pmatrix}
\rho_{q} & -\eta_{q}\\
-\eta_{q} & \kappa_{q}
\end{pmatrix},\qquad
{\bf \Gamma}^{AB}_{q}(t)=\frac{1}{2\Omega_{q0}}C_{q}(\beta)\begin{pmatrix}\mu_{q} & \eta_{q}\\
\eta_{q} &-\chi_{q} 
\end{pmatrix},
\end{align}
where $C_q(\beta)$ is defined in Eq.~\eqref{eq:C_q}, and
\begin{align}
\eta_q &= \gamma_q\dot{\gamma}_q,\qquad
\kappa_q = (1+\gamma_q^2)\omega,\qquad
\chi_q = (\gamma_q^2-1)\omega, \label{Eq: correlation function part 1}\\
\rho_q &= \frac{1}{\gamma_q^2\omega}\left[\Omega_{q0}^2 +\gamma_q^2(\Omega_{q0}^2+\dot{\gamma}_q^2)\right],\qquad
\mu_q = \frac{1}{\gamma_q^2\omega}\left[-\Omega_{q0}^2+\gamma_q^2(\Omega_{q0}^2-\dot{\gamma}_q^2)\right] \label{Eq: correlation function part 2}.
\end{align}

In the end, the total covariance matrix is written as
\begin{equation}
    \mathbf{\Gamma} (t) = \bigoplus_{q>0}\mathbf{\Gamma}_q(t) = \bigoplus_{q>0}\frac{1}{2\Omega_{q0}}C_{q}(\beta)\begin{pmatrix}
\rho_{q} & -\eta_{q} & \mu_{q} & \eta_{q}\\
-\eta_{q} & \kappa_{q} & \eta_{q} & -\chi_{q}\\
\mu_{q} & \eta_{q} & \rho_{q} & -\eta_{q}\\
\eta_{q} & -\chi_{q} & -\eta_{q} & \kappa_{q}
\end{pmatrix}.
\end{equation}
\subsection{Symplectic eigenvalues and the dynamics of the mutual information}\label{appdxB:symp_eigval_mi_dynamics}
The dynamics of quantum information-theoretic quantities for Gaussian states can be extracted from the symplectic eigenvalues of the covariance matrix. For a single-mode covariance matrix, the symplectic eigenvalue is $\lambda = \left|\text{eig}[i\mathbf{\Sigma}\mathbf{\Gamma}]\right|$, where 
$\mathbf{\Sigma} = \bigl( \begin{smallmatrix} 0 & 1 \\ -1 & 0 \end{smallmatrix} \bigr)$ is the symplectic matrix. In other words, we compute the positive eigenvalue(s) of the matrix $i\mathbf{\Sigma}\mathbf{\Gamma}$. For $N$-mode Gaussian states, the symplectic eigenvalues are found by diagonalising $i\left(\bigoplus_{j=1}^{N}\mathbf{\Sigma}\right)\mathbf{\Gamma}$. 
In our case, the covariance matrix can be decomposed into a direct sum of two-modes covariance matrices $\mathbf{\Gamma} = \bigoplus_{q>0} \mathbf{\Gamma}_q$. This means that in each momentum sector, the symplectic eigenvalues are obtained by diagonalising $i(\mathbf{\Sigma} \oplus \mathbf{\Sigma})\mathbf{\Gamma}_q$. Meanwhile, the subsystem covariance matrices $\mathbf{\Gamma}_q^{\sigma\sigma}$ with $\sigma \in \{A, B\}$ is a single-mode system with a symplectic eigenvalue that can be computed by diagonalising $i\mathbf{\Sigma}\mathbf{\Gamma}_{q}^{\sigma\sigma}$. 

Mutual information (MI) for Gaussian states is fully characterised by the symplectic eigenvalues of the (subsystem) covariance matrix. Explicit calculation of $i\mathbf{\Sigma }\mathbf{\Gamma}_q^{AA}$ with Eq.~\eqref{eq:Gamma_AA_t} yields
\begin{align}
i{\bf \Sigma}{\bf \Gamma}^{AA}_{q}&=i\frac{1}{2\Omega_{q0}}C_{q}(\beta)\begin{pmatrix}
-\eta_{q} & \kappa_{q}\\
-\rho_{q} & \eta_{q}
\end{pmatrix}.
\end{align}
Diagonalising this matrix gives positive eigenvalues of the form
\begin{align}\label{eq:subsys_symp_appdx}
    \lambda_q^A(t,\beta) = \frac{1}{2}C_{q}(\beta)\sqrt{\left(\frac{\dot{\gamma}_q(t)}{\Omega_{q0}}\right)^2+\left(\gamma_q(t)+\frac{1}{\gamma_q(t)}\right)^2}     = C_{q}(\beta)\sqrt{1+\mathcal{F}[\gamma_q(t)]} ,
\end{align}
where $\mathcal{F}[\gamma_q(t)]$ is defined in Eq. \eqref{eq:ermakov_com_func} in the main text. The symplectic eigenvalues for the subsystem $B$ are identical.

Meanwhile, the symplectic eigenvalues of the joint covariance matrix are $\lambda^{\pm}_{q}(t)=\Big|{\rm eig}\Big[i({\bf \Sigma}\oplus {\bf \Sigma}){\bf \Gamma}_{q}(t)\Big]\Big|$ with the two-modes covariance matrix taking a block-diagonal form $\mathbf{\Gamma}_q(t) = \mathbf{\Gamma}_q^{++}(t) \oplus \mathbf{\Gamma}_q^{--}(0)$. 
For a $2\times 2$ block matrix, the symplectic eigenvalues are simply $\lambda^{\pm}_{q}(t)=\sqrt{{\rm det}\;{\bf \Gamma}^{\pm\pm}_{q}(t)}$. Because ${\bf G}^{\pm}_{q}(t)$ is symplectic, ${\rm det}\;{\bf G}^{\pm}_{q}(t)=1$, so ${\rm det}\;{\bf \Gamma}^{\pm\pm}_{q}(t)={\rm det}\;{\bf \Gamma}^{\pm\pm}_{q}(0)$. The two time-independent positive eigenvalues are given by $\lambda_q^{\pm}(t) = \lambda_q^{\pm}(0)=C_{q}(\beta)$. Using these expressions for symplectic eigenvalues, the MI can be calculated according to the prescription given in Subsec.~\ref{sec:MI_dynamics}. This proves the result presented in Subsec.~\ref{sec:MI_dynamics} in the main text.
\subsection{Dynamics of the logarithmic negativity \label{appendix_LN}}
Here we derive the general result on the dynamics of entanglement as quantified by the logarithmic negativity (LN), presented in Subsec.~\ref{subsec:ln_dynamics} of the main text. The LN is determined by the symplectic eigenvalues of the partially transposed covariance matrix. Let us first define the partial transpose operator, that corresponds to flipping the sign of the momentum quadratures of $B$, i.e. $(\delta\hat n_q^A,\delta\hat\phi_q^A,\delta\hat n_q^B,\delta\hat\phi_q^B)
\mapsto
(\delta\hat n_q^A,\delta\hat\phi_q^A,\delta\hat n_q^B,-\delta\hat\phi_q^B)$. In other words, we can write
\begin{equation}
    \mathbf{\Gamma}_q^{T_B} = \mathscr{T}_B\begin{pmatrix}
        \mathbf{\Gamma}_q^{AA} & \mathbf{\Gamma}_q^{AB}\\
        \mathbf{\Gamma}_q^{AB} & \mathbf{\Gamma}_q^{AA}\\
    \end{pmatrix}\mathscr{T}_B^T  = \begin{pmatrix}
        \mathbf{\Gamma}_q^{AA} & \mathscr{\Gamma}_q^{AB}\sigma_z\\
        \sigma_z\mathbf{\Gamma}_q^{BA} & \sigma_z\mathbf{\Gamma}_q^{BB}\sigma_z\\
    \end{pmatrix},
\end{equation}
where $\mathscr{T}_B = {\rm diag}(1,1,1,-1) = I_2 \oplus \sigma_z,$ with $\sigma_z$ being the Pauli-Z matrix. 
%
%
%
To solve for LN dynamics, we calculate the symplectic eigenvalues of $\mathbf{\Gamma}_q^{T_B}$, or in other words the eigenvalues of $i(\mathbf{\Sigma} \oplus\mathbf{\Sigma})\mathbf{\Gamma}_q^{T_B}$, where we write out their forms explicitly:
\begin{align}
{\bf \Gamma}^{T_{B}}_{q}&=\frac{C_{q}(\beta)}{2 \Omega_{q0}}\begin{pmatrix}
\rho_{q} & -\eta_{q} & \mu_{q} & -\eta_{q}\\
-\eta_{q} & \kappa_{q} & \eta_{q} & \chi_{q}\\
\mu_{q} & \eta_{q} & \rho_{q} & \eta_{q}\\
-\eta_{q} & \chi_{q}& \eta_{q} & \kappa_{q}
\end{pmatrix}, \qquad i(\mathbf{\Sigma} \oplus\mathbf{\Sigma})\mathbf{\Gamma}_q^{T_B}=\frac{iC_{q}(\beta)}{2 \Omega_{q0}}\begin{pmatrix}
-\eta_{q} & \kappa_{q}& \eta_{q} & \chi_{q}\\
-\rho_{q} & \eta_{q} & -\mu_{q} & \eta_{q}\\
-\eta_{q} & \chi_{q} & \eta_{q} & \kappa_{q}\\
-\mu_{q} & -\eta_{q} & -\rho_{q} & -\eta_{q}
\end{pmatrix}.
\end{align}
The associated symplectic eigenvalues are $\nu^{\pm}_{q}=\left|\text{eig}\left[i(\mathbf{\Sigma} \oplus\mathbf{\Sigma})\mathbf{\Gamma}_q^{T_B}\right]\right| =C_{q}(\beta)\sqrt{V_{1}\pm \sqrt{V_{2}}}$,
where we define
\begin{align}
V_1&\equiv \frac{1}{4\Omega^{2}_{q0}}(\kappa_{q}\rho_{q}+\mu_{q}\chi_{q})=\frac{1}{2}\left[\frac{\dot{\gamma}^{2}_{q}}{\Omega^{2}_{q0}}+\frac{1}{\gamma^{2}_{q}}+\gamma^{2}_{q}\right],\\
V_{2}&\equiv\frac{1}{16\Omega^{4}_{q0}}(\kappa_{q}\mu_{q}+\rho_{q}\chi_{q})^{2}+\frac{1}{4\Omega^{4}_{q0}}\eta^{2}_{q}(\kappa_{q}-\chi_{q})(\rho_{q}+\mu_{q})=\frac{1}{4}\left[\left(\frac{\dot{\gamma}_{q}}{\Omega_{q0}}\right)^{2}+\frac{1}{\gamma^{2}_{q}}-\gamma^{2}_{q}\right]^{2}+\left(\frac{\dot{\gamma}_{q}\gamma_{q}}{\Omega_{q0}}\right)^{2}.
\end{align}
These expressions are obtained using Eqs.~\eqref{Eq: correlation function part 1} -- \eqref{Eq: correlation function part 2}.
Before we proceed, let's simplify the nested square root component.
\begin{enumerate}
    \item \underline{Denesting two nested square roots}: One can check that $V_1 \geq 0$ and $0\le V_2\le V_1^2$. Hence, applying the denesting identity for two nested square roots, we obtain
\begin{align}
    \sqrt{V_1 \pm \sqrt{V_2}}&= \sqrt{\frac{V_1+\sqrt{V_1^2-V_2}}{2}} \pm \sqrt{\frac{V_1-\sqrt{V_1^2-V_2}}{2}},
\end{align}
where we use the denesting identity for two nested square roots in the last line.

    \item \underline{Evaluation of $V_1^2-V_2$}: 
    Directly evaluating the terms give
    \begin{align}
        V_1^2-V_2&= \frac{1}{4} \left[4+4\left(\frac{\dot{\gamma}_{q}\gamma_{q}}{\Omega_{q0}}\right)^{2}\right]-\left(\frac{\dot{\gamma}_{q}\gamma_{q}}{\Omega_{q0}}\right)^{2}=1 \label{eq: V1 squared - V2}.
    \end{align}

     \item \underline{Evaluation of $V_1\pm\sqrt{V_1^2-V_2}$}:
     Using Eq.~\eqref{eq: V1 squared - V2}, we have $V_1\pm\sqrt{V_1^2-V_2}=\frac{1}{2}\left[\left(\frac{1}{\gamma_{q}}\pm \gamma_{q}\right)^2+\left(\frac{\dot{\gamma}_{q}}{\Omega_{q0}}\right)^{2}\right]$.
\end{enumerate}
Combining all results gives rise to
\begin{align}
    \nu_q^{\pm}(t) =C_q(\beta) \sqrt{V_1 \pm \sqrt{V_2}}&=C_q(\beta) \left( \sqrt{\frac{V_1+\sqrt{V_1^2-V_2}}{2}} \pm \sqrt{\frac{V_1-\sqrt{V_1^2-V_2}}{2}}\right)\\
    &=\frac{C_{q}(\beta)}{2}\left[\sqrt{\left(\frac{\dot{\gamma}_q(t)}{\Omega_{q0}}\right)^2+\left(\gamma_q(t)+\frac{1}{\gamma_q(t)}\right)^2}\pm \sqrt{\left(\frac{\dot{\gamma}_q(t)}{\Omega_{q0}}\right)^2+\left(\gamma_q(t)-\frac{1}{\gamma_q(t)}\right)^2}\right] \\
    & =  \frac{C_{q}(\beta)}{2}\bigg(\sqrt{1+\mathcal{F}[\gamma_q(t)]}\pm\sqrt{\mathcal{F}[\gamma_q(t)]}\bigg),
\end{align}
where $\mathcal{F}[\gamma_q(t)]$ is defined in Eq.~\eqref{eq:ermakov_com_func} in the main text. Since $C_q(\beta)  \geq 1/2$ and by construction $\mathcal{F}[\gamma_q(t)]\geq 0$, one can observe that $\nu_q^+(t) \geq \frac{1}{2}$. Consequently, $\nu_q^+(t)$ \textit{does not} contribute to LN, and only $\nu_q^-(t)$ contributes. This proves Eqs. \eqref{LN_sum_q} and \eqref{eq:symp_eigvals_partial_transpose} in the main text, where we have simplified the notation $\nu_q^-(t) \to \nu_q(t,\beta)$.
\section{Early-time power-law scaling of mutual information at zero initial temperature \label{Appendix_early_time}}
This Appendix establishes our results on the early-time power-law scaling of mutual information for zero initial temperature (Sec. \ref{early_time_MI}).  
%
For a sudden quench, we derive the early-time Ermakov factor of Eq.~\eqref{eq:gamma_t_early_time} and the resulting symplectic eigenvalue of Eq.~\eqref{eq: early_symplectic_eigenvalue} in Appendix~\ref{Appendix_early_time_sudden_quench}).
Similarly, for a generic smooth finite-time protocol, the corresponding derivations are presented in Appendix~\ref{Ermakov_equation_early_time}.
These are then used to determine the ultra-early-time and early-time MI scaling laws of Eqs.~\eqref{eq: MI_small_alpha} and~\eqref{eq: MI propto epsilon} (Appendix~\ref{mutual_info_early_time}).
Moreover, we determine the numerical values of the crossover time $t^{(n)}_{\rm cut}$ between these two distinct regimes (Appendix~\ref{subsec: early time breakdown}), plotted in Fig.~\ref{figure_ealy_time} .

\subsection{Early-time expansion of the symplectic eigenvalues for a sudden quench \label{Appendix_early_time_sudden_quench}}
For a sudden quench, we have the exact solution of the Ermakov equation (cf. Eq.~\eqref{eq:solution_sudden_quench}): $\gamma_q(t) = \sqrt{1-\left(1-\zeta_q^2\right)\sin^2(\Omega_{qf} t)}$,
where $\zeta_q \equiv \Omega_{q0}/\Omega_{qf}$. 
At early times $\Omega_{qf}t \ll 1$, Taylor expansion yields
\begin{align}
     \gamma_q(t) &\approx \sqrt{1-\left(1-\zeta_q^2\right)\left[(\Omega_{qf}t)^{2}+\mathcal{O}(t^{4})\right]} 
     \approx 1-\frac{t^2}{2}[\Omega_{qf}^2-\Omega_{q0}^2]+\mathcal{O}(t^{4}) \approx 1 - 2\omega J_ft^2.
     \label{approx_ermakov_sudden_simplified}
\end{align}
The approximation $\Omega_{qf}t = \sqrt{(vq)^2 + 4 \omega J_f}t  \ll 1$ requires both early-time approximation and a momentum cutoff, i.e.
\begin{align}
 t\ll\frac{1}{\sqrt{(vq)^2+4\omega J_f}} < \frac{1}{\sqrt{4\omega J_f}}\equiv t_0^* \;\text{ and }\;
 q &\ll \frac{1}{vt} \text{ for } t \in[0, t_0^*).
\end{align}
We adopt the most restrictive bound for the momentum cutoff 
$q \ll\Lambda \equiv (vt^{*}_{0})^{-1}$. We then approximate
    \begin{align}
     \left(\frac{\dot{\gamma}_{q}(t)}{\Omega_{q0}}\right)^{2}+\left(\gamma_{q}(t)+\frac{1}{\gamma_{q}(t)}\right)^{2}
        &\approx 4+\left(\frac{4\omega J_{f} t}{\Omega_{q0}}\right)^{2}+O(t^{4}).
    \end{align}
In the zero-temperature limit $C_q(\beta)=1/2$, the symplectic eigenvalues simplify into
\begin{align}
\lambda_{q}(t,\beta)&=\frac{1}{4}\sqrt{\left(\frac{\dot{\gamma}_{q}(t)}{vq}\right)^{2}+\left(\gamma_{q}(t)+\frac{1}{\gamma_{q}(t)}\right)^{2}} \approx\frac{1}{2}\sqrt{ 1+\left(\frac{2\omega J_{f}}{\Omega_{q0}}t\right)^{2}}  \approx\frac{1}{2}\sqrt{1+\left(\frac{\varepsilon_{0}(t)}{q}\right)^{2}} \label{app:lambda_epsilon_0},
\end{align}
where on the second line we define $\varepsilon_{0}(t)\equiv(2\omega J_{f} t)/v$, and the subscript $0$ is associated with the sudden quench protocol. 
This notation is chosen to be consistent with the general notation $\varepsilon_{n}(t)$ for protocols whose first non-vanishing derivative at $t=0$ is of order $n$ described in the next section.

\subsection{Early-time expansion of the symplectic eigenvalues for a smooth protocol \label{Ermakov_equation_early_time}}
In this section, we provide an integral representation to the Ermakov Eq. \eqref{Ermakov}, useful to obtain the early-time expansion of the Ermakov factor for general protocols. We first rearrange Eq.~\eqref{Ermakov} into
\begin{align}
     \ddot{\gamma}_q(t) = \frac{\Omega_{q0}^2}{\gamma_q^3(t)}  -\Omega_q^2(t) \gamma_q(t), 
\end{align}
and integrate with respect to time twice to obtain a formal solution
\begin{align}
\gamma_q(t) &= 1+ \int_0^t \left( \int_0^{t'} \left[ \frac{\Omega_{q0}^2}{\gamma_q^3(s)} - \Omega_q^2(s) \gamma_q(s) \right] ds \right) dt',\label{gamma_integral}
\end{align}
where we have used the initial conditions $\gamma_{q}(0)=1$ and $\dot{\gamma}_q(0) = 0$. Using the Cauchy formula for repeated integration
\begin{align}
\int_{a}^{x}\left(\int_{a}^{t}f(s)ds\right)dt=\int_{a}^{x}(x-s)f(s)ds,
\end{align}
Eq.~\eqref{gamma_integral} can be simplified into 
\begin{align}
    \gamma_q(t) &=1 + \int_0^t (t-t') \left[ \frac{\Omega_{q0}^2}{\gamma_q^3(t')} - \Omega_q^2(t') \gamma_q(t') \right] dt' =: 1+u_q(t),\label{eq: integral representation of Ermakov solution}
\end{align}
with $u_q(t)$ being defined as the integral.

Our goal is to approximately solve for $u_q(t)$. To do this, we substitute $\gamma_q(t) = 1+u_q(t)$ into the square bracket in Eq. \eqref{eq: integral representation of Ermakov solution}. 
For early times, we have $|u_q(t)| \ll 1$, and hence the integrand can be approximated as
\begin{align}
    \frac{\Omega_{q0}^2}{[1+u_q(t)]^3} - [\Omega_{q0}^2+4\omega J(t)][1+u_q(t)]
    &\approx -4\left\{u_q(t)\Omega_{q0}^2+\omega J(t)[1+u_q(t)]\right\}+O\left(u_q^2\right) \\
    &= -4\omega J(t) - 4u_q(t)\left[\Omega_{q0}^2+\omega J(t)\right] + O\left(u_q^2\right).\label{eq: square bracket term}
\end{align}
Substituting \eqref{eq: square bracket term} into the integral in \eqref{eq: integral representation of Ermakov solution} yields 
\begin{align}
    u_q(t) &\approx -4\omega\int_{0}^{t}(t-t^\prime)J(t^\prime)\,dt^\prime
    - 4\Omega_{q0}^2\int_{0}^{t}(t-t^\prime)u_q(t^\prime)\,dt^\prime
    - 4\omega\int_{0}^{t}(t-t^\prime)u_q(t^\prime)J(t^\prime)\,dt^\prime.
    \label{eq: uq integral eq full}
\end{align}
Let us focus only on the first (leading) term, i.e.
\begin{equation}
    u_q(t) \approx  -4\omega \int_{0}^{t}(t-t^\prime)J(t^\prime)\; dt^\prime.
    \label{eq: tunneling_kernel}
\end{equation}

The $q$-dependence drops out as it only plays a role in the subleading term. 
The integral in \eqref{eq: tunneling_kernel} is clearly protocol dependent. To simplify the problem further, we perform a Taylor expansion of the tunneling protocol around $t^\prime = 0$ and take only the first leading term, i.e., $J(t^\prime) \approx \frac{J^{(n)}(0)}{n!}(t^\prime)^n$ where $n$ is an integer and $J^{(n)}(0)$ is the first non-null derivative of $J(t)$ evaluated at $t = 0$. This expansion yields
\begin{equation}
    u_q(t) \approx -\frac{4\omega J^{(n)}(0)}{n!}\int_{0}^{t}(t-t^\prime)(t^\prime)^n\;dt^\prime = -\frac{4\omega J^{(n)}(0)}{(n+2)!}t^{n+2},
    \label{eq: u_q approximate}
\end{equation}
where in the last step we used the integral
\begin{align}
\int_{0}^{t}(t-t')(t')^{m}dt'=\frac{t^{m+2}}{(m+1)(m+2)},
\end{align}
valid for $m>-1$. Hence, to the leading order, the Ermakov factor for all modes behave as (cf. Eq. \eqref{eq:gamma_t_early_time})
\begin{equation}
    \gamma_q(t) \approx 1 - \frac{4\omega J^{(n)}(0)}{(n+2)!}t^{n+2},
    \label{eq: gamma_approx}
\end{equation}
in the early-time regime specified more precisely below. Eq.~\eqref{eq: gamma_approx} is consistent with the early-time behavior of $\gamma_q(t)$ for the sudden quench ($n=0$) given by \eqref{approx_ermakov_sudden_simplified}. We next self-consistently define the conditions for which \eqref{eq: gamma_approx} is a good approximation. The first condition is $|u_q(t)| \ll 1$, giving us a precise definition of the early-time regime for a given protocol
\begin{equation}
    t \ll t_n^* = \left[\frac{(n+2)!}{4\omega |J^{(n)}(0)|}\right]^{\frac{1}{n+2}}.
    \label{eq: early time limit}
\end{equation}
The other condition is derived from ignoring the subleading terms in \eqref{eq: uq integral eq full}, i.e.
\begin{align}
   \left| 4(vq)^2\int_{0}^{t}(t-t^\prime)u_q(t^\prime)\,dt^\prime+
    4\omega\int_{0}^{t}(t-t^\prime)u_q(t^\prime)J(t^\prime)\,dt^\prime\right|  \ll \left|u_q(t)\right|. 
\end{align}
The above condition will set a momentum cutoff $\Lambda$ that generally depends on the protocol $J(t)$.

Given the knowledge of the early time behavior of $\gamma_q(t)$, we can now derive the early-time behavior for symplectic eigenvalues, 
\begin{align} \label{app:lambda_q}
\lambda_{q}(t,\beta)&=\frac{C_{q}(\beta)}{2} \sqrt{\left(\frac{\dot{\gamma}_{q}(t)}{\Omega_{q0}}\right)^{2}+\left(\gamma_{q}(t)+\frac{1}{\gamma_{q}(t)}\right)^{2}}.
\end{align}
Let us rewrite Eq.~\eqref{eq: gamma_approx} as $\gamma_q(t)=1-A t^{n+2}+\mathcal{O}(t^{n+3})$, for $n\geq 0$, 
where we set $A=4\omega \frac{J^{(n)}(0)}{(n+2)!}$ for simplicity. We insert this expression in \eqref{app:lambda_q}
\begin{align} \label{app:lambda_beta_early}
     \lambda_{q}(t,\beta) &=\frac{C_{q}(\beta)}{2}  \sqrt{\left(\frac{-A(n+2)t^{n+1}}{\Omega_{q0}}\right)^{2} + \left[1 - At^{n+2} + (1 + At^{n+2} + \mathcal{O}(t^{n+3})) \right]^{2}} \\
     &\approx C_{q}(\beta)\sqrt{1+ \frac{A^2 (n+2)^2}{4 \Omega_{q0}^2} t^{2n+2} +\mathcal{O}(t^{2n+3})}
      \\
      &\approx C_{q}(\beta)\sqrt{1+\left(\frac{\varepsilon_{n}(t)}{q}\right)^{2}},
\end{align}
where in the last line we omitted higher order terms in time of order $\mathcal{O}(t^{2n+3})$, and defined $\varepsilon_{n}(t)\equiv\frac{2\omega J^{(n)}(0)}{v(n+1)!}t^{n+1}$---generalising $\varepsilon_0$ in \eqref{app:lambda_epsilon_0} for a sudden quench to smooth protocols.

\subsection{Early-time expansion of the mutual information (MI) at zero initial temperature \label{mutual_info_early_time}}
In the early-time regime, the symplectic eigenvalue at zero-temperature is approximately $\lambda_{q}(t,\beta\rightarrow\infty) \approx \frac{1}{2}\sqrt{1+\left(\frac{\varepsilon_{n}(t)}{q}\right)^{2}}$,
as obtained in the $\beta\to \infty$ limit of Eq. \eqref{app:lambda_beta_early}.
This form allows us to determine an approximated expression for the mutual information at early time. Consider a system of large size with infrared (IR) and ultraviolet (UV) cutoff, such that the lowest momentum mode is $q_{\text{min}}=\frac{\pi}{L} \neq 0$ and the highest momentum mode is $\Lambda$.
In the zero-temperature limit the mutual information (MI) is given by
\begin{align}
  \lim_{\beta \to \infty}  I(A:B,t) = 2\sum_{q_{\text{min}}}^{\Lambda} \mathcal{S}[\lambda_q(t,\beta)].
 \end{align}
As stated in the previous section, the early-time approximation in general necessitates a finite UV cutoff. Nevertheless, here we set $\Lambda \rightarrow \infty$ to obtain explicit scaling law. This is justified since the contribution to the MI from high momentum modes is suppressed, so that the leading term stays unchanged (see main text).

We study two main phases of the early-time dynamics:
\begin{enumerate}[leftmargin=*]
    \item \underline{Phase 1: The time interval for which $\varepsilon_{n}(t) \ll     q_{\text{min}} $}

    In this case, we cannot replace the discrete sum with an integral. Indeed, the discretisation of the sum is performed with the spacing $\Delta q=q_{\rm min}=\pi/L$. So that, to go from the sum to integral representation
\begin{align}
\sum_{q>0}f\left(\frac{\varepsilon_{n}}{q}\right)\to \frac{L}{\pi}\int_{0}^{\Lambda}dq\; f\left(\frac{\varepsilon_{n}}{q}\right),
\end{align}
one requires that $q_{\rm min}\ll \varepsilon_{n}$. 
    Instead, we should start with the exact sum.
   Since $\frac{\varepsilon_{n}(t)}{q} \ll 1$ for all available modes,
   \begin{align}
        \lambda_{q}(t,\beta) &\approx \frac{1}{2}\sqrt{1+\left(\frac{\varepsilon_{n}(t)}{q}\right)^{2}}\approx  \frac{1}{2}\left[1 +\frac{1}{2}\left(\frac{\varepsilon_{n}(t)}{q}\right)^{2}\right].
   \end{align} 
   Then, the entropy function is roughly
   \begin{align}
       \mathcal{S}[\lambda_q] &\approx \left(1 + \frac{\varepsilon^2_{n}}{4q^2}\right)\ln\left(1 + \frac{\varepsilon^2_{n}}{4q^2}\right) - \left(\frac{\varepsilon^2}{4q^2}\right)\ln\left(\frac{\varepsilon^2_{n}}{4q^2}\right)
       \approx \frac{\varepsilon^2_{n}}{4q^2} \left[ 1 - \ln\left(\frac{\varepsilon^2_{n}}{4q^2}\right) \right] + \mathcal{O} \left( \frac{\varepsilon^4_{n}}{q^4}\right).
   \end{align}
   From the boundary condition, we have $q_{N} = N\pi/L$.
   Hence, the above entropic term can be rewritten as $\mathcal{S}[\lambda_{q_N}] \approx \frac{\varepsilon^2_{n} L^2}{4 \pi^2 N^2} \left[ 1+ 2\ln N - 2\ln\left(\frac{\varepsilon_{n} L}{2\pi}\right) \right]$.
   Now, we perform the summation to obtain the MI:
   \begin{align}
       I(A:B,t) &\approx\frac{\varepsilon^2_{n} L^2}{2 \pi^2} \left\{ \left[1-2\ln\left(\frac{\varepsilon_{n} L}{2\pi}\right) \right]\sum_{N=1}^{\infty}   \frac{1}{N^2} + 2 \sum_{N=1}^{\infty}   \frac{\ln N}{N^2}\right\}\\
       &=\frac{\varepsilon^2_{n} L^2}{2 \pi^2} \left\{ \left[1-2\ln\left(\frac{\varepsilon_{n} L}{2\pi}\right) \right] \frac{\pi^2}{6} - 2 \;\zeta'(2)\right\},
   \end{align}
   where we have used the identities $\sum_{N=1}^{\infty} \frac{1}{N^2} = \frac{\pi^2}{6},
       \sum_{N=1}^{\infty} \frac{\ln N}{N^2} = -\zeta'(2)$,
   with the Riemann zeta function $\zeta(\cdot)$ and $\zeta'(\cdot)$ its derivative. Let us set $\tau_{n}=(\varepsilon_{n} L)/\pi$. Rearranging it yields
    \begin{align}
    I(A:B,t)&\approx C_{1}\tau^{2}_{n}+C_{2}\tau^{2}_{n}\ln(\tau_{n}),
    \end{align}
    where $C_{1}=\frac{\pi^{2}}{12}\left[1+2\ln(2)\right]-\zeta'(2)$ and $C_{2}=-\frac{\pi^{2}}{6}.$
   \item \underline{Phase 2: Thermodynamic-Limit Behaviour ($\varepsilon_{n}(t) \gg q_{\text{min}}$)}\\
   To compute the MI, we can use the approximate expression for the symplectic eigenvalues, i.e. 
\begin{align}
    I(A:B, t) 
    = \frac{2L}{\pi} \int_0^{\infty} S \left[\lambda_{q}(t,\beta)\right] \; dq
    &\approx\frac{2L}{\pi} \int_0^{\Lambda} S \left[\frac{1}{2}\sqrt{1+\frac{\varepsilon^2_{n}(t)}{q^2}} \right] \; dq+\frac{2L}{\pi} \int_{\Lambda}^{\infty} S \left[\lambda_{q}(t,\beta) \right] \; dq\\
     &\approx\frac{2L}{\pi} \int_0^{\infty} S \left[\frac{1}{2}\sqrt{1+\frac{\varepsilon^2_{n}(t)}{q^2}} \right] \; dq.
\end{align}
Let $\displaystyle q=\frac{\varepsilon}{\sinh u}$, hence $\displaystyle dq=-\frac{\varepsilon \cosh u}{\sinh^2 u} du$. By performing this change of variables we get
\begin{align}
        I(A:B, t) 
    &= \frac{2L}{\pi} \int_0^\infty S \left[\frac{1}{2} \cosh u \right] \; \frac{\varepsilon \cosh u}{\sinh^2 u} \; du.
\end{align}
Using integration by parts, i.e.
\begin{align}
    a &=S \left[\frac{1}{2} \cosh u \right]= \cosh^2\left(\frac{u}{2}\right) \ln\left[\cosh^2\left(\frac{u}{2}\right)\right] - \sinh^2\left(\frac{u}{2}\right) \ln\left[\sinh^2\left(\frac{u}{2}\right)\right],\\
    da&= \sinh(u) \ln \coth\left(\frac{u}{2}\right), \qquad
    dv= \frac{\cosh u}{\sinh^2 u}, \qquad
    v=-\frac{1}{\sinh u},
\end{align}
where we use the identity $\cosh^2\left(\frac{u}{2}\right) = \frac{\cosh u + 1}{2}$ and  $\sinh^2\left(\frac{u}{2}\right) = \frac{\cosh u - 1}{2}$.
So,
\begin{align}
     I(A:B, t) 
    &= \frac{2L \varepsilon}{\pi}  \left\{-\frac{1}{\sinh u}S \left[\dfrac{1}{2} \cosh u \right] \Bigg|_0^\infty+ \int_{0}^{\infty} \ln \left( \coth \frac{u}{2} \right) du\right \}.
\end{align}

For $u \to \infty$, $S \left[\dfrac{1}{2} \cosh u \right]\approx \ln \left(\dfrac{1}{2} \cosh u \right) \approx \ln \left(\frac{e^u}{4}\right) \approx u$. 
So, we have $S \left[\dfrac{1}{2} \cosh u \right]/\sinh u\approx 2ue^{-u} \to 0.$

For $u \to 0$,
\begin{align}
       S \left[\dfrac{1}{2} \cosh u \right]&= \cosh^2\left(\frac{u}{2}\right) \ln\left[\cosh^2\left(\frac{u}{2}\right)\right] - \sinh^2\left(\frac{u}{2}\right) \ln\left[\sinh^2\left(\frac{u}{2}\right)\right] \approx 0-\left(\frac{u}{2}\right)^2  \;2 \ln u,
\end{align}
as $\cosh x \to 1$ and $\sinh x \to x$.
So, for $u \to 0^+$, $- S \left[\dfrac{1}{2} \cosh u \right]/\sinh u\approx \dfrac{u\ln u}{2} \to 0$.

Since the first term vanishes at both limits, the MI reduces to 
\begin{align}
      I(A:B, t) 
    &= \frac{2L \varepsilon}{\pi}  \int_{0}^{\infty} \ln \left( \coth \frac{u}{2} \right) du =\frac{\pi L \varepsilon }{2},
\end{align}
where we used the identity $\int_{0}^{\infty} \ln \left( \coth \frac{u}{2} \right) du=\frac{\pi^2}{4}$.

\end{enumerate}
\subsection{Crossover from ultra-early time to early-time regime \label{subsec: early time breakdown}}
We now provide an estimate of the time at which the crossover from one regime to the other happens for different system sizes, in the case of the sudden quench and the linear protocol. 
We estimate the regime in which $\tau_{n}=\varepsilon_{n}(t)/q_{\min}\ll 1$, with
\(q_{\min}=\pi/L\). For the linear protocol ($n=1$), the time-dependent parameter is $\varepsilon_{1}(t)=\frac{\omega \dot J(0)}{v}t^2$, and the validity time $t^{*}_{1}=\left(\frac{3}{2\omega |\dot J(0)|}\right)^{1/3}$. Thus
$\tau_{1}=\frac{\varepsilon_{1}(t)}{q_{\min}}
=
\frac{\omega \dot J(0)L}{v\pi}t^2
=
\frac{3L}{2v\pi t^{*}_{1}}\left(\frac{t}{t^{*}_{1}}\right)^{2}
\ll 1$. Or equivalently
$t^{(1)}_{\rm cut}= \left(\frac{2v\pi t^{*}_{1}}{3L}\right)^{1/2}t^{*}_{1}.$ The time $t^{(n)}_{\rm cut}$ corresponds to the value of $t$ for which the approximation $\tau_{n}\ll 1$ breaks down. For a sudden quench, $\varepsilon_{0}(t)=\frac{2\omega J_f t}{v}$,
$t^{*}_{0}=(4\omega J_f)^{-1/2}$. The condition $\tau_{0}=\varepsilon_{0}(t)/q_{\min}\ll 1$, gives $t^{(0)}_{\rm cut}= 
\frac{v\pi}{\sqrt{\omega J_f}\,L}\,t^{*}_{0}$. 
\begin{table}[h]
\centering
\begin{tabular}{c|c|c|c|c}
\hline
System size 
& \multicolumn{2}{c|}{linear protocol}
& \multicolumn{2}{c}{Sudden quench}
\\
\cline{2-5}
$L$
& $t^{(1)}_{\rm cut}/t^{*}_{1}$
& $t^{(1)}_{\mathrm{cut}}~(\mathrm{s})$
& $t^{(0)}_{\rm cut}/t^{*}_{0}$
& $t^{(0)}_{\rm cut} (\mathrm{s})$
\\
\hline
$10^3\xi$
& $2.054\times 10^{-1}$
& $3.16\times 10^{-4}$
& $6.41\times 10^{-2}$
& $t \ll 5.00\times 10^{-5}$
\\
$10^4\xi$
& $6.50\times 10^{-2}$
& $1.00\times 10^{-4}$
& $6.41\times 10^{-3}$
& $t \ll 5.00\times 10^{-6}$
\\
$10^5\xi$
& $2.05\times 10^{-2}$
& $3.16\times 10^{-5}$
& $6.41\times 10^{-4}$
& $t \ll 5.00\times 10^{-7}$
\\
\hline
\end{tabular}
\caption{Validity window for the condition $\tau_{n}=\varepsilon_n(t)/q_{\min}\ll 1$ for linear and sudden quenches. Here $q_{\min}=\pi/L$, $L$ is measured in units of the healing length $\xi$, and the numerical values are obtained using $J_f=2\pi\times 5~\text{Hz}$, $t_f=10^{-3}~\mathrm{s}$, mean density $n_0=75\times 10^6~\mathrm{m^{-1}}$, frequency $\omega_{\perp}=2\pi\times 2{\rm kHz}$, Luttinger parameter $K=55.67$, velocity $v\simeq 3.1\times 10^{-3}\ \mathrm{m\,s^{-1}}$,
healing length $\xi\simeq 2.36\times 10^{-7}\ \mathrm{m}$. For these parameters we obtain $t^{*}_{0}=7.80\times 10^{-4}s$ and $t^{*}_{1}=1.5\times 10^{-3}s$.}
\label{tab:validity-window-epsilon-qmin}
\end{table}

\section{The Ermakov factors and symplectic eigenvalues after protocol saturation \label{general_quench_solution}}
We consider a finite time ramp, and show that at saturation regime, the full entanglement dynamics can be characterised by three parameter sets: the final saturation value of the protocol,
the values of the Ermakov factors and their time derivatives at the onset of saturation. 
In Appendix~\ref{exact_solution}, we derive the closed-form Ermakov solution of Eq.~\eqref{general_solution_ermakov}, parametrised by $M_q$ and $N_q$.
Then, we use the solution to determine the symplectic eigenvalues in Eqs.~\eqref{exact_symplectic}--\eqref{exact_symplectic_2} (Appendix~\ref{symplectic_eigenvalues}), which are parametrised by $D_q$ and $E_q$.
We further establish the positivity properties and lower bounds of $M_q,N_q,D_q,E_q$ (Appendix~\ref{subsec: Properties of Dq}) that are used throughout the long-time average calculation, i.e. Sec.~\ref{Sec: long time MI} -- Sec.~\ref{expression_Renyi2_arbitrary}.
\subsection{Exact solution to the Ermakov equation after saturation \label{exact_solution}}
\begin{theorem}\label{Th:ermakov solution after quench}
Let $v,q,\omega,J_f>0$, and suppose that $J(t)=J_f$ for all $t\geq t_f$. Define
\begin{align}
    \Omega_{q0}=vq,\qquad
    \Omega_{qf}=\sqrt{(vq)^2+4\omega J_f},
    \qquad
    \zeta_q=\frac{\Omega_{q0}}{\Omega_{qf}}.
\end{align}
For each $q$, let $\gamma_q(t)$ solve $\ddot{\gamma}_q(t)
    +
    \Omega_{qf}^{2}\gamma_q(t)
    =
    \frac{\Omega_{q0}^{2}}{\gamma_q^3(t)}$,
with $\gamma_q(0)=1$ and $\dot{\gamma}_q(0)=0$. Then, for $t\geq t_f$,
\begin{empheq}[box=\fbox]{align}
   \gamma_q(t)
    &=\sqrt{M_q+N_q-2N_q\sin^2\theta_q}, \label{solution_general_ermakov_appendix}\\
N_q&=\sqrt{M_q^2 - \zeta_q^2},\qquad
M_q=\frac{1}{2} \left[ \gamma_q^2(t_{f})+ \left(\frac{\dot \gamma_q(t_{f})}{\Omega_{qf}}\right)^2+ \left(\frac{\zeta_q}{\gamma_q(t_{f})} \right)^2\right],\\
\theta_q&=\Omega_{qf}(t-t_{f})-\frac{\vartheta_q}{2},\qquad
\vartheta_q
=\arctan \left(
\frac{
2\,\Omega_{qf}\,\gamma_q(t_{f})\,\dot\gamma_q(t_{f})
}{
\Omega_{qf}^{2}\gamma_q^{2}(t_{f})
-\dot\gamma_q^{2}(t_{f})
-\dfrac{\Omega_{q0}^{2}}{\gamma_q^{2}(t_{f})}
}
\right).
\end{empheq}
\end{theorem}
\begin{proof}

In the calculation, we drop the subscript $q$ for simplicity.
We start by noting that for $t > t_{f}$, the Ermakov equation simplifies to $\ddot{\gamma}+\Omega_{qf}^2\gamma=\dfrac{\Omega_{q0}^2}{\gamma^{3}}.
$
 We use the Pinney solution of the Ermakov equation, i.e.
\begin{align}
    \gamma(t)=\sqrt{u^2(s)+\frac{\Omega_{q0}^2}{W^2(u,v)}v^2(s)} \ , \label{Eq: post-linear protocol pinney solution}
\end{align}
where $s=t-t_{f}$, $u$ and $v$ are two independent solutions of the corresponding homogeneous equation
\begin{align}
    \ddot{\gamma} + \Omega_{qf}^2 \gamma=0 \ .\label{eq: post-linear protocol homogeneous equation}
\end{align}
The initial conditions for $u$ and $v$ now depend on the control protocol at the final ramp time $t_{f}$
    \begin{align} 
    u(0)  = \gamma(t_{f}), \quad  \dot{u}(0) = \dot{\gamma}(t_{f}), \quad
    v(0)  = 0,  \quad \dot{v}(0)  \neq 0. \label{Eq: post-linear protocol u and v initial condition}
    \end{align}
We can verify that
\begin{align}
    u(s) &= u(0)\,\cos\left(\Omega_{qf}\, s\right) + \frac{\dot u(0)}{\Omega_{qf}}\,\sin\left(\Omega_{qf}\, s\right) ,\qquad
    v(s) =\frac{1}{\Omega_{qf}} \sin\left(\Omega_{qf}\, s\right), \label{Eq: u and v s}
\end{align}
satisfy both Eq.~\eqref{eq: post-linear protocol homogeneous equation} and Eq.~\eqref{Eq: post-linear protocol u and v initial condition}.
Moreover, we check linear independence of $u$ and $v$ by computing their Wronskian, i.e.
\begin{align}
    W[u,v](s) &= u(s)\,\dot v(s) - \dot u(s)\,v(s)=u(0) \neq 0.
\end{align}
Thus, if $u(0)= \gamma(t_{f})\neq 0$, $u$ and $v$ are linearly independent solutions. 
To ensure $ W[u,v](s)=1$, we can set $\dot v(0)=1/u(0)$.
In particular, at time $s=0$, we have $W[u,v](0) = u(0)\,\dot v(0)=1$.
For time $s>0$, $W$ remains unchanged by Abel's identity as there is no first-derivative term. 
Now, substituting Eq.~\eqref{Eq: u and v s} into Eq.~\eqref{Eq: post-linear protocol pinney solution} yields
\begin{align}
    \gamma(t) &=\sqrt{
\left[ \gamma(t_{f}) \cos\left[\Omega_{qf}(t-t_{f})\right] + \frac{\dot \gamma(t_{f})}{\Omega_{qf}} \sin\left[\Omega_{qf}(t-t_{f})\right]  \right]^2
+ \left(\frac{\Omega_{q0}}{\Omega_{qf}\gamma(t_{f})}\sin\left[\Omega_{qf}(t-t_{f})\right] \right)^2
} ,
\end{align}
for time $t \geq t_{f}$. 
For simplicity, we define $A= \gamma(t_{f})$, $B=\dfrac{\dot \gamma(t_{f})}{\Omega_{qf}}$, $C=\dfrac{\zeta_q}{\gamma(t_f)}$, $t'=\Omega_{qf}(t-t_{f})$.
Therefore, it becomes
 \begin{align}
     \gamma(t)&=\sqrt{ (A \cos t' + B \sin t')^2 + C^2 \sin^2 t'}\\
     &=\sqrt{\frac{A^2 + B^2 + C^2}{2} + \left( \frac{A^2 - B^2 - C^2}{2} \right) \cos 2t' + AB \sin 2t'}\\
     &=\sqrt{\frac{A^2 + B^2 + C^2}{2}+ \left(\sqrt{\left( \frac{A^2 - B^2 - C^2}{2} \right)^2 + (AB)^2}\right)\cos(2t' - \vartheta)},
 \end{align}
 where we used the identity $X\cos(2t)+Y\sin(2t)=\sqrt{X^{2}+Y^{2}}\cos(2t-\vartheta)$ with $ \displaystyle \tan \vartheta=\frac{Y}{X}=\frac{2 AB}{A^{2}-B^{2}-C^{2}}$.
 Finally, by setting $M=\dfrac{A^2 + B^2 + C^2}{2}$ and $N=\displaystyle\sqrt{\left( \frac{A^2 - B^2 - C^2}{2} \right)^2 + (AB)^2}$, we obtain
 \begin{align}
      \gamma(t)&=\sqrt{M + N \;\cos(2t' - \vartheta)}=\sqrt{M+N-2N\sin^2(t'-\vartheta/2)} \label{Eq: Gamma in M,N, t'}
 \end{align}
Expressing Eq.~\eqref{Eq: Gamma in M,N, t'} in terms of $\gamma(t_f), \dot{\gamma}(t_f),\zeta_q, \Omega_{qf}$ recovers Theorem.~\ref{Th:ermakov solution after quench}.

\end{proof}

\subsection{The expression for the symplectic eigenvalues \label{symplectic_eigenvalues}}
In this subsection, from the knowledge of $\gamma_{q}(t\geq t_{f})$ for the finite-ramp, we derive an exact expression for the symplectic eigenvalue given by Eq.~\eqref{eq: thermal_subsys_eigs}, that is expressed in terms of the Ermakov factor Eq.~\eqref{eq:ermakov_com_func}. From the previous section, we have the solution to the Ermakov equation for $t\geq t_{f}$ given by Eq.~\eqref{solution_general_ermakov_appendix}, that we recall to be $\gamma_q(t)=\sqrt{M_q + N_q \;\cos(2\theta_q)}$.
%
%
Observe that $ N_q\cos(2\theta_q)= \gamma_q^2(t)-M_q $,
so that the derivative can be reduced to
\begin{align}
        [\dot{\gamma}_q(t) ]^2&= \frac{\Omega_{qf}^2 N_q^2 \sin^2 (2\theta_q)}{ \gamma_q^2(t)}=\frac{\Omega_{qf}^2N_q^2}{ \gamma_q^2(t)} - \frac{\Omega_{qf}^2}{ \gamma_q^2(t)} \left[\gamma_q^2(t)-M_q \right]^2=\frac{\Omega_{qf}^2}{ \gamma_q^2(t)} \left[N_q^2- \gamma_q^4(t)+2M_q \gamma_q^2(t)-M_q^2 \right].
\end{align}
Therefore, using the definition of $\mathcal{F}$ given by Eq.~\eqref{eq:ermakov_com_func}
\begin{align}
    \mathcal{F} &= \frac{1}{4} \left[ \frac{(\dot{\gamma}_q)^2 }{\Omega_{q0}^2} + \gamma_q^2 - 2 + \frac{1}{\gamma_q^2} \right]=\frac{1}{4} \left[ \frac{2M_q}{\zeta_q^2}  -2 +\gamma_q^2\left(1-\frac{1}{\zeta_q^2} \right)  \right],
\end{align}
where we use the relation $M_q^2-N_q^2=A^2C^2=\frac{\Omega^{2}_{q0}}{\Omega_{qf}^2}\equiv \zeta_q^2$. Next,
\begin{align}
    1+\mathcal{F}&=\frac{1}{4}\Big[2+M_{q}\left(1+\frac{1}{\zeta^{2}_{q}}\right)+N_{q}\left(1-\frac{1}{\zeta^{2}_{q}}\right)-2N_{q}\left(1-\frac{1}{\zeta^{2}_{q}}\right)\sin^{2}(\theta_{q})\Big].
\end{align}
In its final form 
\begin{align}
    \lambda_{q}^A(t)&= C_{q}(\beta) \sqrt{\frac{2+M_q+N_q}{4}+\frac{M_q-N_q}{4 \zeta_q^2} + \frac{N_q}{2}\left(\frac{1}{\zeta_q^2} -1\right) \sin^2 \theta_q} =C_{q}(\beta)\sqrt{D_q +E_q \sin^2 \theta_q},\label{eq: symplectic for general protocol}
\end{align}
with $C_q(\beta)$ defined in Eq.~\eqref{eq:C_q}, and $D_q,E_q$ in and above Eq.~\eqref{value_Dq}.

\subsection{Properties of $M_{q},N_{q},D_{q}$ and $E_{q}$. \label{subsec: Properties of Dq}}
We conclude this appendix by collecting the inequalities obeyed by the coefficients parametrizing the post-ramp symplectic eigenvalues.
These inequalities are used in the next two sections to satisfy the assumptions required by the integral theorems.
The inequality $M_{q}\geq \zeta_{q}$ is useful to prove that $D_{q}\geq 1$. The lower bound $N_{q}\geq 0$ is useful to prove $E_{q}\geq 0$. The properties $D_{q}\geq 1$ and $E_{q}\geq 0$ are required to derive the long-time average of of all information-theoretic quantities of interest.

\begin{theorem}
\label{Th: properties of M_q,N_q,D_q,E_q}
Let $v,q,\omega,J_f>0$, and suppose that $J(t)=J_f$ for all $t\geq t_f$. 
Let $\Omega_{q0}$, $\Omega_{qf}$, $\zeta_q$, $M_q$, and $N_q$ be defined as in Theorem~\ref{Th:ermakov solution after quench}. 
Let $D_q$ and $E_q$ denote the coefficients entering the post-ramp symplectic eigenvalue, as defined in and above Eq.~\eqref{value_Dq}.
Then, for every mode $q$, 
\begin{align}
    M_q \geq\zeta_q \Longrightarrow    D_q\geq 1 \quad, 
    \quad N_q\geq 0 \Longrightarrow 
    E_q\geq 0.
\end{align}
\end{theorem}

\begin{proof}
We first prove the bounds on $M_q$ and $N_q$. From the definition of $M_q$,
\begin{align}
     M_q   &=     \frac{1}{2}     \left[
     \gamma_q^2(t_f)
     +
     \left(\frac{\dot \gamma_q(t_f)}{\Omega_{qf}}\right)^2
     +
     \left(\frac{\zeta_q}{\gamma_q(t_f)}\right)^2
     \right] \geq
     \frac{1}{2}
     \left[
     \gamma_q^2(t_f)
     +
     \left(\frac{\zeta_q}{\gamma_q(t_f)}\right)^2
     \right].
\end{align}
By the arithmetic--geometric mean inequality, $ M_q\geq
     \sqrt{
     \gamma_q^2(t_f)
     \left(
     \frac{\zeta_q}{\gamma_q(t_f)}
     \right)^2
     }
     =
     \zeta_q$.
Moreover, since $ N_q=\sqrt{M_q^2-\zeta_q^2}$, it follows immediately that $N_q\geq 0$.
Next, we prove the lower bound on $D_q$.
As $M_q\geq \zeta_q$, we may write $M_q=\zeta_q\cosh r_q$, and $N_q=\sqrt{M_q^2-\zeta_q^2}=\zeta_q\sinh r_q$,
for some $r_q\geq 0$.
Substituting this parametrization into Eq.~\eqref{value_Dq}, we obtain
\begin{align}
    D_q
    &=
    \frac{1}{2}
    +
    \frac{(1+\zeta_q^2)\cosh r_q
    +
    (\zeta_q^2-1)\sinh r_q}{4\zeta_q}
    =
    \frac{1}{2}
    +
    \frac{\zeta_q^2 e^{r_q}+e^{-r_q}}{4\zeta_q}=    1+   \frac{    \left(  \zeta_q e^{r_q/2}-
    e^{-r_q/2}
    \right)^2
    }{4\zeta_q}
    \geq 1.
\end{align}
Finally, note that from before Eq.~\eqref{value_Dq}, we have $E_q=\displaystyle \frac{N_q}{2} \left(\frac{1}{\zeta_q^2}-1 \right)$.
Since $0<\zeta_q\leq 1$ and $N_q\geq 0$, we have
$E_q\geq 0$.

\end{proof}

\section{Long-time averaged mutual information (MI) \label{Sec: long time MI}}
In this section, we present exact results for the time-averaged (MI) at finite temperature after a general protocol.
The sudden quench scenario and the zero-temperature limit correspond to special cases of this general result.
After proving the core integral formula of Theorem~\ref{th: integral formula for MI}, we proceed to establish the general-protocol MI expressions of Table~\ref{Table_general_quench} (Appendix~\ref{sec: MI formula for general protocol at finite temperature}) and the sudden-quench MI expressions of Table~\ref{Results_sudden_quench_table} (Appendix~\ref{sec: MI zero temperature limit and sudden quench}).
Next, in Appendix~\ref{scaling_arbitrary_protocol}, we provide the calculation of the zero-temperature thermodynamic-limit scaling coefficient of MI quoted in Sec.~\ref{Sec: TDL scaling}.
Finally, we prove that the long-time-averaged MI decreases monotonically with temperature for an arbitrary protocol, the statement invoked in Sec.~\ref{subsec:ln_dynamics} (Appendix~\ref{subsec: MI decays monotonically}).

\subsection{Proof of the core integral formula for computing averaged MI}
\begin{theorem} \label{th: integral formula for MI}
Let $A \geq \dfrac{1}{2}$ and $B \geq 0$. Define the function $S[x]\equiv \left(x+\frac{1}{2}\right)\ln\left(x+\frac{1}{2}\right)-\left(x-\frac{1}{2}\right)\ln\left(x-\frac{1}{2}\right)$. Then
    \begin{align}
        \int_0^{\pi/2} S\left( A\sqrt{1+B \sin^2 u}\right) du &=  \pi A \sqrt{1+B}\left(F \left[\arcsin\left(\frac{1}{2A}\right), \frac{1}{\sqrt{1+B}}\right] -E\left[\arcsin\left(\frac{1}{2A}\right), \frac{1}{\sqrt{1+B}}\right] \right)  \nonumber\\
    &\quad +  \frac{\pi}{2} \left[ 1 + \ln\left( \frac{ \sqrt{4A^2(1+B) - 1} + \sqrt{4A^2 - 1} }{ 4 } \right) \right] \ .
    \end{align}
Here, $\displaystyle F(\phi, k), E(\phi, k)$ denotes the incomplete elliptic integral of the first and second kind, respectively
\begin{align}
      F(\phi, k) &= \int_0^{\phi} \frac{d\theta}{\sqrt{1 - k^2 \sin^2 \theta}} \;\;, \;\;
      E(\phi, k) = \int_0^{\phi} \sqrt{1 - k^2 \sin^2 \theta}\, d\theta \ .
\end{align}
\end{theorem}

\begin{proof}
For simplicity, we define $\alpha(B,u)=\sqrt{1+B \sin^2 u}$ and 
\begin{align}
    G(B) = \int_0^{\pi/2} S\left[ A\alpha(B,u)\right] du \ . \label{eq: G(B) integral form for MI}
\end{align}
For proper definition of the integrals below, we assume $A>1/2$ and $B>0$. By the Leibniz integral rule, we have
\begin{align}
    \frac{dG}{dB} &= \int_0^{\pi/2} \; S'(A\alpha) \frac{\partial(A\alpha)}{\partial B}du= \int_0^{\pi/2} \ln\left(\frac{2A \alpha + 1}{2A \alpha - 1}\right) \frac{A \sin^2 u}{2\alpha} du.
\end{align}
where we use the identity $\displaystyle \frac{dS}{dx} = \ln\left( \frac{2x + 1}{2x - 1} \right)$ in the last line.
Next, observe that
\begin{align}
\int \frac{dt}{\alpha^2 - t^2}
=
\frac{1}{2\alpha}
\ln\left|
\frac{\alpha + t}{\alpha - t}
\right|  \longrightarrow \int_0^{1/(2A)} \frac{dt}{\alpha^2 - t^2}  = \frac{1}{2\alpha} \ln\left(\frac{2A \alpha + 1}{2A \alpha - 1}\right),
\end{align}
where we assume $2A\alpha>1$, which is verified since $B>0$ and $A>1/2$. Thus, the above expression can be rewritten as
\begin{align}
       \frac{dG}{dB}
    &=A\int_0^{1/(2A)}  \int_0^{\pi/2} \frac{\sin^2 u}{1+B \sin^2u - t^2} \;du \; dt,
\end{align}
where we use the Fubini–Tonelli theorem to swap the order of integration, as the integrand is a nonnegative bounded function.
Using the identity:
\begin{align}
    \int_0^{\pi/2} \frac{\sin^2 u}{c^2 + B\sin^2 u} du = \frac{\pi}{2B} \left( 1 - \frac{c}{\sqrt{c^2+B}} \right) \;\; \text{ for } c > 0 \ ,
\end{align}
It reduces to
\begin{align}
    \frac{dG}{dB} 
    &=\frac{\pi A}{2B}  \int_0^{1/(2A)}  \left(1-\frac{\sqrt{1-t^2}}{\sqrt{1-t^2+B}}\right) \;dt =\frac{\pi A}{2B}   \left(\frac{1}{2A}- \int_0^{\phi} \frac{\cos^2 \psi}{\sqrt{1+B-\sin^2 \psi}} \; d\psi\right),
\end{align}
where we define $t=\sin \psi$ and  $\phi = \arcsin\left[1/(2A)\right]$. Furthermore, let $k=1/\sqrt{1+B}$, then it becomes
\begin{align}
      \frac{dG}{dB}
    &=\frac{\pi A}{2B}   \left(\frac{1}{2A}- k\int_0^{\phi} \frac{\cos^2 \psi}{\sqrt{1-k^2\sin^2 \psi}} \; d\psi\right) \\
    &=\frac{\pi A}{2B}   \left(\frac{1}{2A}- k\int_0^{\phi} \frac{1}{k^2}\frac{k^2- k^2 \sin^2 \psi-1+1}{\sqrt{1-k^2 \sin^2 \psi}} \; d\psi\right)\\
    &=\frac{\pi A}{2B}   \left(\frac{1}{2A}- \frac{1}{k}\int_0^{\phi}\left[\sqrt{1-k^2 \sin^2 \psi}+\frac{k^2-1}{\sqrt{1-k^2 \sin^2 \psi}}\right] \; d\psi\right) \ .
\end{align}
By the definition of incomplete elliptic integrals of the first $E(\phi,k)$ and second kind $F(\phi,k)$, it is
\begin{align}
           \frac{dG}{dB} &= \frac{\pi}{4B} + \frac{\pi A}{2} \left[ \frac{1}{\sqrt{1+B}} \, F\left(\phi, \frac{1}{\sqrt{1+B}}\right) -\frac{\sqrt{1+B}}{B}E\left(\phi, \frac{1}{\sqrt{1+B}}\right) \right] \ , \label{eq: dG/dB for MI}
\end{align}
We now use the following identities
\begin{align}
    &\quad\int \left[\frac{1}{\sqrt{1+B}}F\left(\phi, \frac{1}{\sqrt{1+B}}\right) - \frac{\sqrt{1+B}}{B}E\left(\phi, \frac{1}{\sqrt{1+B}}\right) \right]  \; dB \nonumber\\
    &=2\sqrt{1+B}\left[F\left(\phi, \frac{1}{\sqrt{1+B}}\right) -E\left(\phi, \frac{1}{\sqrt{1+B}}\right) \right] + 2 \sin \phi  \operatorname{arctanh} \left(\frac{\sqrt{B+\cos^2\phi}}{\cos\phi}\right) \\
    &= 2\sqrt{1+B}\left[F\left(\phi, \frac{1}{\sqrt{1+B}}\right) -E\left(\phi, \frac{1}{\sqrt{1+B}}\right) \right] + \frac{1}{A}  \operatorname{arctanh} \left(\frac{\sqrt{4A^2(1+B) - 1}}{\sqrt{4A^2 - 1}}\right) \ , \label{eq: integral of elliptic for MI}
\end{align}
where we use the fact that $ \displaystyle\sin \phi =\frac{1}{2A} , \;\; \cos \phi  = \frac{\sqrt{4A^2 - 1}}{2A}$ in the last line.
Combining Eq.~\eqref{eq: dG/dB for MI} and Eq.~\eqref{eq: integral of elliptic for MI},  we have
\begin{align}
    G(B)
    &=\frac{\pi}{4} \ln B +\pi A \sqrt{1+B}\left(F \left[\phi, \frac{1}{\sqrt{1+B}}\right] -E\left[\phi, \frac{1}{\sqrt{1+B}}\right] \right) + \frac{\pi}{2}\operatorname{arctanh} \left(\frac{\sqrt{4A^2(1+B) - 1}}{\sqrt{4A^2 - 1}}\right) +C_1,
\end{align}
where $C_1$ is the integration constant.
Note that for $x/y > 1$, we have $\operatorname{arctanh}\left(\frac{x}{y}\right) = \ln(x+y)  -\frac{1}{2}  \ln(x^2-y^2) +\frac{i \pi}{2}
$.
Hence, 
\begin{align}
      G(B)
    &=\frac{\pi}{4} \ln B +\pi A \sqrt{1+B}\left(F \left[\phi, \frac{1}{\sqrt{1+B}}\right] -E\left[\phi, \frac{1}{\sqrt{1+B}}\right] \right) \nonumber\\
    &\quad + \frac{\pi}{2} \left[  \ln \left( \sqrt{4A^2(1+B) - 1} + \sqrt{4A^2 - 1} \right) - \frac{1}{2} \ln(4A^2 B)\right] +C_1\\
    &=\pi A \sqrt{1+B}\left(F \left[\phi, \frac{1}{\sqrt{1+B}}\right] -E\left[\phi, \frac{1}{\sqrt{1+B}}\right] \right)  \nonumber\\
    &\quad + \frac{\pi}{2} \left[  \ln \left( \sqrt{4A^2(1+B) - 1} + \sqrt{4A^2 - 1} \right) \right] -\frac{\pi}{2} \ln(2A) +C_1, \label{eq: G(B) for MI with integration constant}
\end{align}
where we absorb the constant phase $i\pi/2$ into  $C_1$.
Finally, the remaining task is to compute the integration constant $C_1$.
Using Eq.~\eqref{eq: G(B) integral form for MI} and Eq.~\eqref{eq: G(B) for MI with integration constant}, we have
\begin{align}
     G(0)&=\int_0^{\pi/2} S(A) \, du = \frac{\pi}{2} S(A) \ ,\\ 
     G(0)&=\pi A\left(F \left[\phi, 1\right] -E\left[\phi,1\right] \right)  + \frac{\pi}{2} \left[  \ln \left( 2\sqrt{4A^2 - 1} \right) \right] -\frac{\pi}{2} \ln(2A) +C_1  \ .
\end{align}
Equating  both expressions gives $C_1 
    =\frac{\pi}{2} \left[ 1 + \ln\left(\frac{A}{2}\right) \right]$.
Combining all results, we get
\begin{align}
       G(B) 
  &=\pi A \sqrt{1+B}\left(F \left[\arcsin\left(\frac{1}{2A}\right), \frac{1}{\sqrt{1+B}}\right] -E\left[\arcsin\left(\frac{1}{2A}\right), \frac{1}{\sqrt{1+B}}\right] \right)  \nonumber\\
    &\quad +  \frac{\pi}{2} \left[ 1 + \ln\left( \frac{ \sqrt{4A^2(1+B) - 1} + \sqrt{4A^2 - 1} }{ 4 } \right) \right] \ ,
\end{align}
where we substitute back $\phi = \arcsin\left[1/(2A)\right]$. The case $A=1/2$ and/or $B=0$ follows by continuity.
\end{proof}

\subsection{Averaged MI formula for a general protocol at finite and zero temperature \label{sec: MI formula for general protocol at finite temperature}}
In this section, we derive the expression for the MI following a finite-ramp protocol at finite temperature.
In the next section, we will recover the expression for the MI after a general protocol in the zero-temperature limit, as well as the sudden quench at finite temperature and its zero-temperature limit, which are special cases of the present formulation. First, recall the expression for the MI at time $t$ after a finite-ramp protocol:
      \begin{align}
           I(A:B,t) &= 2\sum_{q>0}\left\{\mathcal{S}[\lambda_q(t,\beta)]-\mathcal{S}[\lambda_q^+(0)]\right\},\nonumber\\
               \mathcal{S}[x] &\equiv \left(x+\frac{1}{2}\right)\ln\left(x+\frac{1}{2}\right)-\left(x-\frac{1}{2}\right)\ln\left(x-\frac{1}{2}\right),\nonumber\\
               \lambda_{q}(t,\beta)&= \frac{1}{2}\sqrt{1+\mathcal{F}[\gamma_q(t)]}\coth\left(\frac{\beta \hbar \Omega_{q0}}{2}\right),\quad
                  \ \lambda_q^{+}(0) = \frac{1}{2}\coth\left(\frac{\beta\hbar\Omega_{q0}}{2}\right)=: C_q, \qquad C_q'= C_q\sqrt{D_q}, \label{eq: Cq'}
      \end{align}
   We are interested in computing the asymptotic long-time average value of the MI, which is defined as
   \begin{align}
    \overline{I(A:B)}= \lim_{t_1 \to \infty} \frac{1}{t_1} \int_0^{t_1}  I(A:B, u)  \; du,
    \end{align}
where $u$ denotes the time parameter.
From Eq.~\eqref{eq: symplectic for general protocol}, the symplectic eigenvalue $\lambda_{q}(t,\beta)$ after a general protocol is given by
\begin{align}
    \lambda_{q}^A(t)&=C_q'\times\sqrt{1 +\frac{E_q}{D_{q}} \sin^2 \theta_q}.
\end{align}
Since $\lambda_{q}(t,\beta)$ is a periodic function, the asymptotic long-time average value of the MI can be simplified to
\begin{align}
   \overline{I(A:B)} 
    &= 2 \sum_{q >0}  \left\{ \frac{2}{\pi}\int_0^{\pi/2}  \left[\mathcal{S}\left[\lambda_q(u)\right]  - \mathcal{S}[\lambda_q^+(0)] \right] \; du \right\} \label{eq: effective time averaged MI}=\frac{4}{\pi} \sum_{q >0}  \left\{  \left[\int_0^{\pi/2} \mathcal{S}\left[\lambda_q(u)\right] \; du \right]  - \frac{\pi}{2} \mathcal{S}[\lambda_q^+(0)]   \right\} \ .
\end{align}
\begin{enumerate}
    \item Observe that the first term is of the form:
    \begin{align}
        \int_0^{\pi/2} \mathcal{S}\left[\lambda_q(u)\right] \; du &=   \int_0^{\pi/2} \mathcal{S}\left[ C_q' \times \sqrt{1 +\frac{E_q}{D_q} \sin^2 u} \right] \; du  \ .
    \end{align}
    From Theorem.~\ref{Th: properties of M_q,N_q,D_q,E_q}, we know that $D_q \geq 1$. 
    So, $\displaystyle C_q' \geq  C_q \geq\frac{1}{2}$.
    Combining this observation with Theorem.~\ref{th: integral formula for MI}, we obtain
    \begin{align}
         \int_0^{\pi/2} \mathcal{S}\left[\lambda_q(u)\right] \; du &=  \pi C_q' \sqrt{1+\frac{E_q}{D_q}}\left(F \left[\arcsin\left(\frac{1}{2C_q'}\right), \frac{1}{\sqrt{1+\frac{E_q}{D_q}}}\right] -E\left[\arcsin\left(\frac{1}{2C_q'}\right), \frac{1}{\sqrt{1+\frac{E_q}{D_q}}}\right] \right)  \nonumber\\
    &\quad +  \frac{\pi}{2} \left[ 1 + \ln\left( \frac{ \sqrt{4C_q'^2(1+\frac{E_q}{D_q}) - 1} + \sqrt{4C_q'^2 - 1} }{ 4 } \right) \right] \ , 
    \end{align}
    \item 
    The second term is a time-independent constant, which is given by $\mathcal{S}[\lambda_q^+(0)]=  C_q\ln\left(\frac{2C_q+1}{2C_q-1}\right)+ \frac{1}{2}\ln\left(\frac{4C_q^2 - 1}{4}\right)$.
\end{enumerate}
Combining these two terms, we get
\begin{align}
     \overline{I(A:B)} 
    &=4\sum_{q >0}  \left\{ C_q' \sqrt{1+\frac{E_q}{D_q}}\left(F \left[\arcsin\left(\frac{1}{2C_q'}\right), \sqrt{\frac{D_q}{D_q+E_q}}\right] -E\left[\arcsin\left(\frac{1}{2C_q'}\right), \sqrt{\frac{D_q}{D_q+E_q}}\right] \right) \right. \nonumber\\
    &\qquad  \qquad \left. + \frac{1}{2} \left[ 1 + \ln\left( \frac{ \sqrt{4(C_q')^2(1+E_q/D_q) - 1} + \sqrt{4(C_q')^2 - 1} }{ 2\sqrt{4C_q^2 - 1} } \right) \right]  -\frac{1}{2}C_q\ln\left(\frac{2C_q+1}{2C_q-1}\right)  \right\}  \ .\label{eq: MI for general protocol at finite T}
\end{align}
Thus, we prove the first MI expression stated in the Table.~\ref{Table_general_quench}.

\textbf{Zero-temperature limit---}
For the finite-ramp protocol in the zero-temperature limit, we have $\lim_{\beta \to \infty} C_q = \frac{1}{2}$ and $\lim_{\beta \to \infty} C_q'= \frac{1}{2}\sqrt{D_q}$, where we recall Eq.~\eqref{eq: Cq'}.
So, Eq.~\eqref{eq: MI for general protocol at finite T} becomes
\begin{align*}
    \lim_{C_q \to 1/2}    \overline{I(A:B, t)} &\to 2\sum_{q >0} \left\{ \sqrt{D_q + E_q}\left(F \left[\arcsin\left(\frac{1}{\sqrt{D_q}}\right), \sqrt{\frac{D_q}{D_q+E_q}}\right] - E\left[\arcsin\left(\frac{1}{\sqrt{D_q}}\right), \sqrt{\frac{D_q}{D_q+E_q}}\right] \right) \right. \nonumber\\
    &\qquad \left. + 1 + \ln\left( \sqrt{D_q + E_q - 1} + \sqrt{D_q - 1} \right)  \right\} -\lim_{C_q \to 1/2} \left\{\ln \left( 2\sqrt{4C_q^2 - 1}\right)  +C_q\ln\left(\frac{2C_q+1}{2C_q-1}\right) \right\} .
\end{align*}
The remaining task is to evaluate the last two terms.
Directly expanding it gives
\begin{align}
    \ln \left( 2\sqrt{4C_q^2 - 1}\right)  +C_q\ln\left(\frac{2C_q+1}{2C_q-1}\right) &= \ln2 + \frac{1}{2} \ln \left(2C_q+1\right)+\frac{1}{2} \ln \left(2C_q-1\right) \nonumber\\
    &\qquad+C_q\ln \left(2C_q+1\right)- C_q \ln \left(2C_q-1\right).\\
    \lim_{C_q \to 1/2} \left\{\ln \left( 2\sqrt{4C_q^2 - 1}\right)  +C_q\ln\left(\frac{2C_q+1}{2C_q-1}\right) \right\}&=\ln 2+ 2\times \frac{1}{2} \ln(2) = \ln 4 \ . \label{eq: entropy term at zero temperature limit}
\end{align}
Therefore, Eq.~\eqref{eq: MI for general protocol at finite T} simplifies to
\begin{align}
         \lim_{C_q \to 1/2}    \overline{I(A:B, t)} &= 2\sum_{q >0} \left\{ \sqrt{D_q + E_q}\left(F \left[\arcsin\left(\frac{1}{\sqrt{D_q}}\right), \sqrt{\frac{D_q}{D_q+E_q}}\right] - E\left[\arcsin\left(\frac{1}{\sqrt{D_q}}\right), \sqrt{\frac{D_q}{D_q+E_q}}\right] \right) \right. \nonumber\\
    &\qquad \qquad \left. + 1 + \ln\left( \frac{ \sqrt{D_q + E_q - 1} + \sqrt{D_q - 1} }{ 4 } \right) \right\} \ .
\end{align}
Thus, we prove the second MI expression stated in the Table.~\ref{Table_general_quench}.

\subsection{Averaged MI formula for a sudden quench and finite and zero temperature \label{sec: MI zero temperature limit and sudden quench}}

 For a sudden quench at finite temperature, we simply set $D_q=1, E_q=A_q$.
    So, Eq.~\eqref{eq: MI for general protocol at finite T} reduces to
    \begin{align}
     \overline{I(A:B, \beta)} 
    &=4\sum_{q >0}  \left\{ C_q \sqrt{1+A_q}\left(F \left[\arcsin\left(\frac{1}{2C_q}\right), \frac{1}{\sqrt{1+A_q}}\right] -E\left[\arcsin\left(\frac{1}{2C_q}\right), \frac{1}{\sqrt{1+A_q}}\right] \right) \right. \nonumber\\
    &\qquad  \qquad \left. + \frac{1}{2} \left[ 1 + \ln\left( \frac{ \sqrt{4C_q^2(1+A_q) - 1} + \sqrt{4C_q^2 - 1} }{ 2\sqrt{4C_q^2 - 1} } \right) \right]  -\frac{1}{2}C_q\ln\left(\frac{2C_q+1}{2C_q-1}\right)  \right\}  \ .\label{eq: MI for sudden quench at finite T}
\end{align}

In the zero temperature limit, we have $\lim_{\beta \to \infty} C_q = \frac{1}{2}$.
So, Eq.~\eqref{eq: MI for sudden quench at finite T} transforms to
\begin{align}
    \lim_{C_q \to 1/2}    \overline{I(A:B, t)} &\to 2\sum_{q >0} \left\{ \sqrt{1+A_q}\left( F\left[\frac{\pi}{2}, \frac{1}{\sqrt{1+A_q}}\right] - E\left[\frac{\pi}{2}, \frac{1}{\sqrt{1+A_q}}\right] \right)\right. \nonumber\\
    &\qquad \left. + 1 + \ln\left( \sqrt{A_q} \right)  \right\} -\lim_{C_q \to 1/2} \left\{\ln \left( 2\sqrt{4C_q^2 - 1}\right)  +C_q\ln\left(\frac{2C_q+1}{2C_q-1}\right) \right\} .
\end{align}
Similarly, using Eq.~\eqref{eq: entropy term at zero temperature limit}, we have
\begin{align}
    \lim_{C_q \to 1/2}    \overline{I(A:B, \beta)} &\to \sum_{q >0} \left\{ \ln A_q +2\sqrt{1+A_q}\left( F\left[\frac{\pi}{2}, \frac{1}{\sqrt{1+A_q}}\right] - E\left[\frac{\pi}{2}, \frac{1}{\sqrt{1+A_q}}\right] \right) + 2(1- \ln4) \right\}.
\end{align}
Thus, we obtain the MI formulas stated in the Table.~\ref{Results_sudden_quench_table}.

\subsection{The scaling law of averaged MI for sudden quench at zero temperature \label{scaling_arbitrary_protocol}}

%
We now compute the scaling constant of the averaged MI after a sudden quench at zero-temperature limit and thermodynamic limit.
For simplicity, we define $\displaystyle q=\frac{\sqrt{4\omega J_{f}}}{v}
x$, so that $\displaystyle A_{q}=\frac{1}{4x^{2}(1+x^{2})}$.
Thus, the above MI summand becomes
\begin{align}
{\rm lim}_{\beta\to \infty}\overline{I(A:B,\beta)}&=-\ln[4x^{2}(1+x^{2})]+2\sqrt{1+\frac{1}{4x^{2}(1+x^{2})}}\left[F\left(\frac{\pi}{2},k_{q}\right)-E\left(\frac{\pi}{2},k_{q}\right)\right]+2(1-\ln 4).
\end{align}
Now, let us use the hyperbolic substitution $x=\sinh\theta, dx=\cosh\theta d\theta$.
Since $1+x^{2}=\cosh^{2}\theta$, we get $4x^{2}(1+x^{2})=\sinh^{2}(2\theta)$. Hence, we have the simplification
\begin{align}
A_{q}&=\frac{1}{\sinh^{2}(2\theta)}={\rm csch}^{2}(2\theta).
\end{align}
Therefore, $\ln A_{q}=-2\ln[\sinh(2\theta)]
$. Also, we have $\sqrt{1+A_{q}}=\sqrt{1+{\rm csch}^{2}(2\theta)}=\coth(2\theta)$ since $\theta>0$, so ${\rm coth}(2\theta)>0$. Thus the MI summand becomes
\begin{align}
{\rm lim}_{\beta\to \infty}\overline{I(A:B,\beta)}&=2\ln[\sinh(2\theta)]+2\coth(2\theta)\left[F\left(\frac{\pi}{2},k(\theta)\right)-E\left(\frac{\pi}{2},k(\theta)\right)\right]+2(1-\ln4).
\end{align}
Furthermore, $k_{q}(A_{q})=(1+A_{q})^{-1/2}$. Since $A(\theta)={\rm csch}^{2}(2\theta)$, we get $k(\theta)=[1+\csch^{2}(2\theta)]^{-1/2}={\rm tanh}(2\theta)$.
So the MI integral is ${\rm lim}_{\beta\to \infty}\overline{I(A:B,\beta)}$
\begin{align}
&=\frac{L\sqrt{4\omega J_{f}}}{2\pi v}\int_{0}^{\infty}d\theta\; \cosh\theta\Big\{\ln [\csch^{2}(2\theta)]+2\coth(2\theta)\left[F\left(\frac{\pi}{2},\tanh(2\theta)\right)-E\left(\frac{\pi}{2},\tanh(2\theta)\right)\right]+2(1-\ln 4)\Big\},\\
&=\frac{L\sqrt{4\omega J_{f}}}{2\pi v}(1.201).
\end{align}
Similar calculations can also be done for the LN and R\'enyi-$2$.
Consequently, in the TDL, for a sudden quench, MI, LN, and R{\'e}nyi entropies all collapse onto the single scaling variable $L\sqrt{4\omega J_f}/v$, as mentioned in Sec.~\ref{Sec: TDL scaling}. 

We next provide details why the scaling law for the MI, LN and R{\'e}nyi entropy breaks down for a general protocol. For each mode $q$, the Ermakov equation is given by
\begin{align}
    \ddot{\gamma}_q(t)
+
\left[(vq)^2+4\omega J_f\,\mathcal{P}\!\left(\frac{t}{t_f}\right)\right]\gamma_q(t)
=
\frac{(vq)^2}{\gamma_q^3(t)} .
\end{align}
where $\mathcal{P}$ is the dimensionless protocol.
Define the mass gap $g=\sqrt{4\omega J_f}$, dimensionless frequency/ (or momentum) $x=vq/g$, and the dimensionless time $\tau=gt$.
Substituting these into the Ermakov equation gives
\begin{align}
    g^2\ddot{\gamma}_q(\tau)
+
g^2
\left[
x^2+\mathcal{P}\!\left(\frac{\tau}{g t_f}\right)
\right]\gamma_q(\tau)
=
\frac{g^2x^2}{\gamma_q^3(\tau)}
\end{align}
where we use the relation $\displaystyle \frac{d}{dt}= g\frac{d}{d\tau}$ and $(\cdots)'$ denotes differentiation with respect to $\tau$.
Clearly, at time $t=t_f$, $\gamma_q(t_f)$ is a function of $x$ and $gt_f$.
Therefore, the $D_q$, $E_q$ and the symplectic eigenvalues are all functions of $x$ and $gt_f$.
The integral of those information-theoretic quantities become
\begin{align}
    \frac{L\sqrt{4\omega J_f}}{\pi v} \int f_1[x, \sqrt{4\omega J_f}t_f, C_q(x,J_f)] \;dx ,
\end{align}
where $f_1$ is some known functions.
Using the same substitution $x=vq/\sqrt{4\omega J_{f}}$, we can still factor out the prefactor $   \frac{L\sqrt{4\omega J_f}}{\pi v} $, but the integrand in general is not only a function of $x$. 
For a sudden quench, $\sqrt{4\omega J_f}\;t_f \to 0^+$ so the integrand only depends on $x$ and hence we can claim the linear scaling law.

\subsection{Decay of averaged MI with temperature for arbitrary protocols
\label{subsec: MI decays monotonically}}
\begin{lemma}\label{lemma: xS'(x) decreasing}
Let $x > \dfrac{1}{2}$ and define $f(x)=x\ln\left(\frac{x+1/2}{x-1/2}\right)$. Then $f(x)$ is a strictly decreasing function on $\displaystyle [\frac{1}{2},\infty)$.
\end{lemma}
\begin{proof}
    We calculate the derivative of $f(x)$ as $f'(x) = \ln\left(\frac{2x+1}{2x-1}\right) - \frac{4x}{4x^2-1}$.
     Since $x > 1/2$, let $x = \coth(\alpha)/2$ for some $\alpha > 0$.
     Thus, the derivative can be rewritten as
     \begin{align}
         f'(x) = \ln\left(\frac{\coth\alpha + 1}{\coth\alpha - 1}\right)-\frac{2\coth\alpha}{\coth^2\alpha - 1} 
         &=\ln\left(\frac{\cosh\alpha + \sinh\alpha}{\cosh\alpha - \sinh\alpha}\right)-\frac{2 \cosh\alpha }{\sinh\alpha} \times \sinh^2\alpha\\
         &=\ln(e^{2\alpha}) - 2\cosh\alpha\sinh\alpha =2\alpha - \sinh (2\alpha)
     \end{align}
    Using $\sinh(y) > y$ for all positive $y$, we achieve $f'(x) < 0$.
    Thus, $f(x)$ is a strictly monotonically decreasing function.
\end{proof}

\begin{theorem}[Monotonic decay with temperature]
Let $A\geq1$ and $B >0$. Let $\displaystyle C(T)>\frac12$ be a differentiable function of the temperature $T$ satisfiying  $\displaystyle \frac{\partial C(T)}{\partial T} >0$. Define
\begin{align}
   \overline{I}(T)&= \int_0^{\pi/2} \left[\mathcal{S}\left[C(T)\sqrt{A+B\sin^2u}\right]   - \mathcal{S}\left[C(T) \right]  \right] \; du \ ,
\end{align}    
where $\mathcal{S}[x] \equiv \left(x+\frac{1}{2}\right)\ln\left(x+\frac{1}{2}\right) - \left(x-\frac{1}{2}\right)\ln\left(x-\frac{1}{2}\right)$.
Then $\overline{I}(T)$ is strictly decreasing w.r.t. temperature, i.e. $\displaystyle \frac{d\overline{I}}{dT} < 0$.
\end{theorem}
  
\begin{proof}
    Let 
    \begin{align}
   \overline{I}(T)&= \int_0^{\pi/2} \left[\mathcal{S}\left[C(T)\sqrt{A+B\sin^2u}\right]   - \mathcal{S}\left[C(T) \right]  \right] \; du \ ,
\end{align}   
Using Leibniz's integral rule, we take the partial derivative of $\overline{I}$ with respect to $C$, i.e.
\begin{align}
    \frac{\partial \overline{I}}{\partial C} &= \int_0^{\pi/2} \left[ \sqrt{A+B\sin^2u} \, \mathcal{S}'\left[C(T)\sqrt{A+B\sin^2u}\right] - \mathcal{S}'[C(T)] \right] du \\
    &=\frac{1}{C(T)}\int_0^{\pi/2} \left[ C(T)\sqrt{A+B\sin^2u} \, \mathcal{S}'\left[C(T)\sqrt{A+B\sin^2u}\right] - C(T)\mathcal{S}'[C(T)] \right] du 
\end{align}
For simplicity, we define $f(x) = x \mathcal{S}'[x]$. So, the above derivative becomes
\begin{align}
    \frac{\partial \overline{I}}{\partial C} &= \frac{1}{C(T)}\int_0^{\pi/2} f\left[ C(T)\sqrt{A+B\sin^2u}\right] - f\left[C(T)\right]  \; du .
\end{align}
Using the definition of $\mathcal{S}[x] $, we can directly compute $f(x) =x \mathcal{S}'[x] =x  \ln\left(\frac{x+1/2}{x-1/2}\right)$,
which is strictly decreasing by Lemma.~\ref{lemma: xS'(x) decreasing}.
Thus, under the assumption that $A\geq1$ and $B >0$, we see that
\begin{align}
    C(T)\sqrt{A+B\sin^2u} \geq C(T)  \longrightarrow f\left[ C(T)\sqrt{A+B\sin^2u}\right] \leq f\left[C(T)\right]
\end{align}
Since the integrand is always negative, we 
have
\begin{align}
    \frac{\partial \overline{I}}{\partial C} &= \frac{1}{C(T)}\int_0^{\pi/2} f\left[ C(T)\sqrt{A+B\sin^2u}\right] - f\left[C(T)\right]  \; du <0 \ .
\end{align}
Finally, by chain rule, we obtain $\displaystyle \frac{dI}{dT} = \frac{\partial \overline{I}}{\partial C} \frac{\partial C}{\partial T} <0$,
where we also use the assumption that $\displaystyle \frac{\partial C(T)}{\partial T} >0$.
\end{proof}
\section{
Long-time averaged logarithmic negativity (LN)\label{Sec: long time LN}}

Here we study the long-time average of LN at finite temperature for a finite-ramp protocol and its sudden quench limit. 
First, we prove a few auxiliary dilogarithm integral identities and the core integral formula of Theorem~\ref{Th: integral formula for LN} in Appendix.~\ref{Sec: proof for LN's integral identiteis}.
Then, we provide the integral representation of the averaged LN for a general finite-ramp protocol quoted in Table~\ref{Table_general_quench} (Appendix~\ref{sec: exact_expression_LN_finite_temp}).
Combining the results from these two subsections, we obtain the closed-form expressions of the LN after a sudden quench at finite (Appendix~\ref{sec: LN for sudden quench at T >0}) and zero temperature (Appendix~\ref{sec: LN for sudden quench at T=0}) mentioned in Table~\ref{Results_sudden_quench_table}.
Next, in Appendix~\ref{sec: threshold temperature}, we derive the threshold temperature of Eq.~\eqref{eq: threshold temperature for sudden quench} used in Sec.~\ref{subsec:threshold_temperature}.
Finally, we explain why the linear scaling with $L\sqrt{4\omega J_f}/v$ found for the sudden quench in Sec.~\ref{Sec: TDL scaling} breaks down for a generic finite-ramp protocol (Appendix~\ref{scaling_arbitrary_protocol}).

\subsection{Proof of the core integral formula for computing averaged LN after a sudden quench \label{Sec: proof for LN's integral identiteis}}
Before we prove the key integral formula, let's first prove two useful integral identities, which (including many other quantities in this section) are directly dependent on the dilogarithm function
\begin{align}
    \operatorname{Li}_2(z)= \sum_{n=1}^{\infty} \frac{z^n}{n^2}\;,\;\; \text{ for } \;\; |z| \leq 1 \ .
\end{align} 
\begin{lemma}
\label{lemma: ln cos}
Let $\;0 \leq A \leq \pi/2$. Then
     \begin{align}
               \int_0^{A}  \ln \left(\cos \phi \right)\; d\phi&=-A\ln2 -\frac{1}{2}\operatorname{Im} \left[\Li_2(-e^{2iA})\right] \ .
        \end{align}
\end{lemma}

\begin{proof}
     Since we have the following Fourier series expansion for $-\frac{\pi}{2 } < \phi <\frac{\pi}{2 } $: $\ln(\cos \phi) = -\ln 2 - \sum_{n=1}^{\infty} \frac{(-1)^n}{n} \cos(2n\phi) $,
        \begin{align}
                \int_0^{A}  \ln \left(\cos \phi \right)\; d\phi
            &= -A\ln 2  + \sum_{n=1}^{\infty} \frac{(-1)^{n+1}}{2n^2} \sin(2nA) \ .
        \end{align}
        Let $z = -e^{2iA}$, then
        \begin{align}
          z^n &=  (-e^{i2A})^n
= (-1)^n \left[\cos(2nA) + i\sin(2nA)\right], \quad\implies 
\text{Im} \left(z^n\right)= (-1)^n\sin(2nA) .
        \end{align}
        Using this observation,
        we have
        \begin{align}
          \int_0^{A}  \ln \left(\cos \phi \right)\; d\phi
            &= -A\ln 2  - \frac{1}{2}\sum_{n=1}^{\infty} \frac{ \text{Im}\left(z^n\right)}{n^2} = -A\ln 2 -\frac{1}{2}\text{Im} \left[\operatorname{Li}_2(z) \right] \ .
        \end{align}

\end{proof}

\begin{lemma}
\label{lemma: ln(B^2+sin^2)}
    Let $0 \leq A \leq \pi/2$, $B >0$ and set $C= 2B^2 + 1 - 2B\sqrt{B^2 + 1}$. Then
    \begin{align}
              \int_0^{A}  \ln \left(B^2 + \sin^2\phi\right) \;d\phi&=2A \ln \left(\frac{B + \sqrt{B^2 + 1}}{2}\right) - \operatorname{Im} \left[\Li_2 \left(Ce^{2iA} \right)\right] \ .
    \end{align}
\end{lemma}

\begin{proof}
        Let $B=\sinh \Theta$.
        Observe that
        \begin{align}
            B^2 + \sin^2\phi =\frac{2\sinh^2 \Theta+1- \cos 2 \phi}{2}=\frac{e^{2\Theta}}{4} \left[1- 2 e^{-2\Theta} \cos 2 \phi + e^{-4\Theta}  \right] \ .
        \end{align}
        So, we obtain
        \begin{align}
            \ln \left(B^2 + \sin^2\phi\right)&=\ln \left(\frac{e^{2\Theta}}{4} \right) + \ln\left(1- 2 e^{-2\Theta} \cos 2 \phi + e^{-4\Theta}\right)=2\Theta - 2\ln2 -2\sum_{n=1}^{\infty} \frac{e^{-2n\Theta} }{n} \cos (2n\phi),
        \end{align}
        where we use Fourier series expansion $   \ln\big(1 - 2 a \cos x + a^2\big) = -2 \sum_{n=1}^{\infty} \frac{a^n}{n} \cos(n x)$, for $|a| < 1$.
        We integrate both sides
        \begin{align}
            \int_0^{A}  \ln \left(B^2 + \sin^2\phi\right) \;d\phi&= A[2\Theta - 2\ln2 ] -2\sum_{n=1}^{\infty} \frac{e^{-2n\Theta} }{n} \frac{\sin(2nA)}{2n}.
        \end{align}
        Let $z=e^{-2\Theta}e^{2iA}$. Then 
        \begin{align}
            \int_0^{A}  \ln \left(B^2 + \sin^2\phi\right) \;d\phi&= A[2\Theta - 2\ln2 ] - \sum_{n=1}^{\infty} \frac{1}{n^2} \text{Im} \left(z^n\right) =A[2\Theta - 2\ln2 ] - \text{Im} \left[\Li_2 \left(e^{-2\Theta}e^{2iA} \right)\right],
        \end{align}
       where we use the definition of the dilogarithm function in the last line.
       Finally, since $ B=\sinh \Theta $, we have
        \begin{align}
          \Theta = \ln\left(B + \sqrt{B^2 + 1}\right),\qquad
e^{-2\Theta} = 2B^2 + 1 - 2B\sqrt{B^2 + 1}.
        \end{align}
        Consequently, defining $C= 2B^2 + 1 - 2B\sqrt{B^2 + 1}$, we have
        \begin{align}
                  \int_0^{A}  \ln \left(B^2 + \sin^2\phi\right) \;d\phi&=2A \ln \left(\frac{B + \sqrt{B^2 + 1}}{2}\right) - \text{Im} \left[\Li_2 \left(Ce^{2iA} \right)\right].
        \end{align}
\end{proof}
Now, let's prove the key integral formula for LN after a sudden quench.
\begin{theorem}
\label{Th: integral formula for LN}
Let $B\geq 0$ and $A \geq B^2$. Define
\begin{align}
    x=\arcsin{ \left[\frac{B}{\sqrt{A}}\right]} \qquad
    C=\frac{B+\sqrt{B^2+1}}{2}\qquad
    D = \arcsin\left( \sqrt{\frac{A - B^2}{1 + A}} \right) \ .
\end{align}
Under these assumptions, we have the following identity
    \begin{align}
       \int_{x}^{\pi/2}  {\rm arcsinh}\left(\sin \phi\sqrt{A} \right)\;   d\phi =\frac{1}{2} \left\{D\ln \left(\frac{A}{4C^2}\right) + \operatorname{Im} \left[\Li_2 \left(\frac{e^{2iD} }{4C^2}\right)  -\Li_2\left(-e^{2iD}\right)\right]\right\} +\left(\frac{\pi}{2}-x \right) \ln 2C\ ,
\end{align}
\end{theorem}

\begin{proof}
For simplicity, define $   F(A)=\int_{x}^{\pi/2}  {\rm arcsinh} \left(\sin \phi\sqrt{A} \right)\;   d\phi
$.
By Leibniz integral rule,
\begin{align}
    \frac{dF(A)}{dA} &=\left\{\int_{x}^{\pi/2}  \frac{\partial}{\partial  A}\left[{\rm arcsinh} \left(\sin \phi\sqrt{A} \right)\right]\; d\phi  \right\}- \left[{\rm arcsinh} \left(\sin x\sqrt{A} \right)\right]\frac{dx}{dA}. \label{eq: Leibniz integral for LN formula}
\end{align}
Now, we solve these two terms separately.
\begin{enumerate}
    \item \uline{The second term}:    
    Direct calculation yields
    \begin{align}
        \frac{dx}{dA} &=\frac{1}{\sqrt{1-B^2/A}} \left(-\frac{1}{2} \frac{B}{A^{3/2}}\right)=\frac{-B}{2A\sqrt{A-B^2}}.
    \end{align}
    Next, using the definition of $x$, we have $     {\rm arcsinh} \left(\sin x\sqrt{A} \right)={\rm arcsinh} (B)$.
    So,
    \begin{align}
       - \left[{\rm arcsinh}\left(\sin x\sqrt{A} \right)\right]\frac{dx}{dA} &=\frac{B}{2A\sqrt{A-B^2}}  \;{\rm arcsinh}(B).
    \end{align}

     \item \uline{The first term}:     
     The integral is
     \begin{align}
         \int_{x}^{\pi/2}  \frac{\partial}{\partial  A}\left[{\rm arcsinh}\left(\sin \phi\sqrt{A} \right)\right]\; d\phi &=\int_{x}^{\pi/2}  \frac{1}{2\sqrt{A}}
\frac{\sin \phi}{\sqrt{1+A\sin^2 \phi}} \; d\phi.
     \end{align} 
     Let $t = \cos \phi$ and $dt=-\sin \phi  \; d\phi$. Then it becomes
     \begin{align}
        \int_{x}^{\pi/2}  \left[ \frac{1}{2\sqrt{A}}
\frac{\sin \phi}{\sqrt{1+A\sin^2 \phi}} \right]\; d\phi&= \int_{0}^{\cos x} \frac{1}{2\sqrt{A} \sqrt{1 + A(1 - t^2)}} dt \ .
     \end{align}
     Let $\displaystyle k = \sqrt{\frac{A}{1+A}}$. Then
     \begin{align}
          \int_{0}^{\cos x} \frac{1}{2\sqrt{A} \sqrt{1 + A - A t^2}} dt &=\frac{1}{2\sqrt{A}\sqrt{1+A}}\int_{0}^{\cos x} \frac{1}{\sqrt{1  -k^2 t^2}} dt\\
          &=\frac{1}{2 A} \arcsin\left( \sqrt{\frac{A}{1+A}} \cos x \right) = \frac{1}{2 A} \arcsin\left( \sqrt{\frac{A - B^2}{1 + A}} \right),
     \end{align}
     where we use the relation $ \cos(x) = \sqrt{\frac{A - B^2}{A}}$ in the last line.
     \end{enumerate}
Combining two terms, Eq.~\eqref{eq: Leibniz integral for LN formula} is 
\begin{align}
    \frac{dF}{dA} = \frac{1}{2 A} \arcsin\left( \sqrt{\frac{A - B^2}{1 + A}} \right) + \frac{B}{2 A \sqrt{A - B^2}}  {\rm arcsinh} (B). \label{eq: dF/dA for LN}
\end{align}
Next, we  integrate it from $B^2$ to $A$:
\begin{enumerate}[leftmargin=*]
    \item  \uline{The second term of Eq.~\eqref{eq: dF/dA for LN}}: 
    Consider the integral
    \begin{align}
        \int_{B^2}^{A} \frac{B}{2 \alpha \sqrt{\alpha - B^2}}  {\rm arcsinh}(B) \; d\alpha &=  {\rm arcsinh}(B) \int_0^{\theta_q} \frac{B }{2 (B^2 \sec^2 \theta) (B \tan \theta)} (2 B^2 \sec^2 \theta \tan \theta)d\theta \\
        &= {\rm arcsinh} (B) \; \theta_q \ ,
    \end{align}
    where we use the substitution $\alpha=B^2 \sec^2 \theta$ and  $d\alpha = 2 B^2 \sec^2 \theta \tan \theta \, d\theta$.
    Note that we also define $\cos\theta_q = B/\sqrt{A}$ and use the assumption $A^2 \geq B\geq0$.
    Using trigonometric identity, it becomes
    \begin{align}
         \int_{B^2}^{A} \frac{B}{2 \alpha \sqrt{\alpha - B^2}}  {\rm arcsinh} (B) d\alpha   &= {\rm arcsinh} (B) \left[\frac{\pi}{2}-\arcsin\left(\frac{B}{\sqrt{A}}\right) \right]
         = {\rm arcsinh} (B) \left(\frac{\pi}{2}-x \right).
    \end{align}
    where we recall the definition of $x$ in the last line.

    \item  \uline{The first term of Eq.~\eqref{eq: dF/dA for LN}}:
    Define $\theta = \arcsin\left( \sqrt{\frac{\alpha - B^2}{1 + \alpha}} \right) $, and $dv=\frac{1}{2\alpha} d\alpha, $$v=\frac{1}{2} \ln \alpha$.
    %
    Using this definition of $\theta,v$ and integration by parts, we get
    \begin{align}
           \int_{B^2}^{A} \frac{1}{2 \alpha} \arcsin\left( \sqrt{\frac{\alpha - B^2}{1 + \alpha}} \right) d\alpha &= \left[\frac{1}{2} \ln \alpha \arcsin\left( \sqrt{\frac{\alpha - B^2}{1 + \alpha}} \right) \right]_{B^2}^{A} - \frac{1}{2}\int_{0}^{D} \ln \alpha(\theta)   \;d\theta \\
             &=\frac{1}{2} \ln A \arcsin\left( \sqrt{\frac{A - B^2}{1 + A}} \right)  -\frac{1}{2} \int_0^{D}  \ln \left(\frac{B^2 + \sin^2\theta}{\cos^2 \theta }\right) \; d\theta \\
             &=\frac{1}{2} \left[D\ln A  - \int_0^{D}  \ln \left(B^2 + \sin^2\theta\right) \;d\theta + 2 \int_0^{D}  \ln \left(\cos \theta \right)\; d\theta\right],
    \end{align}
    where we recall $ D = \arcsin\left( \sqrt{\frac{A - B^2}{1 + A}} \right)$ in the first line, and use the relation $ \displaystyle \alpha(\theta) = \frac{B^2 + \sin^2\theta}{\cos^2 \theta} $ in the second line.
    \end{enumerate}
    Therefore, we have
    \begin{align}
        F(A)&=\int_{x}^{\pi/2}  {\rm arcsinh} \left(\sin \phi\sqrt{A} \right)\;   d\phi \\
        &=\int_{B^2}^{A} \left[\frac{1}{2 \alpha} \arcsin\left( \sqrt{\frac{\alpha - B^2}{1 + \alpha}} \right) + \frac{B}{2 \alpha \sqrt{\alpha - B^2}}  {\rm arcsinh} (B)\right] \; d\alpha\\
        &=\frac{1}{2} \left[D\ln A  - \int_0^{D}  \ln \left(B^2 + \sin^2\theta\right) \;d\theta + 2 \int_0^{D}  \ln \left(\cos \theta \right)\; d\theta\right] +{\rm arcsinh} (B) \left(\frac{\pi}{2}-x \right).
    \end{align}
    Using Lemma.~\ref{lemma: ln cos} and Lemma.~\ref{lemma: ln(B^2+sin^2)}, the integral terms can be evaluated analytically: 
    \begin{align}
        \int_0^{D}  \ln \left(\cos \theta \right)\; d\theta&=-D\ln2 -\frac{1}{2}\operatorname{Im} \left[\Li_2(-e^{2iD})\right],\\
        \int_0^{D}  \ln \left(B^2 + \sin^2\theta\right) \;d\theta&=2D \ln \left(\frac{B + \sqrt{B^2 + 1}}{2}\right) - \operatorname{Im} \left[\Li_2 \left(Ee^{2iD} \right)\right],
    \end{align}
    where we define $E=2B^2 + 1 - 2B\sqrt{B^2 + 1}$.
    Therefore,
    \begin{align}
        &\quad\int_{x}^{\pi/2}  {\rm arcsinh} \left(\sin \phi\sqrt{A} \right)\;   d\phi \nonumber\\
        &=\frac{1}{2} \left\{D\ln A  - 2D \ln \left(\frac{B + \sqrt{B^2 + 1}}{2}\right)+ \operatorname{Im} \left[\Li_2 \left(Ee^{2iD} \right)\right] -2D\ln2 -\operatorname{Im} \left[\Li_2(-e^{2iD})\right]\right\} +{\rm arcsinh} (B) \left(\frac{\pi}{2}-x \right)\nonumber\\
        &=\frac{1}{2} \left\{D\ln \left(\frac{A}{4C^2}\right) + \operatorname{Im} \left[\Li_2 \left(\frac{e^{2iD} }{4C^2}\right) -\Li_2\left(-e^{2iD}\right)\right]\right\} +\left(\frac{\pi}{2}-x \right) \ln 2C,
    \end{align}
    where we use the following relations
    \begin{align}
        2C&=B + \sqrt{B^2 + 1}, \qquad E=2B^2 + 1 - 2B\sqrt{B^2 + 1}=
        \frac{1}{4C^2}, \qquad {\rm arcsinh} (B)=\ln 2C \ ,
    \end{align}
    which can be derived from the definition $\displaystyle C=\frac{B+\sqrt{B^2+1}}{2}$.
\end{proof}
\subsection{Integral representation of long-time averaged LN for a finite-ramp at finite temperature \label{sec: exact_expression_LN_finite_temp}}
First, the LN at time $t$ is given by $E_{\mathcal{N}}(t) = \sum_{q>0} \max\{0, -\ln[2\nu_q(t,\beta)]\}$, where $\nu_q(t,\beta) = C_q\bigg(\sqrt{1+\mathcal{F}[\gamma_q(t)]}-\sqrt{\mathcal{F}[\gamma_q(t)]}\bigg)$.
Similar to the construction of Eq.~\eqref{eq: effective time averaged MI}, the time averaged LN is
\begin{align}
    \overline{E_{\mathcal{N}}} = \frac{2}{\pi} \sum_{q >0}  \int_0^{\pi/2} \max\{0, -\ln[2\nu_q(u)]\} \, du.
\end{align}    
For a finite-ramp protocol, we have $\sqrt{1+\mathcal{F}}=\sqrt{D_q+E_q\sin^2 u}$.
So,
\begin{align} 
     -\ln[2\nu_q(u)] &= -\ln(2C_{q}) - \ln \left(\sqrt{D_q+E_q\sin^2 u}-\sqrt{D_q-1+E_q\sin^2 u}\right) \\
     &=-\ln(2 C_{q}) + {\rm arcsinh} \left(\sqrt{D_q-1+E_q\sin^2 u} \right),
\end{align}
where we use the identity $\ln\left(\sqrt{1+x} - \sqrt{x}\right) = -{\rm arcsinh}(\sqrt{x})$. Thus, the averaged LN becomes
\begin{align}
    \overline{E_{\mathcal{N}}} = \frac{2}{\pi} \sum_{q >0}  \int_0^{\pi/2} \max \left\{0, {\rm arcsinh}\left(\sqrt{D_q-1+E_q\sin^2 u} \right) -\ln(2C_{q}) \right\} \, du \ .
\end{align}  
For the most general setting, i.e., a finite-ramp protocol at finite temperature, we did not succeed in evaluating this integral.
However, in the case of a sudden quench, an exact expression of this integral can be obtained (See next section).
\subsection{Exact expression of averaged LN for sudden quench at finite temperature \label{sec: LN for sudden quench at T >0}}
In the sudden quench limit, we set $D_q=1$ and $E_q=A_q$.
The time-averaged LN then reduces to
\begin{align}
    \overline{E_{\mathcal{N}}} = \frac{2}{\pi} \sum_{q >0}  \int_0^{\pi/2} \max \left\{0, {\rm arcsinh} \left(\sqrt{A_q}\sin u \right) -\ln(2C_{q}) \right\} \, du \ . \label{eq: integral form of LN for sudden quench}
\end{align}  
To contribute nonzero entanglement, we require
\begin{align}
    \sqrt{A_q} \sin u &\geq  \sinh \left[\ln\coth\left(\frac{\beta\hbar\Omega_{q0}}{2}\right)\right] = \csch\left(\beta\hbar\Omega_{q0}\right) \ ,
\end{align}
where we use the identity $\sinh(\ln x) = \dfrac{x^2 - 1}{2x}$ and $\displaystyle \frac{\coth^2 x - 1}{2\coth x} = \csch(2x)$.
For simplicity, we define $V_q=  \csch\left(\beta\hbar\Omega_{q0}\right)$. Next, we split the momentum modes $q$ into two categories.
\begin{enumerate}
    \item \underline{$A_q < V_q^2$}: In this case, we have $ \sqrt{A_q} \sin u \leq \sqrt{A_q} <V_q$ for all $u\in [0,\pi/2]$.
    Consequently, this momentum mode gives zero contribution to the entanglement, i.e.
    \begin{align}
         \int_0^{\pi/2} \max \left\{0, {\rm arcsinh}\left(\sqrt{A_q}\sin u \right) -\ln(2C_{q}) \right\} \, du =0.
    \end{align}
    \item \underline{$A_q \geq V_q^2$}: We define the threshold time for each mode as $u_q^*= \arcsin{ \left(\frac{V_q}{\sqrt{A_{q}}} \right)}$. To exhibit entanglement, the time parameter $u$ for momentum mode $q$ must exceed this threshold. Therefore, we have
    \begin{align}
 \int_0^{\pi/2} \max \left\{0, {\rm arcsinh}\left(\sqrt{A_q}\sin u \right) -\ln(2 C_{q}) \right\} \, du 
   &=\int_{u_q^*}^{\pi/2}  \left[{\rm arcsinh} \left(\sqrt{A_q}\sin u \right) -\ln(2 C_{q})  \right]\; du \\
   &=\int_{u_q^*}^{\pi/2}  {\rm arcsinh} \left(\sin u\sqrt{A_{q}} \right)\; du \;\;- \left(\frac{\pi}{2}-  u_q^* \right)\ln(2 C_{q}) .\label{eq: two integral terms for LN for sudden quench}
    \end{align}
    The remaining task is to solve for $F(A_q)=\int_{u_q^*}^{\pi/2}  {\rm arcsinh}\left(\sin u\sqrt{A_q} \right)\; du$. Using Theorem.~\ref{Th: integral formula for LN}, it is
\begin{align}
    \int_{u_q^*}^{\pi/2} {\rm arcsinh} \left(\sin u\sqrt{A_q} \right) \; du &= \frac{1}{2} \left\{\phi_q\ln \left(\frac{A_q}{4 C_q^2}\right) +\operatorname{Im} \left[\operatorname{Li}_2 \left(\frac{1}{4C_q^2}e^{2i\phi_q} \right) -\operatorname{Li}_2(-e^{2i\phi_q})\right] \right\}  +\left(\frac{\pi}{2}-u_q^* \right) \ln (2C_q), \label{eq: first integral term for LN for sudden quench}
\end{align}
where we define
\begin{align}
       \phi_q &= \arcsin\left( \sqrt{\frac{A_q - V_q^2}{1 + A_q}} \right), \qquad C_q=\frac{V_q+\sqrt{V_q^2+1}}{2} =\frac{1}{2}\coth\left(\frac{\beta\hbar\Omega_{q0}}{2} \right).
\end{align}
Combining Eq.~\eqref{eq: two integral terms for LN for sudden quench} and Eq.~\eqref{eq: first integral term for LN for sudden quench}, we obtain
  \begin{align*}
   &\quad\int_0^{\pi/2} \max \left\{0, {\rm arcsinh} \left(\sqrt{A_q}\sin u \right) -\ln (2C_q) \right\} \, du  =\frac{1}{2} \left\{\phi_q\ln \left(\frac{A_q}{4 C_q^2}\right) +\operatorname{Im} \left(\operatorname{Li}_2 \left[\frac{1}{4C_q^2}e^{2i\phi_q} \right] -\operatorname{Li}_2(-e^{2i\phi_q})\right) \right\}.
    \end{align*}
\end{enumerate}

The final expression of time averaged LN is then given by
\begin{empheq}[box=\fbox]{align*}
\overline{E_{\mathcal{N}}}
&= \frac{1}{\pi} \sum_{q >0: \;A_q \geq V_q^2} \left\{\phi_q\ln \left(\frac{A_q}{4 C_q^2}\right) +\text{Im} \left[\Li_2 \left(\frac{e^{2i\phi_q}}{4C_q^2} \right) -\Li_2\left(-e^{2i\phi_q}\right)\right] \right\}\\
      C_q&=\frac{1}{2}\coth\left(\frac{\beta\hbar\Omega_{q0}}{2} \right) \geq \frac{1}{2}
        ,\qquad \phi_q=\arcsin\left( \sqrt{\frac{A_q - V_q^2}{1 + A_q}} \right),\qquad 
        V_q=C_q -\frac{1}{4C_q}= \csch\left(\beta\hbar\Omega_{q0}\right) .
\end{empheq}
Note that LN only considers those momentum modes $q$ that satisfy the condition $A_q \geq V_q^2$.
Thus, we recover the second LN expression stated in the Table.~\ref{Results_sudden_quench_table}.

\subsection{Exact expression of averaged LN for sudden quench in the zero-temperature limit \label{sec: LN for sudden quench at T=0}}

Let us consider the sudden quench limit. In the zero temperature limit $\beta \to \infty$, the time-averaged LN reduces to
\begin{align}
    \overline{E_{\mathcal{N}}}(\beta \to \infty) = \frac{2}{\pi} \sum_{q >0}  \int_0^{\pi/2} \max \left\{0, {\rm arcsinh} \left(\sqrt{A_q}\sin u \right) \right\} \, du \ .
\end{align}  
Since $\sqrt{A_{q}}$ is positive, the argument inside of ${\rm arcsinh}$ is nonnegative for $u\in[0,\pi/2]$, and for a positive argument ${\rm arcsinh}$ is nonnegative, so that we can get rid of the ${\rm max}$ function. Let us set
$F(a)=\int_{0}^{\pi/2}{\rm arcsinh}(a\sin u) du$, and consider its derivative
\begin{align}
\frac{\partial F(a)}{\partial a}&=\int_{0}^{\pi/2}\frac{\partial }{\partial a}{\rm arcsinh}(a\sin u)\, du =\int_{0}^{\pi/2}\frac{\sin u}{\sqrt{1+a^{2}\sin^{2}u}}\, du.
\end{align}
Now, set $x=\cos u$, and $dx=-\sin u\, du$. Therefore 
\begin{align}
\frac{\partial F(a)}{\partial a }&=\int_{0}^{1}\frac{dx}{\sqrt{1+a^{2}(1-x^{2})}}=\frac{1}{\sqrt{1+a^{2}}}\frac{1}{k}{\rm arcsin}(k )=\frac{{\rm arctan}(a)}{a},
\end{align}
with $k=a/\sqrt{1+a^{2}}$, where we used $\int \frac{dx}{\sqrt{1+(kx)^{2}}}=\frac{1}{k}{\rm arcsin}(kx)$ and ${\rm arcsin}(a/\sqrt{1+a^{2}})={\rm arctan}(a)$. By direct integration,
\begin{align}
\overline{E_{\mathcal{N}}}(\beta \to \infty)=\int_0^{\sqrt{A_q}} \frac{\partial F(a)}{\partial a } da&=\frac{2}{\pi}\sum_{q>0}\int_{0}^{\sqrt{A_{q}}}\frac{\;{\rm arctan}(a)}{a}da = \frac{2}{\pi}  \sum_{q>0} \operatorname{Ti}_2(\sqrt{A_q}) \ ,
\end{align}
where the inverse tangent integral is $\operatorname {Ti} _{2}(x)=\int _{0}^{x}{\frac {{\rm arctan}(t)}{t}}\,dt $.
Thus, we arrive at the first LN expression in Table.~\ref{Results_sudden_quench_table}.

\subsection{Exact expression for entanglement threshold temperature for a sudden quench \label{sec: threshold temperature}}
In this section, we would like to derive an expression for the threshold temperature at which the LN is null, signature of the vanishing of the entanglement. From the previous subsection, we have
\begin{align}
    \overline{E_{\mathcal{N}}} = \frac{2}{\pi} \sum_{q >0}  \int_0^{\pi/2} \max \left\{0, {\rm arcsinh}\left(\sqrt{A_q}\sin u \right) -\ln\coth\left(\frac{\beta\hbar\Omega_{q0}}{2}\right) \right\} \, du \ .
\end{align}  
For a fixed mode $q$, entanglement is present for some $u \in [0,\pi/2]$ only if 
\begin{align}
    {\rm arcsinh} \left(\sqrt{A_q}\sin u \right) &\geq \ln\coth\left(\frac{\beta\hbar\Omega_{q0}}{2}\right) \longrightarrow
   \sqrt{A_q}\sin u \geq \sinh \left[\ln\coth\left(\frac{\beta\hbar\Omega_{q0}}{2}\right)\right]=\csch\left(\beta\hbar\Omega_{q0}\right).
\end{align}
If $\sqrt{A_q} <\csch\left(\beta\hbar\Omega_{q0}\right)$, the mode $q$ does not contribute to the entanglement for all $u \in [0,\pi/2]$.
Thus, the threshold temperature $\beta_{*}$ (mode-wise) occurs at
\begin{align}
    \beta_{*} \hbar v q &= \text{arsinh}\left( \frac{1}{\sqrt{A_q}} \right) = \text{arsinh}\left( \frac{2vq\sqrt{(vq)^2+4\omega J_f}}{4\omega J_f} \right)= \text{arsinh} \left( 2y\sqrt{y^2+1} \right) \ ,
\end{align}
where we recall that $  A_q=\dfrac{\left(4\omega J_f\right)^{2}}{4 (vq)^2((vq)^2+4\omega J_f)}$ and define $y=\frac{vq}{\sqrt{4\omega J_f}}$.
Using the identity 
$\text{arsinh}(2u\sqrt{u^2+1}) = 2\text{arsinh}(u)$, 
\begin{align}
    \beta_{*} &= \frac{2}{\hbar v q} \text{arsinh}\left(y \right) \quad \longrightarrow\quad
    \mathcal{T}_{*}=\frac{\hbar v q}{2 k_B \;\text{arsinh} \left( y \right)} = \frac{\hbar v q}{2 k_B} \frac{1}{\text{arsinh} \left( \dfrac{vq}{\sqrt{4\omega J_f}} \right)},
\end{align}
Here the expression has been rewritten in terms of the temperature $\mathcal{T}_{*}$, rather than the inverse temperature $\beta$, for a clearer physical interpretation.
Hence, we recover Eq.~\eqref{eq: threshold temperature for sudden quench} in the main text.

\section{Long-time averaged R{\'e}nyi-$2$ entropy}\label{expression_Renyi2_arbitrary}
Here, we compute an exact expression for the time-average of the R{\'e}nyi-$2$ entropy for a general protocol and for the sudden quench limit. 
The final expressions (cf Eq.~\eqref{Eq: S2 for general protocl} and Eq.~\eqref{exact_Renyi_2}) are listed as the $S_2$ entries in Tables~\ref{Table_general_quench} and~\ref{Results_sudden_quench_table}.

Recall the symplectic eigenvalues for the general protocol at time $t\geq t_{f}$:
\begin{align}
    \lambda_{q}^A(t\geq t_{f})&= \frac{1}{2} \coth\left(\frac{\beta \hbar \Omega_{q0}}{2}\right) \sqrt{D_q +E_q \sin^2 \theta_q} =\frac{1}{2} \coth\left(\frac{\beta \hbar \Omega_{q0}}{2}\right) \sqrt{D_{q}}\times\sqrt{1 +\frac{E_q}{D_{q}} \sin^2 \theta_q},\label{sudden_quench_solution}
\end{align}
with $D_q= \frac{1}{2} + \frac{M_q(\zeta_q^2+1) +N_q (\zeta_q^2-1)}{4\zeta_q^2}$ and $E_q=\frac{N_q}{2} \left(\frac{1}{\zeta_q^2}-1 \right)$. The R{\'e}nyi-$2$ as a function of the symplectic eigenvalue is given by
\begin{align}
S_{2}(\rho_{A})&=\sum_{q>0}\ln\left[\left(\lambda_{q}(t,\beta)+\frac{1}{2}\right)^{2}-\left(\lambda_{q}(t,\beta)-\frac{1}{2}\right)^{2}\right]=\sum_{q>0}\ln\left[2\lambda_{q}(t,\beta)\right].
\end{align}
Furthermore, denote $u=\Omega_{qf} t$. Then the integral is simplified to $\overline{\mathcal{S}_{2} \left[\lambda_q(u)\right]}= \frac{1}{\pi} \int_0^{\pi}  \mathcal{S}_{2}\left[\lambda_q(u)\right] \; du$. This gives averaged EE 
\begin{align}
      \overline{\mathcal{S}_{2} \left[\lambda_q(t,\beta)\right]} \approx \frac{1}{\pi} \int_0^{\pi} du\, \ln \left[2\lambda_q(u)\right]
      &=\frac{1}{\pi} \int_0^{\pi} du \ln \left[2C_{q}\sqrt{D_{q}}\sqrt{1+\frac{E_{q}}{D_{q}}\sin^{2}\theta_{q}} \; \right]\nonumber\\
      &=\ln\left[2C_{q}\right]+\frac{1}{2}\ln\left(D_{q}\right)+\frac{1}{2\pi}\int_{0}^{\pi}du\; \ln\left[1+\frac{E_{q}}{D_{q}}\sin^{2}(u)\right].
\end{align}
We can now use the identity
 $\int_{0}^{\pi} \ln\left(1 + A \sin^2 u \right) \, du
= 2\pi \ln\left( \frac{1 + \sqrt{1 + A}}{2} \right)$. Finally, the total R{\'e}nyi-$2$ is given by 
\begin{align}
S_{2}(\rho_{A})&=\sum_{q>0}\ln\left[2C_{q}\right]+\sum_{q>0}\ln\left(\frac{\sqrt{D_{q}}+\sqrt{D_{q}+E_{q}}}{2}\right). \label{Eq: S2 for general protocl}
\end{align}
\begin{figure}[t]
    \centering
    \includegraphics[width=1\linewidth]{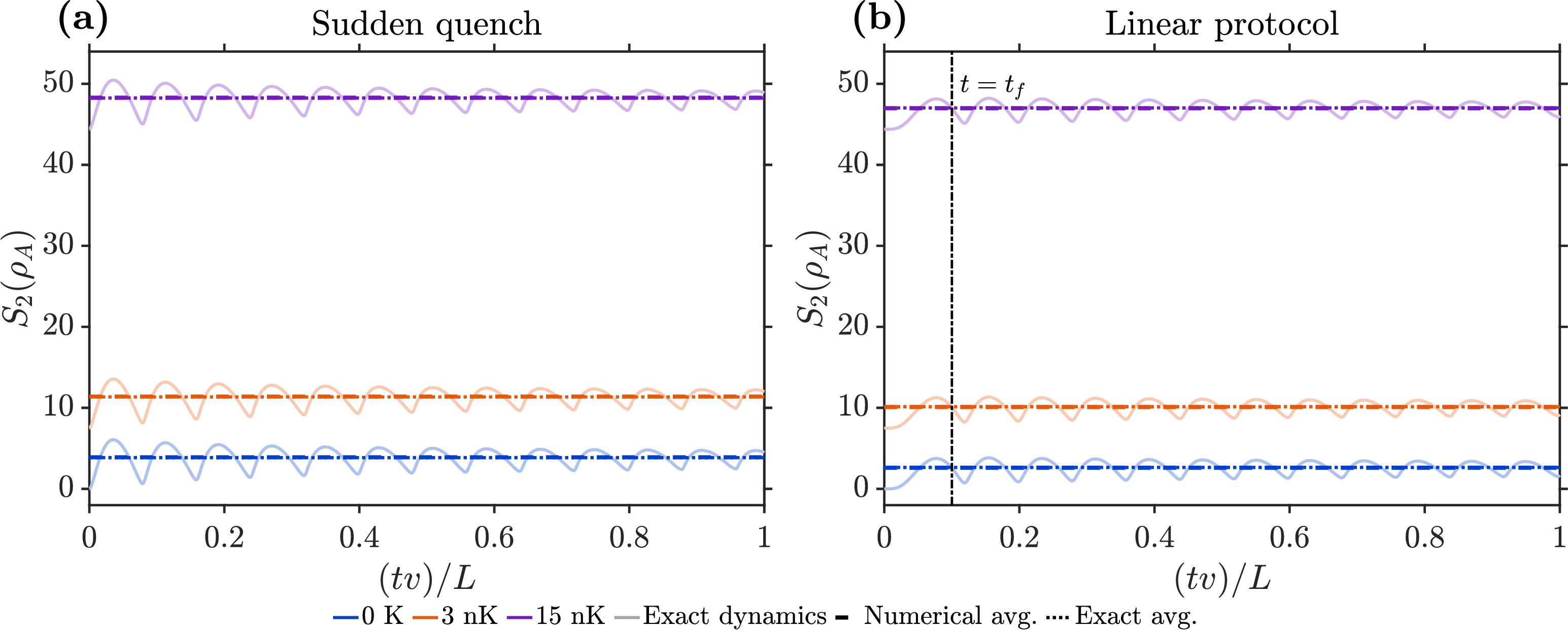}
    \caption{{\bf Time evolution of the R{\'e}nyi-$2$ entropy and comparison to exact time averages for the sudden quench (a) and the linear protocol (b) at different temperatures.--- } The different colors correspond to the different temperatures:  $T=0\; K$ (plain blue), $T=3\; {\rm nK}$ (plain orange), and $T=15\; {\rm nK}$ (plain violet). The numerical time-average of the oscillatory curve is shown in dashed line for each temperature, and this is compared to the exact expressions Eq.~\eqref{exact_Renyi_2} and Eq.~\eqref{sudden_quench_solution} in dotted line. We use the same plotting parameters as in Fig.~\ref{figcomparison_mutual_logneg_sudden_quench}.}
    \label{fig:Renyi_2}
\end{figure}
Recall that the sudden quench results are recovered in taking the limits $D_{q}\to 1, E_{q}\to A_{q}$. Hence, for the sudden quench
\begin{align}
S_{2}(\rho_{A})&=\sum_{q>0}\ln[2 C_{q}]+\sum_{q>0}\ln\left(\frac{1+\sqrt{1+A_{q}}}{2}\right).\label{exact_Renyi_2}
\end{align}
Fig.~\ref{fig:Renyi_2} shows the time evolution of the R{\'e}nyi-$2$ entropy following a sudden quench. Similar to the MI and LN, the R{\'e}nyi-$2$ entropy exhibits oscillatory behavior. As the temperature increases, its time-averaged value is shifted upward. This increase can be attributed to the enhanced thermal mixedness of the state at higher temperatures.%
\section{Details on the adiabatic symplectic eigenvalues \label{adiabatic_approx_details}}
In Sec. \ref{sec: adiabatic smooth_quench_protocol} of the main text, we argue that for a slow driving the solution to the Ermakov equation converges to the adiabatic form of Eq.~\eqref{adiabatic_solution_ermakov}, for which the symplectic eigenvalues simplify to Eqs.~\eqref{eq: thermal state symplectic eigenvalues} and~\eqref{eq: adiabatic_symplectic_partial}.
Here, we elaborate on the derivation of the adiabatic expressions for the symplectic eigenvalues and the conditions required for adiabatic driving.
Moreover, we show that the adiabatically evolved state corresponds to a thermal state with populations following the instantaneous eigenbasis of the time-dependent Hamiltonian.

Let us first detail conditions for the adiabatic solution to hold. Starting from the adiabatic solution to the Ermakov equation $\gamma^{\rm ad}_{q}(t)=\sqrt{\frac{\Omega_{q0}}{\Omega_{q}(t)}}$, obtained in considering $\ddot{\gamma}^{\rm ad}_{q}(t)\approx 0$, we have
\begin{align}
\frac{\ddot{\gamma}^{\rm ad}_{q}(t)}{\Omega^{2}_{q}(t)\gamma^{\rm ad}_{q}(t)}&=\frac{3\dot{\Omega}^{2}_{q}(t)-2\Omega_{q}(t)\ddot{\Omega}_{q}(t)}{4\Omega^{4}_{q}(t)}.
\end{align}
This means that to obtain $\ddot{\gamma}^{\rm ad}_{q}(t)\approx 0$, we need to assume 
\begin{align}
\left|\frac{\ddot{\Omega}_{q}(t)}{2\Omega^{3}_{q}(t)}\right|\ll 1 , \qquad
\sqrt{\frac{3}{4}}\left|\frac{\dot{\Omega}_{q}(t)}{\Omega^{2}_{q}(t)}\right|\ll 1.
\end{align}
Furthermore, for the simplification of the symplectic eigenvalues, we require $\dot{\gamma}^{\rm ad}(t)\approx 0$. The condition gives $\left|\frac{\dot{\gamma}^{\rm ad}_{q}(t)}{\gamma^{\rm ad}_{q}(t)}\right|=\left|\frac{\dot{\Omega}_{q}(t)}{\Omega_{q}^2(t)}\right|\ll 1$, equivalent to the adiabaticity condition for the time-dependent harmonic oscillator.
These conditions are in general never satisfied in the limit $t\to 0$ and thermodynamic limit $q\to 0$, because the denominator would diverge. In practice, however, we find numerically that for finite-size systems the exact solution converges to the adiabatic solution for protocol time $t_f \gg L/v$ where $L/v$ is the slowest period of the system before tunneling is turned on.

Let us now see how the adiabatic solution simplifies the symplectic eigenvalues. Consider the formulas for the symplectic eigenvalues of the reduced the partial transposed covariant matrix
\begin{align}
\lambda_{q}(t,\beta)=& C_{q}(\beta)\sqrt{1+\mathcal{F}[\gamma_q(t)]},\qquad \qquad  \nu_q(t,\beta) = C_{q}(\beta)\bigg(\sqrt{1+\mathcal{F}[\gamma_q(t)]}-\sqrt{\mathcal{F}[\gamma_q(t)]}\bigg).
\end{align}
In the adiabatic limit, $\dot{\gamma}_{q}\approx 0$, and the Ermakov factor simplifies to $\mathcal{F}[\gamma_q(t)]\approx\frac{1}{4} \left(\gamma^{\rm ad}_q(t) - \frac{1}{\gamma^{\rm ad}_q(t)}\right)^2 $. Therefore, we have
\begin{align}
1+\mathcal{F}[\gamma_q(t)]&\approx\frac{1}{4}\left((\gamma^{\rm ad}_{q})^{2}+\frac{1}{(\gamma^{\rm ad}_{q})^{2}}+2\right)=\frac{1}{4}\left((\gamma^{\rm ad}_{q})+\frac{1}{\gamma^{\rm ad}_{q}}\right)^{2}.
\end{align}
Hence, we obtain 
\begin{align}
\lambda_{q}^A(t,\beta)&\approx\frac{C_{q}(\beta)}{2}\left(\gamma^{\rm ad}_{q}(t)+\frac{1}{\gamma^{\rm ad}_{q}(t)}\right).
\end{align}
Similarly
\begin{equation}
\nu_q(t,\beta) \approx \frac{C_q(\beta)}{2}\left[\left(\gamma_q^{\rm ad}+\frac{1}{\gamma_q^{\rm ad}}\right)- \Bigg|\gamma_q^{\rm ad} - \frac{1}{\gamma_q^{\rm ad}}\Bigg|\right] = C_q(\beta)\gamma_q^{\rm ad}(t),
\end{equation}
since we have $\gamma_q^{\rm ad} = \sqrt{\Omega_{q0}/\Omega_q(t)} \leq 1$ and so $(\gamma_q^{\rm ad})^{-1}> \gamma_q^{\rm ad}$.  

These eigenvalues are the same as what we would find for the adiabatically evolved thermal state. Starting from an initial thermal state
\begin{align}
\rho(0)&=\sum_{n^{+}}\frac{e^{-\beta\hbar\Omega^{+}_{q0}}}{{\rm tr}\left(e^{-\beta\hbar\Omega^{+}_{q0}}\right)}|n^{+}(0)\rangle\langle n^{+}(0)|\otimes\sum_{n^{-}}\frac{e^{-\beta\hbar\Omega^{-}_{q0}}}{{\rm tr}\left(e^{-\beta\hbar\Omega^{-}_{q0}}\right)}|n^{-}(0)\rangle\langle n^{-}(0)|.
\end{align}
If the driving is adiabatic, the $n^{th}$ eigenstate evolves as $|n(0)\rangle\to e^{i\phi_{n}(t)}|n(t)\rangle$, so the density matrix becomes
\begin{align}
\rho^{\rm ad}(t)&=\sum_{n^{+}}\frac{e^{-\beta\hbar\Omega^{+}_{q0}}}{{\rm tr}\left(e^{-\beta\hbar\Omega^{+}_{q0}}\right)}|n^{+}(t)\rangle\langle n^{+}(t)|\otimes\sum_{n^{-}}\frac{e^{-\beta\hbar\Omega^{-}_{q0}}}{{\rm tr}\left(e^{-\beta\hbar\Omega^{-}_{q0}}\right)}|n^{-}(t)\rangle\langle n^{-}(t)|.
\end{align}
Note that this is not the instantaneous Gibbs state because the populations are preserved. 
\end{document}